\documentclass[11pt]{article}
\usepackage[left=1in,top=1in,right=1in,bottom=1in,head=.1in,nofoot]{geometry}

\usepackage{setspace,url,bm,amsmath} %

\usepackage{titlesec} %
\titlelabel{\thetitle.\quad} %

\titleformat*{\section}{\bf\Large\center}

\usepackage{graphicx} %
\usepackage{bbm}
\usepackage{latexsym}
\usepackage{caption}
\usepackage[margin=20pt]{subcaption}
\usepackage{hyperref}
\usepackage{booktabs}
\usepackage{multirow}

\usepackage{enumitem}

\usepackage[table]{xcolor}

\newcommand{\GG}[1]{}

\usepackage{amsthm}
\usepackage{amssymb}
\usepackage{amsmath}
\usepackage{color}

\usepackage{comment}
\theoremstyle{definition}

\newtheorem*{theorem*}{Theorem}
\newtheorem{theorem}{Theorem}
\newtheorem*{rmk*}{Remark}

\newtheorem{lemma}{Lemma}

\newtheorem{condition}{Condition}

\newtheorem{remark}{Remark}
\newtheorem{corollary}{Corollary}
\newtheorem*{corollary*}{Corollary}

\usepackage{natbib} %
\bibpunct{(}{)}{;}{a}{}{,} %

\usepackage{etoolbox} %
\apptocmd{\sloppy}{\hbadness 10000\relax}{}{} %

\usepackage{color}
\usepackage{listings}

\DeclareMathOperator*{\argmax}{arg\,max}

\def\Pr{\mathbb{P}}

\def\Var{\text{Var}}
\def\Cov{\text{Cov}}

\def\converged{\stackrel{d}{\longrightarrow}}
\def\convergep{\stackrel{\Pr}{\longrightarrow}}

\def\Pr{\text{pr}}
\def\bs{}
\def\id{\text{id}}
\def\E{E}
\def\loss{\ell}

\def\logit{\text{logit}}
\def\MB{\textup{B}}
\def\MI{\textup{I}}
\def\MA{\textup{A}}
\def\converged{\stackrel{d}{\longrightarrow}}
\def\convergep{\stackrel{p}{\longrightarrow}}
\def\diag{\text{diag}}

\def\limsup{\overline{\lim}}
\def\liminf{\underline{\lim}}
\def\lpc{\dot\sim}
\def\lpdist{\pi}

\def\rev{\color{black}}
\def\lad{\color{black}}
\def\new{}

\RequirePackage[normalem]{ulem}

\usepackage{soul}

\usepackage{subcaption}
\usepackage[textsize=tiny, textwidth = 2cm, shadow]{todonotes}

\begin{document}

\onehalfspacing

\title{\bf 
Randomization-based Z-estimation 
for evaluating average and individual
treatment effects 
}
\author{
Tianyi Qu, Jiangchuan Du, 
and Xinran Li
\footnote{
Tianyi Qu, Department of Statistics, University of Illinois, Champaign, IL (e-mail: \href{mailto:tianyiq3@illinois.edu}{tianyiq3@illinois.edu}).
Jiangchuan Du, Department of Statistics, University of Chicago, Chicago, IL (e-mail: \href{mailto:jiangchuand@uchicago.edu}{jiangchuand@uchicago.edu}). 
Xinran Li, Department of Statistics, University of Chicago, Chicago, IL (e-mail: \href{mailto:xinranli@uchicago.edu}{xinranli@uchicago.edu}). 
}
}
\date{}
\maketitle

\begin{abstract}
Randomized experiments have been the gold standard for drawing causal inference. 
The conventional model-based approach has been one of the most popular ways for analyzing treatment effects from randomized experiments, which is often carried through inference for certain model parameters. 
In this paper, we provide a systematic investigation of model-based analyses for treatment effects 
under the randomization-based inference framework. 
This framework does not impose  
any distributional assumptions on the outcomes, covariates and their dependence, 
and utilizes only randomization as the ``reasoned basis''. 
We first derive the asymptotic theory for Z-estimation in completely randomized experiments, and propose sandwich-type conservative covariance estimation. 
We then apply the developed theory to analyze both average and individual treatment effects in randomized experiments.
For the average treatment effect, we consider three estimation strategies: model-based, model-imputed, and model-assisted,
where the first two can be sensitive to model misspecification or require specific ways for parameter estimation. 
The model-assisted approach is robust to arbitrary model misspecification and always provides consistent average treatment effect estimation. 
We propose optimal ways to conduct model-assisted estimation using generally nonlinear least squares for parameter estimation.  
For the individual treatment effects, we 
propose to directly model the relationship between individual effects and covariates, 
and discuss the model's identifiability, inference and interpretation allowing model misspecification.
\end{abstract}


\section{Introduction}\label{sec:intro}

Since \citet{fisher1925statistical}, 
randomized experiments have become the ``gold standard'' for drawing causal inference. 
There are at least two popular inferential frameworks for analyzing causal/treatment effects in randomized experiments.
One is the finite population inference \citep{fisher1935design, neyman1923}, also called randomization-based or design-based inference, which uses solely randomization of the treatment assignment as the ``reasoned basis'' and focuses on the units in the experiment. 
The other is the usual superpopulation inference that assumes random, e.g., independent and identically distributed (i.i.d.), sampling of units from some population, 
often with model assumptions on potential outcomes as well as their dependence on covariates. 
The finite population inference can be more preferable when experiments involve convenient samples, 
or when the superpopulation is hard to articulate.   
For example, when evaluating certain policy effects on the 50 states of US, it may not be clear what the corresponding superpopulation is.
In addition, the finite population inference naturally avoids distributional or model assumptions on potential outcomes and covariates.
See 
\citet{Little2004},
\citet{KG2006}, \citet{Rosenbaum02a}, \citet{Rosenberger2019} and \citet{Abadie2020}
for more comparisons between these two modes of inference.

Superpopulation inference is more flexible for incorporating covariates, often through modeling the conditional distributions of potential outcomes given covariates. 
It can help conducting efficient inference for treatment effects.  
However, such an inference can be sensitive to model misspecification. 
We will conduct a systematic investigation for analyses of randomized experiments that involve statistical models, under the finite population inference allowing arbitrary model misspecification. 
That is, we view the invoked model as a \textit{working} model 
rather than a true \textit{data-generating} model.
We will investigate 
the robust utilization of models 
for evaluating both average and individual treatment effects. 
To do this, an important first step is to understand large-sample properties of the parameter estimation under randomization.

We first establish the theory for M-estimation and Z-estimation under the randomization-based inference framework. In contrast to usual superpopulation inference, the observations can be dependent across units due to the dependence in their treatment assignments, 
and we may not be able to consistently estimate the asymptotic distributions of the estimators due to the unidentifiable joint distributions of the potential outcomes. 
We give rigorous conditions for the consistency and asymptotic Normality of the M-estimators and Z-estimators under randomization, 
and propose asymptotically conservative covariance estimators. 
We then apply the theory to analyze both average and individual treatment effects.

We consider three strategies for inferring average treatment effects.
The model-based approach, which often uses the estimated model parameters, is generally biased. 
The model-imputed approach, which uses models as a tool for imputing the potential outcomes, is  generally biased as well.  
However,
when model parameters are estimated in particular ways, 
it can provide consistent treatment effect estimation that is robust to any model misspecification. 
Finally, we recommend 
the model-assisted approach, which uses models to adjust potential outcomes and 
can always provide consistent treatment effect estimation, 
and suggest optimal adjustment using generally nonlinear least squares for parameter estimation. 
Our study on average treatment effects 
is related to many existing literature. 
Specifically, 
our study for model-based and model-imputed approaches generalizes \citet{Freedman08logistic}'s study on logistic regression to general models, 
and 
our study for model-imputed and model-assisted approaches 
intersects
with \citet{GuoBasse21} and \citet{Fogarty2020calibration}, 
but with different perspectives. 
Our work also shares similar spirits as 
\citet{Robins2005}, 
\citet{Rv2010}, 
\citet{Wooldridge21}, 
and 
\citet{VD2002}. 
However, 
we differ from these by not imposing i.i.d. assumptions.

We further explore how to robustly use models to  understand individual treatment effects. 
We propose to directly impose working models for the relationship between individual effects and covariates,  despite none of individual effects are observed, and discuss the models' identifiability and inference. 
The fitted model can be viewed as ``best approximation'' for the distribution of individual effects given covariates.

\section{Notation and Framework}
\subsection{Potential outcomes, covariates and treatment assignments}\label{sec:po}

We consider an experiment with $N$ units. 
For each unit $i$, let $Y_i(1)$ and $Y_i(0)$ denote the treatment and control potential outcomes, 
and $\bs{X}_i$ denote the pretreatment covariate vector. 
The average treatment and control potential outcomes are, respectively, 
$\bar{Y}(1) = N^{-1} \sum_{i=1}^N Y_i(1)$ and $\bar{Y}(0) = N^{-1} \sum_{i=1}^N Y_i(0).$
The average treatment effect in the difference scale is then $\bar{Y}(1) - \bar{Y}(0)$. 
More generally, we may be interested in the average treatment effect in a $g$-scale defined as 
$
\tau_g  = g(\bar{Y}(1)) - g(\bar{Y}(0))
$
for some prespecified $g$. 
For example, with binary outcomes, 
$g = \id$ corresponds to the risk difference, 
$g = \log$ corresponds to logarithm of the risk ratio, 
and 
$g = \text{logit}$ corresponds to logarithm of the odds ratio.

For $1\le i \le N$, 
let $Z_i$ be the treatment assignment indicator for unit $i$, 
where $Z_i$ equals $1$ if the unit receives treatment and $0$ otherwise. 
The observed outcome for each unit is then one of its two potential outcomes depending on the treatment assignment, i.e., $Y_i = Z_i Y_i(1) + (1-Z_i) Y_i(0)$.

To simplify the notation, 
for any finite population $\{c_1, c_2, \ldots, c_N\}$ of size $N$, 
let $\E_N(c_i) = N^{-1} \sum_{i=1}^N c_i$ and $\Var_N(c_i) = (N-1)^{-1} \sum_{j=1}^N \{ c_j - \E_N(c_i) \}^2$ be the finite population average and variance, respectively. 
Given the treatment assignment $(Z_1, \ldots, Z_N)$, 
we further introduce $\E_N^z(c_i) = n_z^{-1} \sum_{i:Z_i=z} c_i$ and $\Var_N^{z} = (n_z-1)^{-1} \sum_{j: Z_j = z}\{c_j - \E_N^z(c_i)\}^2$ to denote the sample average and variance under treatment arm $z$, for $z = 0,1$, 
{\rev where $n_1$ and $n_0$ denote the numbers of treated and control units.} 
We define analogously the finite population covariance $\Cov_N(\cdot)$ and the sample covariance $\Cov_N^{z}(\cdot)$ for $z=0,1$. 
For descriptive convenience, 
we introduce $\mathcal{Y}$ and $\mathcal{X}$ to denote the supports of outcome and covariate variables, respectively. 
For a general function $a(y, \bs{x}; \bs{\theta})$ of the outcome, covariates and parameter, 
let $\dot a(y, \bs{x}; \bs{\theta})$ 
denote its 
partial derivative over $\bs{\theta}$.

\subsection{Randomization-based and model-based inference}
As discussed in \S \ref{sec:intro}, 
there are two popular inferential frameworks for analyzing randomized experiments: 
finite population inference 
and usual superpopulation inference. 
Under finite population inference, all the potential outcomes and covariates are viewed as fixed constants (or equivalently being conditioned on), and we use solely the randomization of treatment assignment as the reasoned basis. 
Under superpopulation inference, some model assumptions are often imposed on
the relation between potential outcomes and covariates, 
and the inference is often driven by classical likelihood theory.

In general, 
if the model is correctly specified, 
then the model-based approach can provide the most efficient estimation for model parameters, which are often linked to and interpreted as treatment effects. 
However, in practice, the model may be misspecified. 
Furthermore, even if the model is correctly specified, the targeted model parameter may be different from true causal effects of interest. 
{\rev For example, 
under the finite population inference, 
\citet{Freedman08logistic} 
showed that,} with binary outcomes, the estimator from usual logistic regression model is generally inconsistent for the true average effect $\tau_g$ with $g=\logit$, an issue that cannot be avoided even when the logistic regression model is correctly specified.

We will investigate the model-based analysis for randomized experiments under the finite population inference framework, which allows arbitrary model misspecification. 
We focus on the completely randomized experiment (CRE), 
under which $n_1$ units are randomly assigned to treatment and the remaining $n_0 = N-n_1$ to control. 
Let $r_1 = n_1/N$ and $r_0 = n_0/N$ be the proportions of treated and control units.

\section{M-estimation and Z-estimation under Randomization}\label{sec:m_est}

\subsection{M-estimators and Z-estimators}

We consider estimating a parameter $\theta\in \mathbb{R}^p$ under a CRE, by either minimizing a certain risk function:
\begin{align}\label{eq:M_hat}
	\hat{M}_N (\bs{\theta})
	= 
    \E_N\{
    \loss_{Z_i} (Y_i, \bs{X}_i; \bs{\theta})
    \}
	=
	r_1 \E_{N}^1 \{\loss_1(Y_i, \bs X_i;{\bs\theta})\}
	+
	r_0 \E_{N}^0 \{\loss_0(Y_i, \bs X_i;{\bs\theta})\},
\end{align}
or finding the root of a certain estimating equation: 
\begin{align}\label{eq:Z_hat}
	\hat{\Psi}_N (\bs{\theta})
	= 
    \E_N\{
    \psi_{Z_i} (Y_i, \bs{X}_i; \bs{\theta})
    \}
	=
	r_1 \E_{N}^1 \{\psi_1(Y_i, \bs X_i;{\bs\theta})\}
	+
	r_0 \E_{N}^0 \{\psi_0(Y_i, \bs X_i;{\bs\theta})\}.
\end{align}
In \eqref{eq:M_hat} and \eqref{eq:Z_hat}, $\loss_z$ and $\psi_z$ are prespecified loss and estimating functions taking values in $\mathbb{R}$ and $\mathbb{R}^p$, respectively, for observations under treatment arm $z \in \{0, 1\}$. 
The estimators from \eqref{eq:M_hat} and \eqref{eq:Z_hat} are often called the M-estimator and Z-estimator, respectively. 
We will focus on the Z-estimator in the main paper, since the M-estimator can also be viewed as a Z-estimator by setting $\psi_z = \dot{\loss}_z$ for $z=0,1$. 
When we choose $\psi_1$ and $\psi_0$ to be the derivatives of the log-density functions from certain working models on the conditional distributions of $Y(1)$ and $Y(0)$ given $X$, the Z-estimator then reduces to the usual maximum likelihood estimator. 
{\rev See also \citet{Xu20} and \citet{han2024introduction} for recent progresses on M-estimation in finite population. 
}

\subsection{Finite population asymptotics and regularity conditions}\label{sec:fp_asym}

Studying the exact distribution of the Z-estimator $\hat{\bs{\theta}}_N$ 
is generally difficult and intractable. 
We therefore invoke the finite population asymptotics to study the large-sample properties of  $\hat{\bs{\theta}}_N$ \citep{fpclt2017}.  
Specifically, 
we embed the finite population of $N$ units into a sequence of finite populations, and study the limiting distribution of an estimator along the sequence of finite populations. 
Similar to usual superpopulation inference, we need to impose some regularity conditions on the sequence of finite populations, 
as well as some smoothness conditions on the estimating functions. %
Below we discuss these conditions.

First, 
we impose some regularity conditions on 
the finite population estimating equation: 
\begin{align}\label{eq:population_risk}
	\Psi_N (\bs{\theta})
	\equiv
	\E\{ \hat{\Psi}_N(\bs{\theta}) \}
	=
	r_1 \E_N\{\psi_1(Y_i(1), \bs X_i;{\bs\theta})\}
	+r_0 \E_N\{\psi_0 (Y_i(0), \bs X_i; {\bs\theta})\}, 
\end{align}
where the last equality holds due to the properties of the CRE.

\begin{condition}[\rev Interior root]\label{cond:population_minimum}
    Let $\Theta \subset \mathbb{R}^p$ be the parameter space, and 
	$\theta_N \in \Theta$ be the root of 
 $\Psi_N(\bs{\theta})$ in \eqref{eq:population_risk}.  
	There exists $\varepsilon_0 > 0$ such that $\mathcal{B}({\bs\theta}_N, \varepsilon_0 ) \equiv \{\theta: \|\theta-\theta_N\|<\varepsilon_0\} \subset \Theta$ for 
	all $N$.
\end{condition}

Second, we impose some smoothness conditions.  
For $z=0,1$, 
define 
$\psi_{zi}(\theta) = \psi_z(Y_i(z), \bs X_i;{\bs\theta})$, 
$
\Delta(\psi_{zi}, \theta_N, \delta)=\sup _{\bs\theta \in \Theta:\|\bs\theta-\bs\theta_N\|\le \delta}
\| 
\psi_{zi}(\theta) - \psi_{zi}(\theta_N)
\|, 
$ 
and analogously 
$\Delta(\dot{\psi}_{zi}, \theta_N, \delta)$, for $1\le i\le N$.

\begin{condition}[\rev Smoothness]\label{cond:continuous}
	For $z=0,1$ and $(\bs{x}, y)\in \mathcal{X} \times \mathcal{Y}$,
	$\psi_z(y, \bs{x}; \bs{\theta})$ 
    are continuously differentiable in $\bs{\theta} \in \Theta$.
    As $\delta \rightarrow 0$, 
     $\sup_{N} \E_N\{ \Delta(\dot \psi_{zi}, \theta_N, \delta)\}$ and 
    $\sup_{N} \E_N\{ \Delta^2(\psi_{zi}, \theta_N, \delta)\}$
    converge to zero.
\end{condition}

Third, 
we assume that some finite population quantities are bounded or 
not too heavy-tailed. We use $\sigma_{\max}$ and $\sigma_{\min}$ to denote the largest and smallest singular value of a matrix.

\begin{condition}[\rev Boundedness]\label{cond:stable_limits}
    Both $r_1$ and $r_0$ are bounded away from $0$ and $1$. 
    For $z=0,1$, 
    \vspace{-0.08in}
	\begin{enumerate}[label = (\roman*)]
		\item 
            there exists $c_0 > 0$ such that
            $\sigma_{\max}( \Cov_N\{\psi_{zi}(\bs{\theta}_N)\} ) \le  c_0$ for all $N$;
    
		\item  
            there exists $c_2 > c_1 > 0$ such that 
            $c_1 \le \sigma_{\text{min}}(\dot\Psi_N(\theta_N)) \le \sigma_{\text{max}}(\dot\Psi_N(\theta_N)) \le c_2$ for all $N$;
		
		\item  
		$
		N^{-1} \max_{1\le j \le N} \| \psi_{zj}(\bs{\theta}_N) - \E_N\{ \psi_{zi}(\bs{\theta}_N) \} \|^2 
		\rightarrow 0
		$
        as $N \rightarrow \infty$. 
	\end{enumerate}
\end{condition}

Fourth, we assume uniform convergence of the derivative of the empirical estimating function, as well as the consistency of the Z-estimator. 
In the supplementary material, we give sufficient conditions for them, using 
either compactness of the parameter space or convexity of the risk functions.

\begin{condition}[\rev Consistency]\label{cond:empirical_minimum}
The estimator $\hat{\bs{\theta}}_N$ is a root of 
$\hat{\Psi}_N(\bs{\theta})$ and is consistent for $\theta_N$, 
i.e., 
$\hat \theta_N - \theta_N = o_{\Pr}(1)$ as $N \rightarrow \infty$. 
There exists $r_0>0$ such that $\sup_{\bs\theta \in \Theta:\|\bs\theta-\bs\theta_N\|\le r_0}     \|\dot{\hat\Psi}_N(\bs\theta) - \dot{\Psi}_N(\bs\theta)\| = o_{\Pr}(1)$. 
\end{condition}

\subsection{Large-sample properties of the Z-estimator under randomization}

{\rev We say $A_N \lpc B_N$ if the
Lévy--Prokhorov
distance between probability measures induced by $A_N$ and $B_N$ converges to $0$ as $N\rightarrow \infty$.
In the theorem below, 
if further $\sigma_{\min}(\Sigma_N)$ is uniformly bounded away from zero for all $N$, then $\lpc$ can also be understood as the convergence of the Kolmogorov distance.

}

\begin{theorem}\label{thm:theta_clt}
	Under Conditions \ref{cond:population_minimum}--\ref{cond:empirical_minimum},
	$
	\sqrt{N} 
	( \hat{\bs{\theta}}_N - \bs{\theta}_N )
	\lpc 
	\mathcal{N}(\bs{0}, \bs{\Sigma}_N), 
	$
	where 
	\begin{align*}
		\bs{\Sigma}_N = r_1 r_0 \{\dot \Psi_N(\bs{\theta}_N)\}^{-1} \Cov_N\{ 
        \psi_{1i}(\bs{\theta}_N) - \psi_{0i}(\bs{\theta}_N)
		\}  \{\dot \Psi_N(\bs{\theta}_N)^\top \}^{-1}.
	\end{align*}

\end{theorem}

From Theorem \ref{thm:theta_clt}, 
the Z-estimator $\hat{\bs{\theta}}_N$ is consistent for its population analogue $\bs{\theta}_N$, and its asymptotic covariance matrix 
has 
the usual sandwich form. 
Importantly, 
the asymptotics in Theorem \ref{thm:theta_clt} is not driven by i.i.d.\ sampling of units, but rather the randomization of treatment assignments. 
Thus, Theorem \ref{thm:theta_clt} can be more general than the classical theory of Z-estimation in the context of randomized experiments, 
allowing flexible forms of heterogeneity and dependence across units' potential outcomes and covariates.

In practice, we are often interested in constructing confidence intervals for $\theta_N$ or its transformations, which further requires estimating the asymptotic covariance $\Sigma_N$. 
However, due to no joint observation of both potential outcomes for the same unit, there is generally no consistent covariance estimator
\citep{neyman1923}. 
Nevertheless, we can still conservatively estimate $\Sigma_N$ by, e.g., the following sample analogue:  
\begin{align*}
	\hat{\bs{\Sigma}}_N 
	& =\{\dot{\hat{\Psi}}_N(\hat{\bs{\theta}}_N)\}^{-1} 
	\Big[
	r_0 {\Cov}_{N}^1\{\psi_{1i}(\hat{\bs{\theta}}_N) \} + 
	r_1 {\Cov}_{N}^0\{\psi_{0i}(\hat{\bs{\theta}}_N) \}
	\Big]
	\{\dot{\hat{\Psi}}_N(\hat{\bs{\theta}}_N)^\top\}^{-1}. 
\end{align*}

\begin{theorem}\label{thm:m_est_variance_estimate}
	Under Conditions \ref{cond:population_minimum}--\ref{cond:empirical_minimum}, 
     {\rev 
     $\hat{\bs{\Sigma}}_N = \tilde{\bs{\Sigma}}_N+o_{\Pr}(1)$
     }
      for some $\tilde{\bs{\Sigma}}_N \ge \bs{\Sigma}_N$.
\end{theorem}

\begin{corollary}\label{cor:wald_ci_coverage}
{\rev 
In Theorem \ref{thm:m_est_variance_estimate}, 
if further $\sigma_{\min}(\tilde{\bs{\Sigma}}_N)$
is uniformly bounded away from zero, 
then for any $\bs{v} \in \mathbb{R}^{p\times m}$ of full column rank and $\alpha \in (0,1)$, 
an asymptotic $1-\alpha$ confidence set for $v^\top\theta_N$ is
$$\{\omega: 
N (v^\top\hat{\bs{\theta}}_N - \omega)^\top 
(v^\top \hat{\Sigma}_N v)^{-1} (v^\top\hat{\bs{\theta}}_N - \omega) 
\le \chi^2_{m, 1-\alpha}
\},
$$
where $\chi^2_{m, 1-\alpha}$ is the $1-\alpha$ quantile of the 
chi-squared distribution with degrees of freedom $m$. 
}
\end{corollary}

\section{Average treatment effects}\label{sec:ave_effect}

\subsection{Model-based and Model-imputed approaches}\label{sec:base_imp}

In this subsection we summarize main results for
the model-based and model-imputed approaches, with their details relegated to the supplementary material. 
The model-based and model-imputed estimators for the average treatment effect $\tau_g$ at the $g$-scale are, respectively, of forms:
\begin{align}\label{eq:MB_est}
    \hat{\tau}_{g\MB} & = 
	\E_N \big\{ g( h_1(\bs{X}_i; \hat{\bs{\theta}}_N) )  - g( h_0(\bs{X}_i; \hat{\bs{\theta}}_N)) \big\}, \\
    \label{eq:tau_hat_MI}
    \hat{\tau}_{g \MI} & = 
	g \big( \E_N\{ h_1(\bs{X}_i; \hat{\bs{\theta}}_N ) \} \big) - g\big( \E_N \{ h_0(\bs{X}_i; \hat{\bs{\theta}}_N ) \} \big). 
\end{align}
In \eqref{eq:MB_est} and \eqref{eq:tau_hat_MI}, $h_1$ and $h_0$ are often conditional mean functions of $Y(1)$ and $Y(0)$ given $X$ under a working model, and $\hat\theta_N$ is the corresponding maximum likelihood estimator. 
The asymptotic properties of 
$\hat{\tau}_{g\MB}$ and $\hat{\tau}_{g\MI}$ and their inference  
can be derived using Theorems \ref{thm:theta_clt} and \ref{thm:m_est_variance_estimate} with Delta method \citep{PASHLEY2022109540}. 

The model-based estimator in \eqref{eq:MB_est} often corresponds to model parameter estimator. For example, when the working model is a generalized linear model with canonical link $g$ and no treatment-covariate interaction, $\hat{\tau}_{g\MB}$ becomes the same as the estimator for the coefficient of the treatment indicator. 
In general, $\hat{\tau}_{g\MB}$ is not consistent for $\tau_g$, 
underscoring the risk of misinterpreting model parameters as treatment effects.
Moreover, 
this is true even when the model  is correctly specified, which 
is related to the distinction between conditional and marginal effects \citep{Rosenblum15, Ding2022comment}.

The model-imputed estimator in \eqref{eq:tau_hat_MI} uses the fitted model  to estimate or impute the potential outcomes and then estimates $\tau_g$ in a plug-in way. 
Unless the model is correctly specified, $\hat{\tau}_{g \MI}$ is in general inconsistent for $\tau_g$. 
Fortunately, 
$\hat{\tau}_{g \MI}$ can be consistent for $\tau_g$ under arbitrary model specification, given that the parameter $\theta_N$ is estimated in the following way: 
$\hat{\theta}_N$ is a Z-estimator from an estimating equation as in \eqref{eq:Z_hat}, 
and there exist differentiable 
functions $\bs{q}_1(\bs{\theta})$ and $\bs{q}_0(\bs{\theta})$ such that,  
for all $\bs{x}, y, \bs{\theta}$ and $z=0,1$, 
\begin{align}\label{eq:cond_MI_cons}
    \bs{q}_z(\bs{\theta})^\top \psi_z(y, \bs{x}; \bs{\theta}) = y - h_z(\bs{x}; \bs{\theta}), 
    \quad 
    \bs{q}_z(\bs{\theta})^\top  \psi_{1-z}(y, \bs{x}; \bs{\theta}) = 0. 
\end{align}
This condition is automatically satisfied when the working model is a generalized linear model with a canonical link and 
the parameter is estimated using the maximum likelihood approach.

\subsection{Model-assisted approach: consistency, asymptotic Normality and variance estimation}\label{sec:MA_consist}

In the remaining of this section, 
we focus on the model-assisted approach, which 
can give consistent treatment effect estimation under arbitrary model misspecification, 
without requiring
any specific ways for parameter estimation. 
Given general adjustment functions $h_z(\bs{x}; \bs{\theta})$ for $z=0,1$ and $\theta \in \Theta$, 
let $Y_i(z; \bs{\theta}) = Y_i(z) - h_z( \bs{X}_i; \bs{\theta} ) + \E_N \{ h_z(\bs{X}_i; \bs{\theta}) \}$ 
be an adjusted potential outcome for unit $i$. 
Importantly, 
the average of the adjusted potential outcomes is the same as that of the original ones, 
i.e., 
$\E_N\{Y_i(z; \bs{\theta})\} = \E_N\{Y_i(z)\}$ for $z=0,1$. 
This is crucial for the consistency of the model-assisted estimator, and its idea dates back to regression estimates in survey sampling
\citep{cochran1977sampling}. 
For each unit $i$, let 
$Y_i(\bs{\theta}) = Z_i Y_i(1;\bs{\theta}) + (1-Z_i) Y_i(0;\bs{\theta})$
be the observed adjusted outcome. 
Under the CRE, 
for any fixed $\bs{\theta}$ and $z=0,1$, 
the average observed adjusted outcome 
$
\bar{Y}_z(\bs{\theta})$$
\equiv 
\E_{N}^z \{Y_i(\bs{\theta})\}
$ 
under treatment arm $z$
is 
unbiased 
for the average potential outcome $\bar{Y}(z)$. 
This motivates the following model-assisted estimator: 
\begin{align}\label{eq:tau_hat_MA}
	\hat{\tau}_{g\MA} & = 
	g( \bar{Y}_1(\hat{\bs{\theta}}_N) ) - 
	g( \bar{Y}_0(\hat{\bs{\theta}}_N) )
    = 
	g\big( \E_N^1 \{ Y_i(\hat{\bs{\theta}}_N) \} \big)
	-
	g\big( \E_N^0\{ Y_i(\hat{\bs{\theta}}_N) \} \big),
\end{align} 
where $\hat{\bs{\theta}}_N$ is either predetermined or estimated from the data, such as M- or Z-estimators from \eqref{eq:M_hat} or \eqref{eq:Z_hat}.

We now study the asymptotic distribution of $\hat{\tau}_{g\MA}$. 
The key is to prove  $\hat{\tau}_{g\MA}$ is asymptotically equivalent to the one in \eqref{eq:tau_hat_MA} with $\hat{\bs{\theta}}_N$ replaced by its probability limit $\bs{\theta}_N$. 
Below we introduce some conditions on the smoothness of adjustment functions and boundedness of some finite population quantities. 
For $z=0,1$,  
define   
$
\Delta(h_{zi}, \theta_N, \delta)
\equiv 
\sup _{\bs\theta \in \Theta:\|\bs\theta-\bs\theta_N\| \leq \delta}
\| h_z(\bs{X}_i,\bs\theta)- h_z( \bs{X}_i, \bs\theta_N )\|,  
$
and analogously $\Delta(\dot{h}_{zi}, \theta_N, \delta)$.

\begin{condition}
[\rev Smoothness of adjustment]
\label{cond:ma_continuous}
	Both $h_1(\bs{x}, \bs{\theta})$ and $h_0(\bs{x}, \bs{\theta})$ are continuously differentiable in $\bs{\theta} \in \Theta$ for  $\bs{x} \in \mathcal{X}$, 
	and 
	$g(y)$ is continuously differentiable in $y \in \mathcal{Y}$.
	Moreover, 
	for $z=0,1$, as $\delta \rightarrow 0$, 
	both $\sup_N \E_N\{ \Delta(\dot{h}_{zi}, \theta_N, \delta) \}$ and $\sup_N \E_N\{ \Delta^2(h_{zi}, \theta_N, \delta) \}$ converge to zero. 
\end{condition}

\begin{condition}
[\rev Boundedness of adjusted outcomes]
\label{cond:ma_stable}
    $\mathcal{B}({\bs\theta}_N, \varepsilon_0 ) \subset \Theta$ for some $\varepsilon_0 > 0$ and 
	all $N$, and%
    \vspace{-0.08in}
	\begin{enumerate}[label = (\roman*)]
		\item for all $N$, both $r_1$ and $r_0$ are bounded away from $0$ and $1$, 
        and
		both $\bar Y(1)$ and $\bar Y(0)$ are bounded; 
		\item the finite population variances 
        of $Y_i(1; \bs{\theta}_N)$ and $Y_i(0; \bs{\theta}_N)$ are bounded for all $N$,  
		and, as $N\rightarrow \infty$, $ N^{-1} \max _{1 \leq j \leq N}[ Y_j(z; \bs{\theta}_N)- \E_N\{Y_i(z; \bs{\theta}_N)\} ]^{2} \rightarrow 0$ for $z=0,1$; 
		\item as $N\rightarrow \infty$, 
		$N^{-1}\Cov_N\{\dot h_z(\bs{X}_i,{\bs\theta}_N)\}\rightarrow \bs{0}$ for $z=0,1$. 
		
	\end{enumerate}
\end{condition}

\begin{theorem}\label{thm:clt_mae}
	If $\hat{\bs{\theta}}_N-\bs{\theta}_N=O_{\Pr}(N^{-1/2})$, 
	 and Conditions \ref{cond:ma_continuous} and \ref{cond:ma_stable} hold,  
	then 
	$\sqrt{N}( \hat{\tau}_{g\MA} - \tau_{g} ) \lpc \mathcal{N}(0, V_{g\MA})$, where, 
    with $\dot{g}$ being the derivative of $g$, 
	\begin{align}\label{eq:V_g_MA}
		V_{g\MA} 
		& = 
		r_1^{-1} \Var_N \big\{ \dot g(\bar{Y}(1)) Y_i(1; \bs{\theta}_N) \big\} + 
		r_0^{-1} \Var_N \big\{ \dot g(\bar{Y}(0)) Y_i(0; \bs{\theta}_N) \big\} 
		\nonumber
		\\
		& \quad \ 
		- 
		\Var_N \big\{ \dot g(\bar{Y}(1)) Y_i(1; \bs{\theta}_N) - \dot g(\bar{Y}(0)) Y_i(0; \bs{\theta}_N)  \big\}. 
	\end{align}
\end{theorem}

Unlike the model-imputed estimator whose consistency requires \eqref{eq:cond_MI_cons}, 
Theorem \ref{thm:clt_mae}
does not impose any condition on the parameter estimation and 
works for general adjustment functions $h_z$s.  
When $\hat{\theta}_N$ is a Z-estimator from an estimating equation satisfying \eqref{eq:cond_MI_cons}, 
the model-imputed and model-assisted estimators are asymptotically equivalent.
In addition,
we can use the model-assisted approach to infer more general estimands of form $\nu(\bar{Y}(1), \bar{Y}(0))$. 
We relegate the details to the supplementary material.

Both the model-imputed estimator with parameter estimation satisfying \eqref{eq:cond_MI_cons} and the model-assisted estimator fall in the class of Oaxaca-Blinder estimators considered by \citet{GuoBasse21}. 
Here we distinguish these two estimators based on the difference in their construction, and highlight the inconsistency of the general model-imputed estimator.
Moreover, we introduce the model-assisted approach through the lens of adjusted potential outcomes, which can not only ease the understanding of its consistency and robustness but also facilitate the study on optimal adjustment discussed shortly in \S \ref{sec:optimal_loss_MA}.

We finally 
consider variance estimation. 
Similar to 
Theorem \ref{thm:m_est_variance_estimate}, since we cannot observe both potential outcomes for any unit, we generally cannot consistently estimate $V_{g\MA}$, particularly the last term in \eqref{eq:V_g_MA}. 
Instead, we consider 
a conservative variance estimator
$
	\hat{V}_{g\MA} 
	=
	r_1^{-1} \Var_N^1 \{ \dot g(\bar{Y}_1(\hat{\theta}_N)) Y_i(\hat{\bs{\theta}}_N) \} + 
	r_0^{-1} 
	\Var_N^0 \{ \dot g(\bar{Y}_0(\hat{\theta}_N)) Y_i(\hat{\bs{\theta}}_N) \}, 
$ 
which is a sample analogue of $V_{g\MA}$ excluding the last unidentifiable term.

\begin{theorem}\label{thm:MA_ci}
	If $\hat{\bs\theta}_N-{\bs\theta}_N=O_\Pr(N^{-1/2})$
	and 
	Conditions \ref{cond:ma_continuous} and \ref{cond:ma_stable} hold, 
	then 
    {\rev $\hat{V}_{g\MA} =  \tilde{V}_{g\MA}+ o_{\Pr}(1)$}, with 
	\begin{align}\label{eq:Vtilde_gMA}
	    \tilde{V}_{g\MA} = r_1^{-1} \Var_N \big\{ \dot g(\bar{Y}(1)) Y_i(1; \bs{\theta}_N) \big\} + 
		r_0^{-1} \Var_N \big\{ \dot g(\bar{Y}(0)) Y_i(0; \bs{\theta}_N) \big\} 
		\ge V_{g\MA}. 
	\end{align}
\end{theorem}

\subsection{Model-assisted approach: optimal parameter estimation}\label{sec:optimal_loss_MA}

From Theorem \ref{thm:clt_mae}, 
the model-assisted estimator is always consistent for the average effect $\tau_g$, regardless of the value of $\hat{\bs{\theta}}_N$ or its  probability limit $\bs{\theta}_N$.
It is natural to ask: what is the best choice of $\bs{\theta}_N$? 
Given the forms of $h_1$ and $h_0$, 
asymptotically, 
the optimal $\bs{\theta}_N$ 
is the one minimizing the asymptotic variance $V_{g\MA}$. 
However, 
such an optimal choice may not be achievable,  
since 
there may be no consistent estimation for $V_{g\MA}$ 
as discussed in \S \ref{sec:MA_consist}. 
Instead, we will consider 
$\bs{\theta}_N$ that lead to the smallest estimated variance for $\hat{\tau}_{g\MA}$.

Below we assume that $h_1$ and $h_0$ do not share any common parameters, which can be enforced by the analyzer. 
To minimize $\tilde{V}_{g\MA}$ in \eqref{eq:Vtilde_gMA}, it suffices to minimize both 
$\Var_N \{ Y_i(1; \bs{\theta}) \}$ and $\Var_N \{ Y_i(0; \bs{\theta}) \}$ simultaneously. 
Note that the model-assisted estimator $\hat{\tau}_{g\MA}$ in \eqref{eq:tau_hat_MA} is invariant under any constant shift of the adjustment functions $h_1$ and $h_0$, and  
$(N-1)\Var_N \{ Y_i(z; \bs{\theta}) \} = \min_{\kappa_z} \sum_{i=1}^N \{Y_i(z) - \kappa_z -  h_z(\bs{x}; \bs{\theta})\}^2$ for $z=0,1$.
We can always augment $h_1$ and $h_0$ with intercept terms if they do not have, 
under which the optimal $\theta$ is achieved with 
squared loss functions, as summarized in the following theorem.

\begin{theorem}\label{thm:ma_optimal}
	Consider adjustment functions $h_1(\bs{x}; \bs{\theta})$ and $h_0(\bs{x}; \bs{\theta})$ augmented with intercept terms and depending on disjoint subsets of the parameter $\bs{\theta}$. 
	The
	probability limit $\tilde{V}_{g\MA}$ in \eqref{eq:Vtilde_gMA} of the conservative variance estimator 
	achieves its minimum at $\bs{\theta} = \bs{\theta}_N$, 
    which minimizes the expected value of the empirical risk function in \eqref{eq:M_hat} with squared losses 
	$\loss_z(y, \bs{x}; \bs{\theta}) = \{ y - h_z(\bs{x}; \bs{\theta}) \}^2$ for $z=0,1$.  
	Consequently, 
	under Conditions \ref{cond:population_minimum}--\ref{cond:ma_stable}, 
    the model-assisted estimator $\hat{\tau}_{g\MA}$ in \eqref{eq:tau_hat_MA} with $\hat{\bs{\theta}}_N$, obtained from minimizing the empirical risk function $\hat{M}_N(\bs{\theta})$ in \eqref{eq:M_hat} with squared losses, achieves the optimal estimated precision. 
\end{theorem}

From Theorem \ref{thm:ma_optimal}, 
the generally nonlinear least squares estimator for $\theta$ 
can be more preferable than usual maximum likelihood estimator based on some working models. 
Moreover, 
if there exists a value of $\theta$ such that both $h_1(\bs{x}; \bs{\theta})$ and $h_0(\bs{x}; \bs{\theta})$ are constants that do not vary with $x$, then the optimal model-assisted estimator from Theorem \ref{thm:ma_optimal} will be non-inferior to the unadjusted estimator.
Theorem \ref{thm:ma_optimal} focuses on choice of parameter given the form of adjustment functions;
the choice of the latter can also be critical.

\section{Individual Treatment Effects}\label{sec:ite}

We now study how to robustly utilize models to understand individual treatment effects and their heterogeneity. 
As discussed in \S \ref{sec:base_imp}, 
we often impose models on the relationship between potential outcomes and covariates. We can then use the fitted models to estimate individual treatment effects. 
However, the interpretation can be difficult under model misspecification. 
We relegate the details to the supplementary material. 
Below we propose an alternative strategy focusing directly on individual treatment effects.

We consider directly imposing \textit{working} models on the 
the relationship between individual effects and covariates. 
This has the advantage that, even under model misspecification, the fitted model still provides meaningful approximations for individual effects. 
However, 
because none of the individual effects are observed, the model fitting may not be feasible based on observed data, or the model may not be identifiable. 

To overcome the challenge, we restrict to 
the following exponential dispersion family of models:  
\begin{align}\label{eq:model_tau}
    \text{working model of }\tau \mid x 
    \ \sim \ 
    f(\tau, x; \theta, \phi) 
    = 
    \exp\left\{ \frac{\tau \cdot v(x; \theta) - u(x;\theta)}{a(\bs{\phi})} + c(\tau, x, \bs{\phi})\right\}
\end{align}
with $a(\phi)>0$, 
which includes many generalized linear models as special instances. 
We are interested in 
\begin{align*}%
    (\theta_N, \phi_N) & = 
    \argmax_{(\theta, \phi)} \E_N [ \log f(\tau_i, X_i; \theta, \phi) ]
    = 
    \argmax_{\theta, \phi} \E_N \Big\{  
    \frac{\tau_i \cdot v(X_i; \theta)  - u(X_i;\theta)}{a(\phi)} + c(\tau, x, \bs{\phi}) \Big\}. 
\end{align*} 
When the model is correctly specified, $f(\tau_i, X_i; \theta_N, \phi_N)$ essentially gives the conditional distribution of individual treatment effects given covariates. 
In general, when the model can be misspecified, 
similar to usual superpopulation inference \citep{White82}, $(\theta_N, \phi_N)$ can be understood as the limit of the quasi-maximum likelihood estimator, and 
the corresponding $f(\tau, x; \theta_N, \phi_N)$ can be viewed as the best distributional approximation among \eqref{eq:model_tau} for the relation between $\tau_i$s and $X_i$s. 
This can help us understand individual treatment effects and their heterogeneity with respect to the observed covariates.

We now discuss the identifiability and inference 
of $(\theta_N, \phi_N)$. 
In general, 
the dispersion parameter $\phi_N$ may not be identifiable from the observed data, since none of the $\tau_i$s are observed. However, $\theta_N$ is always identifiable. 
Assuming both $v$ and $u$ are differentiable over $\theta$, 
$\theta_N$ is equivalently the maximizer or root of 
\begin{align}\label{eq:M_Z_ind}
    M_N(\theta) = \E_N \{  \tau_i \cdot v(X_i; \theta)  - u(X_i;\theta) \}
    \ \text{ or } \ 
    \Psi_N(\theta) = \E_N \{ \tau_i \cdot \dot{v}(X_i; \theta) - \dot{u}(X_i;\theta) \}. 
\end{align}
Importantly, $M_N(\theta)$ and $\Psi_N(\theta)$ are linear in the $\tau_i$s, which is crucial for the identifiability of $\theta_N$ and explains our choice of models in \eqref{eq:model_tau}. 
This is because we can unbiasedly estimate the individual effect $\tau_i$ by $\hat{\tau}_i = Z_i Y_i/r_1 -  (1-Z_i) Y_i/r_0$, which implies the following 
intuitive unbiased estimator for $\Psi_N(\theta)$: 
\begin{align}\label{eq:Psi_ind_effect_est}
     \hat{\Psi}_N(\theta) & 
     = \E_N \{ \hat{\tau}_i \cdot \dot{v}(X_i; \theta) - \dot{u}(X_i;\theta) \}
     = r_1 \E_N^1 \{\psi_1(Y_i, \bs X_i;{\bs\theta})\}
	+r_0 \E_N^0 \{\psi_0 (Y_i, \bs X_i; {\bs\theta})\},
\end{align}
where $\psi_1(y,x;\theta) =  
 y/r_1 \cdot \dot{v}(x; \theta)  - \dot{u}(x;\theta)$ 
and 
$\psi_1(y,x;\theta) =  
 -y/r_0 \cdot \dot{v}(x; \theta)  - \dot{u}(x;\theta)$.
We can find the root of \eqref{eq:Psi_ind_effect_est} to obtain a Z-estimator $\hat{\theta}_N$.  
From 
\S \ref{sec:m_est}, 
$\hat{\theta}_N$ can be consistent for $\theta_N$ and asymptotically Normal, based on which we can further construct confidence intervals for $\theta_N$. 
As a side note, 
the unbiased estimators $\hat{\tau}_i$s for individual effects in \eqref{eq:Psi_ind_effect_est} could potentially be improved using an idea in \citet{Semenova23}; we briefly discuss it in the supplementary material and leave details for future study.

Our results differ from and generalize \citet{Tian2014}.  
First, we propose to model directly the relationship between individual effects and covariates, which can ease the understanding and interpretation of the fitted models. 
Second, we propose a general class of models for which the parameters are estimable or identifiable; the corresponding maximum likelihood estimation includes several approaches in \citet{Tian2014} as special cases. Third, we conduct finite population randomization-based inference. 

Below we give an illustrative example. 
Consider the Normal working model: 
$\tau\mid x \sim \mathcal{N}(\mu(x;\theta), \phi^2)$. In this case, $\theta_N$ from \eqref{eq:M_Z_ind} minimizes $\E_N [ \{\tau_i - u(X_i;\theta)\}^2 ]$. 
Thus, we can view $u(X_i;\theta_N)$s as the best approximation for the $\tau_i$s.
When $u(x;\theta)$ is linear in $x$, 
$u(X_i;\theta_N)$s becomes the linear projections of individual effects on covariates. 
This leads to the decomposition
$\Var_N(\tau_i) = \Var_N\{u(X_i;\theta_N)\} + \Var_N\{\tau_i - u(X_i;\theta_N)\}$
used by \citet{dingdecomposition2019} to access treatment effect heterogeneity. 
In the supplementary material, we consider also 
analyses for binary outcomes, under which $\tau_i\in \{-1, 0, 1\}$.

\section*{Supplementary material}
Supplementary material 
includes 
(1) details for uniform convergence and consistency under M- and Z-estimation, 
(2) details of model-based and model-imputed estimators, 
(3) linear adjustment using imputed potential outcomes \citep{Fogarty2020calibration},   
(4) simulation studies comparing various estimators for average effects, 
(5) proofs 
of all theorems and corollaries, 
and (6) additional remarks and technical details.

\bibliographystyle{plainnat}
\bibliography{model.bib}

\newpage
\setcounter{equation}{0}
\setcounter{section}{0}
\setcounter{figure}{0}
\setcounter{example}{0}
\setcounter{proposition}{0}
\setcounter{corollary}{0}
\setcounter{theorem}{0}
\setcounter{lemma}{0}
\setcounter{table}{0}
\setcounter{condition}{0}

\renewcommand {\theproposition} {A\arabic{proposition}}
\renewcommand {\theexample} {A\arabic{example}}
\renewcommand {\thefigure} {A\arabic{figure}}
\renewcommand {\thetable} {A\arabic{table}}
\renewcommand {\theequation} {A\arabic{equation}}
\renewcommand {\thelemma} {A\arabic{lemma}}
\renewcommand {\thesection} {A\arabic{section}}
\renewcommand {\thetheorem} {A\arabic{theorem}}
\renewcommand {\thecorollary} {A\arabic{corollary}}
\renewcommand {\thecondition} {A\arabic{condition}}
\renewcommand {\thermk} {A\arabic{rmk}}

\renewcommand {\thepage} {A\arabic{page}}
\setcounter{page}{1}

\begin{center}
	\bf \LARGE 
	Supplementary Material 
\end{center}

Appendix \ref{sec:model} introduces 
the setup for a general model-based analysis of treatment effects. 

Appendix \ref{sec:based} studies asymptotic theory for model-based approach.

Appendix \ref{sec:impute} studies asymptotic theory for model-imputed approach.

Appendix \ref{sec:conn_mi_ma} studies connection between model-imputed and model-assisted approach.

Appendix \ref{sec:ai} studies linear adjustment based on imputed potential outcomes or transformed covariates.

Appendix \ref{sec:ind_po} studies individual effects from working models on potential outcomes.

Appendix \ref{sec:simu} includes simulation studies. 

Appendices  \ref{sec:M_srs} and \ref{sec:clt_srs} study M- and Z-estimation under simple random sampling. 

Appendices \ref{sec:M_cre} and \ref{sec:clt_cre} study M- and Z-estimation under completely randomized experiments.

Appendix \ref{sec:var_cre} studies variance estimation under completely randomized experiments.

Appendix \ref{sec:proof_thm12_cor1} proves Theorems \ref{thm:theta_clt} and \ref{thm:m_est_variance_estimate} and Corollary \ref{cor:wald_ci_coverage}. 

Appendix \ref{sec:proof_thm345} proves Theorems \ref{thm:clt_mae}, \ref{thm:MA_ci} and \ref{thm:ma_optimal}.

Appendix \ref{sec:proof_thm_cor_A} proves Theorems \ref{thm:clt_mbe}, \ref{thm:clt_ibe} and \ref{thm:clt_ai} and Corollaries \ref{cor:glm_est_eqn}, \ref{cor:model_imp_consist} and \ref{cor:equi_i_a}.

Appendix \ref{sec:add_detail} gives additional remarks and technical details.

\section{Working models for potential outcomes and covariates and maximum likelihood approach}\label{sec:model}

\subsection{Model setup}

We discuss the setup for a general model-based analysis of treatment effects. 
In the following discussion, we will view the models as \textit{working} models instead of true data-generating models, unless otherwise noted. 
Typically, 
we will impose some statistical models on the distribution of potential outcomes given covariates:
\begin{align}\label{eq:model}
	Y(1)\mid \bs{X} \sim f_1 (y, \bs{X}; \bs{\theta}), \quad 
	Y(0)\mid \bs{X} \sim f_0 (y, \bs{X}; \bs{\theta}),
\end{align}
where $f_1(\cdot, \bs{x}; \bs{\theta})$ and $f_0(\cdot, \bs{x}; \bs{\theta})$ denote two generic probability densities that can depend on the covariate vector $\bs{x}$ and some (unknown) parameter vector $\bs{\theta}\in \Theta$. 
The experimental units 
$(Y_i(1), Y_i(0), \bs{X}_i)$'s are then assumed to be i.i.d.\  samples from a superpopulation satisfying  \eqref{eq:model}. 
Here we model both potential outcomes in \eqref{eq:model} using the same parameter $\bs{\theta}$, which can allow the two models to share same parameters.  
Note that we are considering a CRE where the treatment assignment is unconfounded \citep{pscore1983}.  
Hence, 
the observed outcome for each unit given its covariates follows one of the two potential outcome models in \eqref{eq:model}, depending on the treatment assignment, i.e., 
$$
Y_i \mid \bs{X}_i, Z_i \sim Z_i \cdot f_1 (y, \bs{X}; \bs{\theta}) + (1-Z_i) \cdot f_0 (y, \bs{X};  \bs{\theta}). 
$$

\subsection{Generalized linear models}
We will 
pay special attention to the class of generalized linear models  
{\lad with canonical links}, 
which has been widely studied and used \citep[see, e.g.,][]{nelder1989glm, agresti2015foundations}. 
Specifically, we assume that 
the densities of both potential outcomes given covariates belong in a certain exponential dispersion family  
of the following forms:
\begin{align}\label{eq:glm}
	Y(z) \mid \bs{x} & \sim f_z(y, \bs{x}; \bs{\theta}) = \exp\left\{ \frac{y \cdot (\alpha_z + \bs{\beta}_z^\top \bs{x}) - b_z(\alpha_z + \bs{\beta}_z^\top \bs{x})}{a_z(\bs{\phi}_z)} + c_z(y, \bs{\phi}_z)\right\}, 
\end{align}
where $\bs{\theta} = (\alpha_1, \alpha_0, \bs{\beta}_1, \bs{\beta}_0, \bs{\phi}_1, \bs{\phi}_0)$ and $z=0,1$. 
By the property of the exponential dispersion family, 
the conditional mean functions of the potential outcomes given covariates under model \eqref{eq:glm} 
are 
$
h_z(\bs{x}; \bs{\theta}) 
= \dot b_z ( \alpha_z + \bs{\beta}_z^\top \bs{x} ) 
$
for $z=0,1$. 
When $b_1(\cdot) = b_0(\cdot)$ equals some $b(\cdot)$, 
both potential outcome models 
will share the same canonical link function, denoted by $g = (\dot b)^{-1}$, 
where $\dot{b}$ denotes the first derivative of $b$. 
Consequently, 
$g(h_z(\bs{x}; \bs{\theta})) = \alpha_z + \bs{\beta}_z^\top \bs{x}$, 
and 
the treatment effect at the $g$-scale for a unit with covariate value $\bs{x}$ is often interpreted as 
\begin{align}\label{eq:diff_g_h_ind}
	g\left( h_1\left( \bs{x}; \bs{\theta} \right) \right)
	-
	g\left( h_0\left( \bs{x}; \bs{\theta} \right) \right)
	& 
	=
	\alpha_1 - \alpha_0 + \left( \bs{\beta}_1 - \bs{\beta}_0 \right)^\top \bs{x}. 
\end{align}
Sometimes it is assumed that both potential outcome models share the same coefficient for covariates, i.e., $\bs{\beta}_1 = \bs{\beta}_0$. This is often termed as no interaction between treatment and covariates. 
In the absence of interaction, 
\eqref{eq:diff_g_h_ind} reduces to a constant $\alpha_1 - \alpha_0$ that does not depend on the covariates, 
under which it is often interpreted as the constant treatment effect at the $g$-scale. 
For example, in practice, we often write the working model in terms of the observed outcome given treatment assignment and covariates, 
which has the following equivalent forms under \eqref{eq:glm}:
\begin{align*}
    & \quad \  Y\mid Z, X 
    \\
    & \sim 
    f_Z(y, \bs{x}; \bs{\theta}) = \exp\left\{ \frac{y \cdot (\alpha_Z + \bs{\beta}_Z^\top \bs{x}) - b(\alpha_Z + \bs{\beta}_Z^\top \bs{x})}{a_Z(\bs{\phi}_Z)} + c_Z(y, \bs{\phi}_Z)\right\}
    \\
    & = 
    \exp\Big\{ \frac{y \{ \alpha_0 + (\alpha_1 - \alpha_0) Z + \bs{\beta}_0^\top \bs{x} + (\beta_1 - \beta_0)^\top \bs{x} Z \} - b( \alpha_0 + (\alpha_1 - \alpha_0) Z + \bs{\beta}_0^\top \bs{x} + (\beta_1 - \beta_0)^\top \bs{x} Z )}{a_Z(\bs{\phi}_Z)} \Big\} \\
    & \quad \ \cdot \exp\{c_Z(y, \bs{\phi}_Z)\}.
\end{align*}
The coefficient of $Z$, $\alpha_1 - \alpha_0$, is often used to infer treatment effects; 
when there is no treatment-covariate interaction, i.e., $\beta_1 = \beta_0$, $\alpha_1 - \alpha_0$ is often interpreted as certain constant treatment effects.

Table \ref{table:glm} lists the forms of $h_1$, $h_0$ and $g$ for three commonly used generalized linear models {\lad with canonical links}, i.e., linear, logistic and Poisson regression models for continuous, binary and count outcomes, respectively.

\begin{table}[ht]
	\centering
	\caption{
	Three commonly used generalized linear models.
	The four columns show the model name, the distribution of potential outcome given covariates, the canonical link, and the conditional mean function of potential outcome given covariates, respectively. 
	}\label{table:glm}
	\resizebox{\columnwidth}{!}{%
		\begin{tabular}{llll}
			\toprule
			Model & Distribution & Canonical link & Mean function\\
			\midrule
			Linear  & $Y(z) \mid \bs{X} \sim \mathcal{N}(\alpha_z + \bs{\beta}_z^\top \bs{X}, \sigma_z^2)$ & $g = \id$ & $h_z(\bs{x}; \bs{\theta}) = \alpha_z + \bs{\beta}_z^\top \bs{x}$
			\\
			Logistic & $\logit\left\{ \Pr\left( Y(z) = 1 \mid \bs{X} \right) \right\} 
			= \alpha_z + \bs{\beta}_z ^\top \bs{X}$ & $g = \logit$ 
			& 
			$h_z(\bs{x}; \bs{\theta}) = \logit^{-1}( \alpha_z + \bs{\beta}_z^\top \bs{x} )$
			\\
			Poisson & $
            Y(z) \mid \bs{X} \sim\text{Poisson}(\exp(\alpha_z + \bs{\beta}_z^\top \bs{X}))$
            & $g = \log$ 
			& 
			$h_z(\bs{x}; \bs{\theta}) = \exp( \alpha_z + \bs{\beta}_z^\top \bs{x} )$
			\\
			\bottomrule
		\end{tabular}
	}
\end{table}

\subsection{Maximum likelihood estimation}

In practice, the parameter $\theta_N$ is often estimated by the maximum likelihood approach under the imposed model. 
This is equivalent to the M-estimator in \eqref{eq:M_hat} with $\loss_z(y, \bs{x}; \bs{\theta}) = -\log f_z (y, \bs{x}; \bs{\theta})$ being the corresponding minus log-density function, 
as well as the Z-estimator in \eqref{eq:Z_hat} with $\psi_z(y, \bs{x}; \bs{\theta}) = - \partial \log f_z (y, \bs{x}; \bs{\theta})/\partial \theta$ being the derivative of the corresponding minus log-density function over $\theta$.  

Below we pay special attention to the maximum likelihood estimation under a working generalized linear model \eqref{eq:glm}.
In particular, we study properties of $\bs{\theta}_N$, 
the root of the population estimating equation in \eqref{eq:population_risk} with $\psi_z(y, \bs{x}; \bs{\theta}) = - \partial \log f_z (y, \bs{x}; \bs{\theta})/\partial \theta$ implied by the working model. 
By the density forms in \eqref{eq:glm}, 
under a generalized linear model \eqref{eq:glm} allowing for interaction, 
$\theta_N$
must satisfy the following equations: for $z=0,1$, 
\begin{align}\label{eq:est_eqn_cglm}
	\frac{1}{N} \sum_{i=1}^N \big\{ Y_i(z) - h_z(X_i; \theta) \big\}
	\begin{pmatrix}
		1\\
		\bs{X}_i
	\end{pmatrix}
	=
	\frac{1}{N} \sum_{i=1}^N \big\{ Y_i(z) - \dot{b}_z(\alpha_z+\beta_z^\top X_i) \big\}
	\begin{pmatrix}
		1\\
		\bs{X}_i
	\end{pmatrix}
	= \bs{0},
\end{align}
where $h_z(x;\theta)$ denotes the conditional mean of the potential outcome $Y(z)$ given covariates under model \eqref{eq:glm}.  
Even if we additionally impose the constraint $\bs{\beta}_1 = \bs{\beta}_0$ (i.e., no treatment-covariate interaction), 
the equation in the first row of \eqref{eq:est_eqn_cglm} still holds, which has important implications for the later discussion. 
We summarize the results below.  

\begin{corollary}\label{cor:glm_est_eqn}
	Under a working generalized linear model \eqref{eq:glm} {\lad with canonical links}, regardless of whether we allow treatment-covariate interaction, 
	under Conditions \ref{cond:population_minimum}--\ref{cond:empirical_minimum}, 
	the probability limit $\bs{\theta}_N$ of 
	the maximum likelihood estimator %
	satisfies: 
	\begin{align}\label{eq:ave_h_Y_bar}
		\bar{Y}(z) =
		\E_N\{h_z(X_i;\theta_N)\} = 
		\E_N \big\{ \dot{b}_z (\alpha_{Nz} + \bs{X}_i^\top \bs{\beta}_{Nz}) \big\}, 
    \qquad (z = 0,1)
	\end{align}
	where $\alpha_{N1}, \alpha_{N0}, \bs{\beta}_{N1}$ and $\bs{\beta}_{N0}$ are subvectors of $\bs{\theta}_N$. 
\end{corollary}

\section{Model-Based Approach}\label{sec:based}

From the discussion in \S \ref{sec:model} and as shown in \eqref{eq:diff_g_h_ind}, 
under a generalized linear model \eqref{eq:glm} with shared canonical links, 
$
g( h_1(\bs{X}_i; \bs{\theta}) ) - g( h_0(\bs{X}_i; \bs{\theta}))
$
is often interpreted as the treatment effect at the $g$-scale for unit $i$. 
Moreover, if there is no treatment-covariate interaction,
then it becomes constant across all units and is often referred as the constant treatment effect at the $g$-scale.  
These then motivate us to consider the following model-based estimator: 
\begin{equation}\label{eq:MB_est}
	\hat{\tau}_{g\MB}  \equiv 
	\E_N \big\{ g( h_1(\bs{X}_i; \hat{\bs{\theta}}_N) ) - g( h_0(\bs{X}_i; \hat{\bs{\theta}}_N)) \big\}. 
\end{equation}
Below we will investigate its property for estimating 
the average treatment effect $\tau_g$ at the $g$-scale introduced in \S \ref{sec:po}. 

In the following, 
we consider general functions $h_1, h_0$ and $g$, as well as general functions $\psi_1, \psi_0\in \mathbb{R}^p$ for obtaining $\hat{\theta}_N$ from the estimating equation in \eqref{eq:Z_hat}, 
for the model-based estimator $\hat{\tau}_{g\MB}$ in \eqref{eq:MB_est}, not necessarily those implied by a statistical model.  
To derive the asymptotic distribution of $\hat{\tau}_{g \MB}$, we assume some smoothness conditions on the functions $h_1, h_0$ and $g$, 
as well as 
boundedness 
of some finite population quantities. 
Define 
$G_{N\MB}(\bs{\theta}) = \E_N\{g(h_1(\bs{X}_i; \bs{\theta})) -g(h_0(\bs{X}_i; \bs{\theta}))\}$.

\begin{condition}\label{cond:gh_MB}
	(i) $G_{N\MB}(\bs{\theta})$ is continuously differentiable in $\bs{\theta} \in \Theta$ for all $N$; 
	(ii) there exists $\varepsilon_0>0$ such that $\mathcal{B}(\bs{\theta}_N, \varepsilon_0) \subset \Theta$ for all $N$;  
	(iii)
	for any $\varepsilon > 0$, there exists $\delta > 0$ such that, for all $N$ and $\bs{\theta}$ with $\|{\bs\theta}-{\bs\theta}_N\|\le \delta$, 
	$\|\dot G_{N\MB}(\bs{\theta}) - \dot G_{N\MB}(\bs{\theta}_N)\|\le \varepsilon$; 
	(iv)
	$\dot G_{N\MB}(\bs{\theta}_N)$ is bounded for all $N$.
\end{condition}

The following theorem shows the asymptotic distribution of the model-based estimator. 

\begin{theorem}\label{thm:clt_mbe}
	If $\sqrt{N} (\hat{\bs\theta}_N-{\bs\theta}_N ) \lpc \mathcal{N} ( \bs{0},\bs{\Sigma}_N )$ with $\sigma_{\max}(\bs{\Sigma}_N)$ being bounded for all $N$,  
	and Condition \ref{cond:gh_MB} holds, 
	then 
	\begin{align*}
		\sqrt{N}( \hat{\tau}_{g\MB}-\tau_{g\MB})
		= 
		\sqrt{N} \{ G_{N\MB}(\hat{\bs{\theta}}_N) - G_{N\MB}(\bs{\theta}_N) \}
		\lpc 
		\mathcal{N}
		\big( 
		0, \ 
		\dot G_{N\MB}(\bs{\theta}_N)^\top \bs{\Sigma}_N \dot G_{N\MB}(\bs{\theta}_N)
		\big),
	\end{align*}
	where 
	$
	\tau_{g\MB} 
	= 
	G_{N\MB}(\bs{\theta}_N)
	= \E_N\{ g( h_1(\bs{X}_i; \bs{\theta}_N) ) - g( h_0(\bs{X}_i; \bs{\theta}_N) )\}. 
	$
\end{theorem}

From Theorem \ref{thm:clt_mbe}, 
the model-based estimator is consistent for $\tau_{g \MB}$. 
However, $\tau_{g\MB}$ is generally different from the average effect $\tau_g$ at the $g$-scale. 
Below we discuss two special cases, where $h_1$ and $h_0$ are chosen to be the conditional mean functions of potential outcomes given covariates implied by the imposed working model. 

We first consider the case where the working model is a generalized linear model \eqref{eq:glm} {\lad with canonical link}.  
From Corollary \ref{cor:glm_est_eqn}, 
$
\E_N \{ h_{z} (\bs{X}_i; \bs{\theta}_N) \} = \bar{Y}(z)
$
for $z=0,1$.
Consequently, for a nonlinear $g$, $\E_N\{ g( h_z(\bs{X}_i; \bs{\theta}_N) ) \}$ is generally different from $g( \E_N\{  h_z(\bs{X}_i; \bs{\theta}_N) \}) = g(\bar{Y}(z))$, 
because a nonlinear $g$ is generally not exchangeable with the average operator.  
For example, 
when $g$ is convex or concave, by Jensen's inequality, one is often greater than the other. 
These indicate that generally
$\tau_{g \MB} \ne \tau_g$.

We then consider the case where 
we correctly specify the data generating model.  
In this case, 
we expect that 
$h_z(\bs{x}; \bs{\theta}_N) \approx \E(Y(z) \mid \bs{x})$, the true conditional mean of potential outcome given covariates, for $z=0,1$. 
By the same logic as before, 
for a nonlinear $g$, 
$
\E_N\{ g(h_z(\bs{X}_i;\bs{\theta}_N)) \}
\approx 
\E[ g\{ \E( Y(z)\mid \bs{X} ) \} ]
$
is generally different from 
$
g(\E_N\{ h_z(\bs{X}_i; \bs{\theta}_N) \})
\approx g( \E\{ \E( Y(z)\mid \bs{X} ) \} ) = g(\E(Y(z))). 
$
Thus, even with correct model specification, 
$\tau_{g \MB}$ is generally different from $\tau_g$. 
{\new The difference between $\tau_{g\MB}$ and $\tau_g$ in the case of correct model specification has also been noticed in the literature, where the former is often termed as conditional effects while the latter is termed as average or marginal effects; see, e.g., \citet{Rosenblum15} and \citet{Ding2022comment}. 
}

From the above, the model-based approach generally provides a biased estimation of average treatment effects 
and should 
be used with caution in practice.  
Note that for the generalized linear model \eqref{eq:glm} with 
shared canonical links
and without treatment-covariate interaction, 
the model-based estimator $\hat{\tau}_{g\MB}$ reduces to the maximum likelihood estimator for $\hat{\alpha}_1 - \hat{\alpha}_0$, which is often used in practice to access treatment effects. Therefore, Theorem \ref{thm:clt_mbe} reminds us that usual parameter estimation based on certain hypothesized statistical models may not be desirable for evaluating treatment effects, 
even when the models are correctly specified. 

\section{Model-Imputed Approach}\label{sec:impute}

In this section, we consider another way of utilizing models. 
Specifically, we use the model as a tool to impute the potential outcomes. 
We call this approach as model-imputed approach. 

Specifically, given a model with conditional mean function $h_z(\bs{x}; \bs{\theta})$ and a parameter estimator $\hat{\bs{\theta}}_N$, 
we can impute or estimate the potential outcome $Y_i(z)$ for unit $i$ by $h_z(\bs{X}_i; \hat{\bs{\theta}}_N)$. 
This then motivates the following model-imputed estimator for the average effect $\tau_g$ at the $g$-scale: 
\begin{align}\label{eq:tau_hat_MI}
	\hat{\tau}_{g \MI} & = 
	g \big( \E_N\{ h_1(\bs{X}_i; \hat{\bs{\theta}}_N ) \} \big) - g\big( \E_N \{ h_0(\bs{X}_i; \hat{\bs{\theta}}_N ) \} \big). 
\end{align}
The estimation in \eqref{eq:tau_hat_MI} is also called the model standardization approach \citep{Rosenblum15}. 
\citet{GuoBasse21} discussed a closely related approach that uses $h_z(\bs{X}_i; \hat{\bs{\theta}}_N)$'s only for the missing potential outcomes, and termed the it as Oaxaca-Blinder method \citep{Blinder73, Oaxaca73}. 
These two approaches are equivalent under the prediction unbiasedness condition in \citet{GuoBasse21}, but in general they are not the same. 

In the following discussion, we will allow $h_1$, $h_0$ and $g$ to be general functions,  
and assume some smoothness conditions on them, as well as boundedness of some finite population quantities. 
Define $G_{N\MI}(\bs{\theta})=g(\E_N \{ h_1(\bs X_i;\bs {\theta}) \} )-g(\E_N \{ h_0(\bs X_i;\bs {\theta}) \} )$. 
\begin{condition}\label{cond:MI_g_h}
    Condition \ref{cond:gh_MB} holds with $G_{N\MB}(\cdot)$ replaced by $G_{N\MI}(\cdot)$. 
\end{condition}

\begin{theorem}\label{thm:clt_ibe}
	If $\sqrt{N} (\hat{\bs\theta}_N-{\bs\theta}_N) \lpc \mathcal{N}( \bs{0},\bs{\Sigma}_N)$ with $\sigma_{\max}(\bs{\Sigma}_N)$ being bounded for all $N$, 
    and Condition 
	\ref{cond:MI_g_h} holds, 
	then
	\begin{align*}
		\sqrt{N}(\hat{\tau}_{g\MI}-\tau_{g\MI}) 
		= 
		\sqrt{N}( G_{N\MI}(\hat{\bs{\theta}}_N) - G_{N\MI}(\bs{\theta}_N) ) 
		\lpc  
		\mathcal{N}\big( 
		0, \ 
		\dot G_{N\MI}(\bs{\theta}_N)^\top 
		\bs{\Sigma}_N \dot G_{N\MI}(\bs{\theta}_N)
		\big),
	\end{align*}
	where 
	$
	\tau_{g\MI} 
	= 
	g( \E_N\{ h_1(\bs{X}_i; \bs{\theta}_N) \} ) - g( \E_N\{ h_0(\bs{X}_i; \bs{\theta}_N )\} ).
	$
\end{theorem}

From 
Corollary \ref{cor:glm_est_eqn}, 
under the generalized linear model \eqref{eq:glm} {\lad with canonical link},  
if we use the maximum likelihood approach to estimate the model parameters, 
then $\tau_{g\MI}$ is actually the same as the average treatment effect $\tau_g$ at the $g$-scale, implying the consistency of the model-imputed estimator $\hat{\tau}_{g\MI}$. 
Below we generalize this to general model-imputed estimation. 
We introduce the following condition on the estimating functions $\psi_z$s in the estimating equation \eqref{eq:Z_hat} for $\theta$ and the imputation functions $h_z$s to ensure  \eqref{eq:ave_h_Y_bar}, i.e., 
the average imputed potential outcomes are the same as the true 
ones, 
which is sufficient and almost necessary for the consistency of the model-imputed estimator.

\begin{condition}\label{cond:consis_MI}
	There exist differentiable 
	vector functions $\bs{q}_1(\bs{\theta})$ and $\bs{q}_0(\bs{\theta})$ such that 
	for all $(\bs{x}, y)\in \mathcal{X} \times \mathcal{Y}$, $\bs{\theta}\in \Theta$ and $z=0,1$, 
	$\bs{q}_z(\bs{\theta})^\top  \psi_z(y, \bs{x}; \bs{\theta}) = y - h_z(\bs{x}; \bs{\theta})$ 
	and 
	$\bs{q}_z(\bs{\theta})^\top  \psi_{1-z}(y, \bs{x}; \bs{\theta}) = 0$. 
\end{condition}

\begin{corollary}\label{cor:model_imp_consist}
	Under Conditions \ref{cond:population_minimum}--\ref{cond:empirical_minimum}, \ref{cond:MI_g_h} and \ref{cond:consis_MI}, 
    $\hat{\tau}_{g\MI}$ in \eqref{eq:tau_hat_MI} is consistent for the average treatment effect $\tau_g$ at the $g$-scale, 
	and 
	$
	\sqrt{N}(\hat{\tau}_{g\MI}-\tau_{g}) \lpc \mathcal{N}(0, V_{g\MI}), 
	$
	where 
	\begin{align}\label{eq:V_g_MI}
		V_{g\MI} & = 
		r_1 r_0 \dot G_{N\MI}(\bs{\theta}_N) ^\top 
		\{\dot \Psi_N(\bs{\theta_N})\}^{-1} \Cov_N\{ \psi_{1i}(\bs{\theta}_N) - \psi_{0i}(\bs{\theta}_N) \}  \{\dot \Psi_N(\bs{\theta_N})^\top \}^{-1}  \dot G_{N\MI}(\bs{\theta}_N) . 
	\end{align}
\end{corollary}

From Corollary \ref{cor:model_imp_consist}, as long as Condition \ref{cond:consis_MI} holds, the model-imputed estimator is consistent for $\tau_g$ and is robust to arbitrary model misspecification. 
Therefore, in practice, it is often desirable to design estimating and imputation functions such that Condition \ref{cond:consis_MI} holds. 
As verified in
\S \ref{sec:glm_cond}, 
if we choose the estimating functions $\psi_z$s to be the derivative of the minus log-density functions of some generalized linear models 
in \eqref{eq:glm} {\lad with canonical links}, then Condition \ref{cond:consis_MI} holds automatically. 
Therefore, generalized linear models {\lad with canonical links}
can be
preferable for imputing the potential outcomes and should be recommended in practice. 
{\new Moreover, we can further use Theorem \ref{thm:m_est_variance_estimate} to construct 
large-sample confidence intervals for the average treatment effects.}

As a side note, 
if the units are i.i.d.\ samples from a superpopulation and we correctly specify the data generating model, 
then we expect that $
N^{-1} \sum_{i=1}^N  h_z(\bs{X}_i; \bs{\theta}_N)
\approx 
N^{-1} \sum_{i=1}^N  \E(Y(z) \mid \bs{X}_i )
\approx 
\E\{ \E(Y(z) \mid \bs{X} ) \}
= 
\E(Y(z))
\approx \bar{Y}(z), 
$
implying that $\tau_{g\MI} \approx \tau_g$. 
This indicates the consistency of model-imputed estimator under correct model specification. 
Nevertheless,
it is still preferable to 
choose estimating and imputation functions such that Condition \ref{cond:consis_MI} holds, 
which makes the model-imputed estimation robust to model misspecification. 

{\new Under Condition \ref{cond:consis_MI}, we can show that $\E_N^z \{h_z(X_i; \hat{\theta}_N)\} = \E_N^z \{Y_i\}$ for $z=0,1$, a property termed as prediction unbiasedness by \citet{GuoBasse21}. 
Therefore, $\hat{\tau}_{g\MI}$ with $g=\id$ is also a generalized Oaxaca-Blinder estimator studied in \citet{GuoBasse21}. 
A technical difference is that 
the asymptotic Normality in Corollary \ref{cor:model_imp_consist} is driven mainly by that of the Z-estimator $\hat{\theta}_N$; 
see \S \ref{sec:MA_consist} and \S \ref{sec:conn_mi_ma} for more related discussion.
}

\section{Connection between model-imputed and model-assisted approach}\label{sec:conn_mi_ma}

From \S \ref{sec:impute}, the model-imputed estimator is consistent for the average treatment effect $\tau_g$ and robust to model misspecification 
as long as the parameter estimation satisfies Condition \ref{cor:model_imp_consist}, as shown in Corollary \ref{cor:model_imp_consist}. 
From \S \ref{sec:MA_consist}, 
the model-assisted estimator is always consistent for the average treatment effect $\tau_g$ and robust to model misspecification, as shown in Theorem \ref{thm:clt_mae}. 
These two estimators are indeed closely related, as demonstrated in the corollary below.

\begin{corollary}\label{cor:equi_i_a}
	Under the conditions in Corollary \ref{cor:model_imp_consist} and Theorem \ref{thm:clt_mae}, 
	the model-imputed and model-assisted estimators are equivalent, i.e.,  $\hat{\tau}_{g\MI} = \hat{\tau}_{g\MA}$, 
	and their asymptotic variances in \eqref{eq:V_g_MI} and \eqref{eq:V_g_MA}, although having quite different forms, are exactly the same, i.e., $V_{g\MI} = V_{g\MA}$.  
\end{corollary}

From Corollary \ref{cor:equi_i_a}, any model-imputed estimator with parameter estimation satisfying Condition \ref{cond:consis_MI} can be equivalently viewed as a model-assisted estimator. 
On the contrary, we can also view the model-assisted estimator as a model-imputed estimator, as explained below. 
Consider any model-imputed estimator with adjustment functions $h_z$s. We assume that $h_z$s do not have intercept terms; this does not lose any generality since the model-assisted estimator is invariant to any constant shift of the adjustment function. 
We then consider the model-imputed estimator using imputation function $\tilde{h}_z$s of the following form:  
$
\tilde{h}_z(x; \tilde{\theta}) = \alpha_z + h_z(x; \theta)
$
for $z=0,1$, 
which augments the adjustment function with an intercept term. Here, $\tilde{\theta} = (\alpha_1, \alpha_0, \theta^\top)^\top$. 
We do not impose any constraint on the estimation of $\theta$, which can be estimated in the same way as in the model-assisted estimator. 
For the intercepts $\alpha_1$ and $\alpha_0$, we assume that they are estimated based on the following estimating functions: 
\begin{align}\label{eq:psi_ma_mi}
    \tilde{\psi}_1(y, x; \tilde{\theta}) = 
    \begin{pmatrix}
        y - \alpha_1 - h_1(x; \theta)\\
        0
    \end{pmatrix}, 
    \quad 
    \tilde{\psi}_0(y, x; \tilde{\theta}) = 
    \begin{pmatrix}
        0\\
        y - \alpha_0 - h_0(x; \theta)
    \end{pmatrix},
\end{align}
noting that they can be combined with any estimating functions for $\theta$ to perform Z-estimation for the whole parameter vector $\tilde{\theta}$. 
As verified in \S \ref{sec:proof_equ_mi_ma}, the resulting model-imputed estimator will be equivalent to the original model-assisted estimator.

Despite the above equivalence relationship between the model-imputed and model-assisted estimators, 
we find it more preferable to distinguish them. 
We emphasize that the general model-imputed estimator is inconsistent for the average treatment effect and is not robust to model misspecification, unless the parameters are estimated in specific ways.

\section{Linear adjustment based on imputed potential outcomes or transformed covariates}\label{sec:ai}

In practice, we can hypothesize several plausible models for the relation between potential outcomes and covariates, 
and the resulting estimated models (often based on the maximum likelihood approach) and their implied conditional means of potential outcomes given covariates can then help us impute the potential outcomes. 
As discussed in \S \ref{sec:impute}, such an approach may provide inconsistent treatment effect estimation if the corresponding model assumption fails or the imputation and estimating functions do not satisfy Condition \ref{cond:consis_MI}. 
However, we can still view these imputed potential outcomes as transformed covariates, and further use model-assisted approach to adjust these covariates. 
When some hypothesized models 
can explain the variability of potential outcomes well, such an approach can be more efficient than that adjusting only original covariates. 
Below we consider linear adjustment for these imputed potential outcomes or essentially transformed covariates.
Such a strategy has also been studied by \citet{GuoBasse21} and 
\citet{Fogarty2020calibration}. 
{\new Specifically, \citet{GuoBasse21} considered adjusting each potential outcome using the corresponding imputation, and \citet{Fogarty2020calibration} considered adjusting additionally the imputation for the other potential outcome of the same unit.  
We extend their discussion by considering multiple parametric models for imputation, 
and provide general and concrete regularity conditions. 
Relatedly, \citet{gagnon2021precise} considered adjustments using imputations fitted from auxiliary observational data.}

Let $J$ be the number of hypothesized parametric models under consideration. 
For each hypothesized model $1\le j\le J$, 
we use $h_{jz}^{\MI}(x; \theta^{\MI}_j)$ ($z=0,1$) to denote the form of the imputation function. 
We 
estimate the parameter $\theta^{\MI}_j$ using a Z-estimator  $\hat{\theta}^{\MI}_{Nj}$ from 
the estimating equation as in \eqref{eq:Z_hat}
with some estimating function $\psi^{\MI}_{jz}$ for $z=0,1$, which is often the derivative of the minus log-density function implied by the corresponding model.  
We then consider the following linear adjustment function $h_{z}^{\MA}(x; \theta^{\MA})$:
\begin{align}\label{eq:h_z_MA}
    h_{z}^{\MA}(x; \theta^{\MA})
    & = 
    \alpha_z
    +
    \sum_{j=1}^J \beta_{zj} h_{j0}^{\MI}(x; \hat{\theta}^{\MI}_{Nj}) + \sum_{j=1}^J \gamma_{zj} h_{j1}^{\MI}(x; \hat{\theta}^{\MI}_{Nj}), 
    \quad 
    (z=0,1)
\end{align}
where the parameter vector
$\theta^{\MA}$ consists of $\alpha_z, \beta_{zj}$ and $\gamma_{zj}$ for $z=0,1$ and $1\le j \le J$. 
We will then use $\hat{\theta}^{\MA}_N$ that minimizes the empirical risk function as in \eqref{eq:M_hat} with squared losses, due to the optimality discussed in \S \ref{sec:optimal_loss_MA}, 
and denote the resulting model-assisted estimator by $\hat{\tau}_{g\MA\MI}$, with the subscript indicating model assisting based on imputed potential outcomes. 

Let $\theta^{\MI}_{Nj}$ be the minimizer of the finite population estimating equation as in \eqref{eq:population_risk} with estimating functions $\psi^{\MI}_{j1}$ and $\psi^{\MI}_{j0}$, for $1\le j \le J$. 
Let $\tilde{\tau}_{g\MA\MI}$ be the model-assisted estimator with adjustment function $\tilde{h}_{z}^{\MA}(x; \theta^{\MA})$, which has the same form as that in \eqref{eq:h_z_MA} but with fixed transformed covariates $h_{j0}^{\MI}(x; \theta^{\MI}_{Nj})$ and $h_{j1}^{\MI}(x; \theta^{\MI}_{Nj})$ for $1\le j \le J$, 
and squared loss for estimating $\theta^{\MA}$. 
As demonstrated below, under some regularity conditions,  
$\hat{\tau}_{g\MA\MI}$ is asymptotically equivalent to $\tilde{\tau}_{g\MA\MI}$, which, by Theorem \ref{thm:clt_mae}, is consistent for $\tau_g$ and asymptotically Normal. 

\begin{condition}\label{cond:impute_adjust}
    (a) Conditions \ref{cond:ma_continuous} and \ref{cond:ma_stable} hold for $(h_1, h_2, \theta_N) = (h_{j1}^\MI, h_{j0}^\MI, \theta_{Nj}^\MI)$, $1\le j \le J$. 
    (b)
    the finite population variances of $Y_i(z)$, $h_{j1}^\MI(X_i; \theta_{Nj}^\MI)$ and $h_{j0}^\MI(X_i; \theta_{Nj}^\MI)$ $(z=0,1; 1 \le j \le J)$ 
    are bounded for all $N$. 
    (c)
    As $N\rightarrow \infty$, 
    $N^{-1} \max_{1\le l \le N} \{Y_l(z) - \bar{Y}(z)\}^2$ and 
    $N^{-1} \max_{1\le l \le N} [h_{jz}^\MI(X_l;\theta_{Nj}^\MI) - \E_N\{h_{jz}^\MI(X_i;\theta_{Nj}^\MI)\}]^2$ converge to zero for all $z$ and $j$.  
\end{condition}

\begin{theorem}\label{thm:clt_ai}
    Under Condition \ref{cond:impute_adjust}, 
    $\hat{\tau}_{g\MA\MI} - \tilde{\tau}_{g\MA\MI} = o_{\Pr}(N^{-1/2})$, and they follow the same asymptotic Normal distribution. 
    Moreover, the asymptotic variance can be conservatively estimated analogously as in Theorem \ref{thm:MA_ci}. 
\end{theorem}

For conciseness, 
we relegate the expressions for the asymptotic variance and conservative variance estimator to \S \ref{app:ai}. 

\begin{remark}
We can also consider adjustment functions that have the same form as in \eqref{eq:h_z_MA} but include both parameters $\theta^A$ for linear adjustment coefficients and $\theta^{\MI}_j$'s for transforming the covariates, and estimate these parameters simultaneously using, say, the squared loss functions. 
This can lead to model-assisted estimator with smaller estimated standard error due to the optimality discussed in \S \ref{sec:optimal_loss_MA}. 
However, the two-step estimation as in \eqref{eq:h_z_MA} can sometimes be computationally easier and more stable, since the first-step estimation for $\hat{\theta}^{\MI}_j$'s can be classical maximum likelihood estimation,
and the second-step estimation is simply a least squares fitting. 
For example, for many generalized linear models, the minus log-likelihood is a convex function \citep{Wedderburn76}, which eases the computation for the maximum likelihood estimator. 
\end{remark}

\begin{remark}\label{rmk:tran_ind_effect}
    The idea of using transformed covariates can also be applied to \S \ref{sec:ite}. In particular, when modeling the relation between individual treatment effects and covariates, we can also use transformed covariates obtained from some fitted models instead of original ones. 
\end{remark}

\section{Understanding individual effects from working models on potential outcomes}\label{sec:ind_po}

We now discuss how to understand individual treatment effects by working models on the relation between potential outcomes and covariates, which is briefly mentioned in \S \ref{sec:ite}. 
Suppose $h_z(X_i; \theta)$ is the hypothesized functional form for estimating the individual potential outcome $Y_i(z)$ for $z=0,1$, and let the ``oracle'' parameter $\theta_N$ be the root of some finite population estimating equation in \eqref{eq:population_risk} with some estimating functions $\psi_z$s, or the minimizer of the expected value of the risk function in \eqref{eq:M_hat} with some loss functions $\loss_z$s.  
We can then decompose the potential outcome $Y_i(z)$ into two parts: 
the mean part $h_z(\bs{X}_i; \bs{\theta}_N)$ and the corresponding residual $e_i(z) \equiv Y_i(z) - h_z(\bs{X}_i; \bs{\theta}_N)$. 
The decomposition of potential outcomes naturally provide the following decomposition of individual treatment effects: for $1\le i \le N$, 
\begin{align}\label{eq:tau_i}
	\tau_i \equiv Y_i(1) - Y_i(0) = \{ h_1(\bs{X}_i; \bs{\theta}_N) - h_0(\bs{X}_i; \bs{\theta}_N)\} + \{e_i(1) - e_i(0)\} \equiv \mu_i + \varepsilon_i, 
\end{align}
where $\mu_i$ and $\varepsilon_i$ denote the mean and residual of the individual treatment effect. 
When $h_z$s and $\loss_z$s correspond to the conditional mean functions and minus log-density functions of correctly specified models for both potential outcomes, then $\mu_i \approx \E( Y(1) - Y(0) \mid \bs{X} = \bs{X}_i)$, the \textit{true} conditional average treatment effect given covariate value $\bs{X}_i$. 
However, in the presence of model misspecification, especially under our finite population inference framework, 
the interpretation of the decomposition 
in \eqref{eq:tau_i} becomes difficult. 
In general, we can understand $h_z(X_i, \theta_N)$s as the ``best'' approximations of potential outcomes $Y_i(z)$s in terms of some prespecified loss functions $\loss_z$s, 
e.g., the squared loss
$\loss_z(Y_i(z), X_i;\theta) = \{Y_i(z) - h(X_i;\theta)\}^2$, 
and can therefore intuitively understand $\mu_i$'s as good approximations for the individual treatment effects $\tau_i$'s.  
Theorems \ref{thm:theta_clt} and \ref{thm:m_est_variance_estimate} can then help us conduct statistical inference for $\mu_i$s in \eqref{eq:tau_i}.  

Compared to the above strategy, the approach in \S \ref{sec:ite} can be more preferable. 
In \S \ref{sec:ite}, we suggest directly imposing models on the relation between individual treatment effects and potential outcomes, which is our target of interest. 

\begin{remark}
The fitted models for both potential outcomes could potentially improve the estimation for individual treatment effects. 
Following the idea in \citet{Semenova23}, for all $i$,  we can estimate $\tau_i$ by 
\begin{align}\label{eq:tau_i_imp}
   \hat{\tau}_i =\{h_1(X_i;\hat{\theta}_N)-h_0(X_i;\hat{\theta}_N)\}+Z_i\{Y_i -h_1(X_i;\hat{\theta}_N)\}/r_1 - (1-Z_i)\{Y_i-h_0 (X_i;\hat{\theta}_N)\}/r_0, 
\end{align}
where $h_1(X_i;\hat{\theta}_N)$ and $h_0(X_i;\hat{\theta}_N)$ can be fitted conditional means of potential outcomes from some working models. 
When $\hat{\theta}_N$ is predetermined and fixed, 
$\hat{\tau}_i$ in \eqref{eq:tau_i_imp} is exactly unbiased for $\tau_i$ 
and can be used to construct the estimating equation in \eqref{eq:Psi_ind_effect_est} as well;  
the corresponding Z-estimator from \eqref{eq:Psi_ind_effect_est} can then be analyzed using the developed Z-estimation theory in \S \ref{sec:m_est}. 
However, when $\hat{\theta}_N$ is estimated using the observed data,  $\hat{\tau}_i$ in \eqref{eq:tau_i_imp} may not be exactly unbiased for $\tau_i$. 
This will lead to some technical difficulty when we use \eqref{eq:tau_i_imp} to construct the estimating equation in \eqref{eq:Psi_ind_effect_est}. 
It will also be interesting to investigate the optimal choice of the $h_z$s. 
These are beyond the scope of this paper, and we leave them for future study.
\end{remark}

\section{Simulation studies}\label{sec:simu}

\subsection{Heterogeneous treatment effects}\label{sec:simu_heter}

We conduct a simulation to illustrate the performance of various approaches for estimating average treatment effects. 
We simulate potential outcomes and covariates as i.i.d.\ draws from the following model with $n =1000$: 
\begin{align}\label{eq:po_gen}
    & Y(0) =  6 + \exp(1+X_1+X_2) + 3 \varepsilon_0, \ \ 
    Y(1) = 11 + \exp(3-X_1^2-X_2) + 3 \varepsilon_1, 
\end{align}
where $X_1, X_2, \varepsilon_0$ and $\varepsilon_1$ are mutually independent and follow the standard Normal distribution. 
We further round the potential outcomes to the nearest nonnegative integers, making them count outcomes. 
Once generated, we fix the potential outcomes and covariates, and generate $10^4$ random treatment assignments from the CRE with half of the units receiving treatment and control, respectively, mimicking the finite population inference. 

We focus on estimating the average treatment effect at the $\log$-scale, i.e., $\tau_{\log}$. 
We consider Poisson, negative binomial and linear regression models, where the former two are the most popular models for analyzing count data, 
and consider cases with and without treatment-covariate interaction. 
We conduct simulation for model-based, model-imputed and model-assisted estimators based on these regression models as discussed in \S \ref{sec:ave_effect}. 
Here we do not consider negative binomial regression model without treatment-covariate interaction since it often leads to numerical convergence issue for maximizing the likelihood function, as well as the model-based approach for linear models since it will involve logarithm of negative values.  
The numerical results for these estimators are shown in the first six rows and last three rows of Table \ref{tab:simu}, where we also include the unadjusted estimator without using any covariate as a benchmark. 

We consider then the optimal model-assisted estimator discussed in \S \ref{sec:optimal_loss_MA} with squared loss functions. 
The form of the adjustment function is based on the mean specification from either the Poisson or Negative binomial regression models; noting that both models give the same mean specification. 
The numerical results for this estimator are shown in the seventh row of Table \ref{tab:simu}. 

We finally consider the model-assisted estimator linearly adjusting the imputed potential outcomes as discussed in \S \ref{sec:ai}. 
We consider three ways for imputing the potential outcomes: the first two use the fitted values from the Poisson or Negative binomial regression models, and the third one uses the Poisson mean specification with squared losses for estimating the parameter. 
In other words, all three ways use the same form of imputation function, but with different estimation for the parameter: the first two uses maximum likelihood estimator from either the Poisson or Negative binomial regression models, and the third one minimizes the squared loss as discussed in \S \ref{sec:optimal_loss_MA}. 
The numerical results for this estimator are shown in the eighth to tenth rows of Table \ref{tab:simu}.

\begin{table}
    \centering
    \caption{The performance of various approaches using different statistical models for estimating the average treatment effect at the $\log$-scale. The data are generated from \eqref{eq:po_gen}. 
    ``\MB'' denotes the model-based approach in \S \ref{sec:based}, 
    ``\MI'' denotes the model-imputed approach in \S \ref{sec:impute}, 
    and 
    ``\MA'' denotes the model-assisted approach in \S \ref{sec:MA_consist}, where the unknown parameters are estimated by maximizing the model-implied likelihood function, unless otherwise stated.
    ``\MA (squared loss)'' denotes the model-assisted approach with squared loss as in \S \ref{sec:optimal_loss_MA}, where we use the Poisson mean specification with treatment-covariate interaction for the form of the adjustment functions. 
    ``\MA\MI'' denotes the model-assisted approach in \S \ref{sec:ai} linearly adjusting the imputed potential outcomes from fitted Poisson regression, fitted Negative binomial regression, and the Poisson mean specification with squared losses for estimating the parameter as in \S \ref{sec:optimal_loss_MA}. 
    The ``\textup{Unadjusted}'' denotes the model-assisted approach without using any covariate adjustment. 
    The equality sign in Column 3 indicates that several approaches provide the same estimation. 
    Columns 4--7 show the bias, standard deviation, root mean squared error and estimated standard error, all scaled by $\sqrt{n}$, for each of the estimators. 
    }
    \label{tab:simu}
    \begin{tabular}{cclrrrrr}
    \toprule
    Model & Interaction & Estimation & Bias & SD & RMSE & ESE & Coverage \\
    \midrule
    Pois & No & B$=$I$=$A & $0.080$ & $2.246$ & $2.248$ & $3.261$ & $0.9708$
    \\
    \midrule
    Pois & Yes & B & $7.691$ & $1.225$ & $7.788$ & $1.560$ & $0.0000$
    \\
    & & I$=$A & $0.255$ & $1.336$ & $1.360$ & $1.591$ & $0.9435$
    \\
    \midrule
    NBin & Yes & B & $4.523$ & $0.644$ & $4.569$ & $0.928$ & $0.0000$ 
    \\
     & & I & $2.238$ & $0.987$ & $2.447$ & $1.322$ & $0.6420$
    \\
    & & A & $0.093$ & $1.814$ & $1.817$ & $2.390$ & $0.9477$
    \\
    \midrule
    Pois & Yes & A (squared loss) & $0.121$ & $0.785$ & $0.794$ & $0.950$ & $0.9640$
    \\
    mean & & AI (Pois) & $0.322$ & $1.164$ & $1.207$ & $1.352$ & $0.9497$
    \\
    & & AI (Nbin) & $0.227$ & $1.472$ & $1.490$ & $1.828$ & $0.9479$
    \\
    & & AI (square loss) & $0.108$ & $0.778$ & $0.786$ & $0.945$ & $0.9651$
    \\
    \midrule
    Linear & No & I$=$A & $0.090$ & $2.282$ & $2.284$ & $3.324$ & $0.9780$
    \\
    \midrule
    Linear & Yes & I$=$A & $0.104$ & $2.197$ & $2.199$ & $2.817$ & $0.9502$
    \\
    \midrule
    Unadjusted & & & $0.053$ & $2.308$ & $2.308$ & $3.305$ & $0.9728$
    \\
    \bottomrule
    \end{tabular}
\end{table}

We now compare the above estimators based on their simulation results shown in Table \ref{tab:simu}. 
The model-based approach generally provides biased estimation for the average treatment effect, except for the Poisson regression without treatment-covariate interaction, under which the model-based approach leads to the same estimation as the model-imputed and model-assisted approaches; note that such equivalence generally fails when we consider average effects at other scales.  
In addition, the model-imputed approach using negative binomial regression also provides biased estimation,
since the corresponding likelihood equations do not satisfy Condition \ref{cond:consis_MI}. 
Interestingly, for the negative binomial regression,
the model-imputed approach has smaller bias and mean squared error than the model-based approach.

Compare the four model-assisted estimators using the same adjustment function but different ways for estimating the parameter, shown in the first, third, sixth and seventh row of Table \ref{tab:simu}. 
Among them, the model-assisted estimator with squared loss has the smallest root mean squared error and the smallest estimated standard error, confirm its optimality established in \S \ref{sec:optimal_loss_MA}. 
Moreover, 
except the one uses Poisson regression without treatment-covariate interaction, 
these estimators are more precise than those based on linear adjustment functions. This indicates the importance of the form of adjustment functions.

We now discuss the three model-assisted estimators linearly adjusting the imputed potential outcomes, shown in the eighth to tenth rows of Table \ref{tab:simu}. 
All of them improve their original model-assisted estimators in the third, sixth and seventh rows. This is not surprising because (i) we enrich the form of adjustments, and (ii) we estimate the linear adjustment coefficients in the second step using the squared loss that is optimal as studied in \S \ref{sec:optimal_loss_MA}. 
Under our simulation setting, the improvement is quite small when we use squared loss to estimate the parameter in the imputation function in the first step, 
whereas it is most substantial we use the maximum likelihood estimator from the Negative binomial model.  

The coverage probabilities of confidence intervals for the biased estimators are either zero or around $64\%$, which is substantially lower than the nominal level $95\%$. This is not surprising since the biases of these estimator are not negligible compared to their standard deviations.
Among the three inconsistent estimators, the model-imputed estimator based on the negative binomial regression model has the smallest bias, which results in the highest coverage probability of about $64\%$ among them.  
The coverage probabilities of confidence intervals for the consistent estimators are close or above the nominal level, due to the conservativeness in variance estimation.

Finally, we try to explain individual effect heterogeneity using covariates. 
We consider the best linear approximation of the individual effects using either the original covariates or the imputed potential outcomes from the Poisson regression model with treatment-covariate interaction as discussed in Remark \ref{rmk:tran_ind_effect}. 
We use $R^2 = 1 - \sum_{i=1}^n(\tau_i - u_i)^2/\sum_{i=1}^n \tau_i^2$ to denote the proportion of treatment effect variation explained by the covariates, 
where $u_i$'s denote the estimated best approximation. 
Across all simulations, the average values of $R^2$ are, respectively, $0.50$ and $0.71$, when using the original and transformed covariates. 
These indicate that 
models
can provide useful transformations of original covariates and help better explain treatment effect heterogeneity.

\begin{table}
    \centering
    \caption{ 
    The description of the table is the same as Table \ref{tab:simu}, except that the data are generated from \eqref{eq:po_gen_zero}, under which the treatment has no effect for any unit. 
    }
    \label{tab:simu_null}
    \begin{tabular}{cclrrrrr}
    \toprule
    Model & Interaction & Estimation & Bias & SD & RMSE & ESE & Coverage \\
    \midrule
    Pois & No & B$=$I$=$A & $-0.008$ & $0.645$ & $0.645$ & $0.638$ & $0.9449$
    \\
    \midrule
    Pois & Yes & B & $-0.006$ & $0.584$ & $0.584$ & $0.581$ & $0.9463$
    \\
    & & I$=$A & $-0.008$ & $0.647$ & $0.647$ & $0.633$ & $0.9425$
    \\
    \midrule
    NBin & Yes & B & $-0.006$ & $0.576$ & $0.576$ & $0.572$ & $0.9460$
    \\
     & & I & $-0.008$ & $0.630$ & $0.630$ & $0.615$ & $0.9397$
    \\
    & & A & $-0.008$ & $0.661$ & $0.661$ & $0.649$ & $0.9433$
    \\
    \midrule
    Pois/Nbin & Yes & A (squared loss) & $0.004$ & $0.413$ & $0.413$ & $0.413$ & $0.9500$
    \\
    mean & & AI (Pois) & $-0.007$ & $0.580$ & $0.580$ & $0.580$ & $0.9460$
    \\
    specification & & AI (Nbin) & $-0.007$ & $0.593$ & $0.593$ & $0.590$ & $0.9443$
    \\
    & & AI (square loss) & $-0.001$ & $0.410$ & $0.410$ & $0.409$ & $0.9495$
    \\
    \midrule
    Linear & No & I$=$A & $-0.011$ & $0.788$ & $0.788$ & $0.780$ & $0.9501$
    \\
    \midrule
    Linear & Yes & I$=$A & $-0.011$ & $0.789$ & $0.789$ & $0.778$ & $0.9487$
    \\
    \midrule
    Unadjusted & & & $-0.007$ & $1.093$ & $1.093$ & $1.086$ & $0.9479$
    \\
    \bottomrule
    \end{tabular}
\end{table}

\subsection{Constant treatment effects}

We conduct an additional simulation in which the treatment has no effect for every unit. 
This is the setting where our variance estimation becomes asymptotically exact, and can help better assess the coverage properties of the confidence intervals. 
Specifically, we generate the potential outcomes in the following way: 
\begin{align}\label{eq:po_gen_zero}
    Y(1) = Y(0) = 10 + \exp(1+X_{1}-0.5 \cdot X_{2}) + 3 \varepsilon, 
\end{align}
where $X_1, X_2$ and $\varepsilon$ are mutually independent and follow the standard Normal distribution. 
We round the potential outcomes to the nearest nonnegative integers, making them count outcomes.
All the potential outcomes and covariates will be fixed once generated.

We then conduct the same simulation as in \S \ref{sec:simu_heter}. The results are analogously shown in Table \ref{tab:simu_null}. Note that in this case, all the estimators are consistent for the zero average treatment effects. 
Therefore, the biases are all close to zero, and the coverage probabilities are all close to the nominal level $95\%$. 
The other results are similar to \S \ref{sec:simu_heter}, and thus we omit the detailed discussion. 

\section{Uniform convergence and consistency for M-estimation and Z-estimation under simple random sampling}\label{sec:M_srs}

In this section, we consider M-estimation and Z-estimation under simple random sampling. 
We establish the uniform convergence of the empirical risk function and the estimating equation, and the consistency of the M-estimator and Z-estimator. 
We will consider two types of sufficient conditions, based on either the compactness of the parameter space or the convexity of the risk function. 
We present the results in the following four subsections. With in each subsection, we first present the two theorems for uniform convergence and consistency, and then present their proofs.

\subsection{M-estimation with compactness}

\begin{theorem}\label{thm:uniform_converge}
        Let $m_{i}(\theta) \in \mathbb{R}$, $1\le i \le N$, and $\tilde{m}(\bs\theta) \in \mathbb{R}$ be functions of $\theta$ in a set $\Theta\subset \mathbb{R}^p$, 
        where these functions can also depend on $N$ but we make this dependence implicit for notational convenience. 
		Let 
		$(Z_{1}, Z_{2}, \ldots, Z_{N}) \in \{0,1\}^N$ be a random vector 
		whose probability of taking value $(z_1, z_2, \ldots, z_N) \in \{0,1\}^N$ is $1/\binom{N}{n}$ if $\sum_{i=1}^N z_i = n$ and zero otherwise, 
		where $n$ is a fixed constant. 
		Define 
		\begin{align*}
		    \hat{M}_N(\bs\theta)= \frac{1}{n}\sum_{i=1}^N Z_{i} m_{i}(\bs\theta)+\tilde{m}(\bs\theta), 
		    \quad 
		    M_N(\bs\theta)=\frac{1}{N}\sum_{i=1}^N m_{i}(\bs\theta)+\tilde{m}(\bs\theta),
		\end{align*}
		and 
		$
		\Delta(m_{i}, \delta)=\sup _{(\bs\theta, \bs\theta')\in \Theta^2:\|\bs\theta-\bs\theta'\|\leq\delta}|
		m_{i}(\bs\theta)-m_{i}(\bs\theta' )|.
		$
		If the following conditions hold:
		\begin{itemize}
                \item[(i)] $\Theta$ is compact,
			\item[(ii)] $\sup_N \E_N \{ \Delta(m_{i},\delta) \} \rightarrow 0$ as $\delta \rightarrow 0$, 
			\item[(iii)]  for any $\bs\theta\in \Theta$, $\Var_N\{m_{i}(\bs\theta)\}$ is $o(n)$, 
		\end{itemize}   
		then $\sup_{\bs\theta\in\Theta}|\hat{M}_N(\bs\theta)-M_N(\bs\theta)| = o_{\Pr}(1)$ as $N\rightarrow \infty$. 
\end{theorem}

\begin{theorem}\label{thm:consistency}
	Under the same setting as in Theorem \ref{thm:uniform_converge}, if conditions (i), (ii) and (iii) in Theorem \ref{thm:uniform_converge} holds, 
	and the following conditions hold:
	\begin{itemize}
		\item[(i)] 
        $\bs{\theta}_N$ is the unique minimizer of $M_N(\bs{\theta})$  satisfying that for any $\varepsilon>0$, there exists $\eta>0$ such that
		$\inf_{\bs\theta\in\Theta: \|\bs\theta-\bs\theta_{N}\| \geq \varepsilon} M_N(\bs\theta)>M_N\left(\bs\theta_{N}\right)+\eta$ for all $N$,
		\item[(ii)] $\hat{\bs{\theta}}_N \in \Theta$ 
		satisfies $\hat{M}_N(\hat{\bs\theta}_N)\leq \hat{M}_N(\bs{\theta}_N)$, 
	\end{itemize}
	then $\hat{\bs\theta}_N-\bs\theta_N = o_{\Pr}(1)$ as $N \rightarrow \infty$. 
\end{theorem}

\begin{proof}[of Theorem \ref{thm:uniform_converge}]
	We first prove that, for any $\bs{\theta} \in \Theta$, 
	$\hat{M}_N(\bs{\theta}) - M_N(\bs{\theta}) = o_{\Pr}(1)$. 
	By Chebyshev's inequality and the property of simple random sampling, 
	\begin{align}\label{eq:M_hat_point_wise}
		\hat{M}_N(\bs\theta)-M_N(\bs\theta) 
		& = \frac{1}{n}\sum_{i=1}^N Z_i m_{i}(\bs\theta) -  \frac{1}{N}\sum_{i=1}^N m_{i}(\bs\theta)
		= 
		O_{\Pr}\left(
		\Var^{1/2}\left\{ 
		\frac{1}{n}\sum_{i=1}^N Z_i m_{i}(\bs\theta)
		\right\}
		\right)
		\nonumber
		\\
		& = 
		O_{\Pr} \left(
		\sqrt{
			n^{-1} \Var_N \left\{ m_{i}(\bs\theta)\right\}
		}
		\right)
		= o_{\Pr}(1), 
	\end{align}
	where the last equality holds due to condition (iii) in Theorem \ref{thm:uniform_converge}. 
	
	We then prove $\sup_{\bs\theta\in\Theta}|\hat{M}_N(\bs\theta)-M_N(\bs\theta)| = o_{\Pr}(1)$ as $N\rightarrow \infty$. 
	For any $\varepsilon>0$ and $\eta \in (0,1)$, by condition (ii) in Theorem \ref{thm:uniform_converge}, there 
	there exists $\delta > 0$ such that for all $N$,
	\begin{align}\label{eq:delta_epsilon_6}
		\E_N \{ \Delta(m_{i},\delta) \} = 
		\frac{1}{N}\sum_{i=1}^N \Delta(m_{i},\delta)   < \eta \varepsilon/3 <\varepsilon/3. 
	\end{align}
	Because $\Theta$ is a compact set by condition (i),
	we can find a finite set $ \Theta_{\delta} \equiv \{\bs{\theta}_1,\dots,\bs{\theta}_J\} \subset \Theta$ 
	such that 
	$
	\Theta \subset \bigcup_{j=1}^J \mathcal{B}(\bs{\theta}_j, \delta), 
	$
	where $\mathcal{B}(\bs\theta_j, \delta) \equiv \{\bs\theta: \|\bs\theta - \bs\theta_j\| < \delta\}$.  
	Thus, for any $\bs\theta\in\Theta$, there must exist $\bs\theta_j \in \Theta_{\delta}$ such that $\|\bs\theta - \bs\theta_j\|< \delta$, 
	and the absolute difference between $\hat{M}_N(\bs\theta)$ and $M_N(\bs\theta)$ can then be bounded by 
	\begin{align}\label{eq:diff_bet_Fn_FN_1}
		& \quad \ |\hat{M}_N(\bs\theta)-M_N(\bs\theta)| 
		= 
		\left| \frac{1}{n}\sum_{i=1}^N Z_i m_{i}(\bs\theta) -  \frac{1}{N}\sum_{i=1}^N m_{i}(\bs\theta) \right|
		\nonumber
		\\
		& 
		\leq 
		\left|
		\frac{1}{n}\sum_{i=1}^N Z_i m_{i}(\bs\theta)- \frac{1}{n}\sum_{i=1}^N Z_i m_{i}(\bs\theta_j)
		\right|
		+
		\left|
		\frac{1}{N}\sum_{i=1}^N m_{i}(\bs\theta) - \frac{1}{N}\sum_{i=1}^N m_{i}(\bs\theta_j)
		\right|
		\nonumber
		\\
		& \quad \ + 
		\left|
		\frac{1}{n}\sum_{i=1}^N Z_i m_{i}(\bs\theta_j) - \frac{1}{N}\sum_{i=1}^N m_{i}(\bs\theta_j)
		\right| 
		\nonumber
		\\
		& \le 
		\frac1n\sum_{i=1}^N Z_i \left| m_i(\bs\theta)- 
		m_i(\bs\theta_j)
		\right|
		+
		\frac{1}{N} \sum_{i=1}^N \left|
		m_i(\bs\theta)- 
		m_i(\bs\theta_j)
		\right|
		+ 
		\left|
		\hat{M}_N(\bs\theta_j) - M_N(\bs\theta_j)
		\right| 
		\nonumber
		\\
		& 
		\le 
		\frac1n\sum_{i=1}^N Z_i \Delta(m_i, \delta)
		+
		\frac{1}{N} \sum_{i=1}^N \Delta(m_i, \delta)
		+ 
		\left|
		\hat{M}_N(\bs\theta_j) - M_N(\bs\theta_j)
		\right|, 
	\end{align} 
	where the last inequality holds by 
    the definition of $\Delta(m_i, \delta)$. 
	Consequently, 
	\begin{align*}
		\sup_{\bs\theta \in \Theta} |\hat{M}_N(\bs\theta)-M_N(\bs\theta)| 
		& \le 
		\frac1n\sum_{i=1}^N Z_i \Delta(m_i, \delta)
		+
		\frac{1}{N} \sum_{i=1}^N \Delta(m_i, \delta)
		+ 
		\max_{1\le j \le J}
		\left|
		\hat{M}_N(\bs\theta_j) - M_N(\bs\theta_j)
		\right|. 
	\end{align*}
	By the Markov inequality and \eqref{eq:delta_epsilon_6}, we then have 
	\begin{align*}%
		\Pr\left(
		\sup_{\bs\theta \in \Theta} |\hat{M}_N(\bs\theta)-M_N(\bs\theta)| > \varepsilon
		\right)
		& \le 
		\Pr\left(
		\frac1n\sum_{i=1}^N Z_i \Delta(m_i, \delta) > \frac{\varepsilon}{3}
		\right)
		+ 
		\Pr\left(
		\frac{1}{N} \sum_{i=1}^N \Delta(m_i, \delta) > \frac{\varepsilon}{3}
		\right)
		\\
		& +
		\Pr\left(
		\max_{1\le j \le J}
		\left|\hat{M}_N(\bs\theta_j) - M_N(\bs\theta_j)\right| > \frac{\varepsilon}{3}
		\right)
		\nonumber
		\\
		& \le 
		\frac{\E_N \{\Delta(m_i, \delta)\}}{\varepsilon/3} 
		+
		\sum_{j=1}^J
		\Pr\left(
		\left|\hat{M}_N(\bs\theta_j) - M_N(\bs\theta_j)\right| > \frac{\varepsilon}{3}
		\right)
		\nonumber
		\\
		& \le 
		\eta + \sum_{j=1}^J
		\Pr\left(
		\left|\hat{M}_N(\bs\theta_j) - M_N(\bs\theta_j)\right| > \frac{\varepsilon}{3}
		\right). 
	\end{align*}
	From \eqref{eq:M_hat_point_wise}, letting $N$ go to infinity, we then have 
	$$
	\limsup_{N\rightarrow \infty}\Pr\left(
	\sup_{\bs\theta \in \Theta} |\hat{M}_N(\bs\theta)-M_N(\bs\theta)| > \varepsilon
	\right)
	\le \eta. 
	$$
	Because the above inequality holds for any $\eta\in (0,1)$, we must have 
	$
	\Pr(
	\sup_{\bs\theta \in \Theta} |\hat{M}_N(\bs\theta)-M_N(\bs\theta)| > \varepsilon
	) \rightarrow 0
	$
	as $N\rightarrow \infty$. 
	Therefore, $\sup_{\bs\theta \in \Theta} |\hat{M}_N(\bs\theta)-M_N(\bs\theta)| = o_{\Pr}(1)$, i.e., Theorem \ref{thm:uniform_converge} holds.
\end{proof}

\begin{proof}[of Theorem \ref{thm:consistency}]
	From Theorem \ref{thm:uniform_converge} and 
	conditions (i) and (ii) in Theorem \ref{thm:consistency}, 
	\begin{align*}
		0 & \le M_N(\hat{\bs\theta}_N)-M_N({\bs\theta}_N)
		= 
		M_N(\hat{\bs\theta}_N) - \hat{M}_N(\hat{\bs\theta}_N)
		+ 
		\hat{M}_N(\hat{\bs\theta}_N) - \hat{M}_N({\bs\theta}_N)
		+ 
		\hat{M}_N({\bs\theta}_N) - M_N({\bs\theta}_N)\\
		& \le 
		\left| 
		M_N(\hat{\bs\theta}_N) - \hat{M}_N(\hat{\bs\theta}_N)
		\right|
		+ 
		\big\{
		\hat{M}_N(\hat{\bs\theta}_N) - \hat{M}_N({\bs\theta}_N)
		\big\}
		+ 
		\left|
		\hat{M}_N({\bs\theta}_N) - M_N({\bs\theta}_N)
		\right|
		\\
		& \le 
		2\sup_{{\bs\theta}\in\Theta}|\hat{M}_N({\bs\theta})-M_N({\bs\theta})|+\big\{
		\hat{M}_N(\hat{\bs\theta}_N) - \hat{M}_N({\bs\theta}_N)
		\big\}
		\\
		& \le 2\sup_{{\bs\theta}\in\Theta}|\hat{M}_N({\bs\theta})-M_N({\bs\theta})|
		= o_{\Pr}(1). 
	\end{align*}
	From condition (i) in Theorem \ref{thm:consistency}, 
	for any $\varepsilon > 0,$ 
	there exists $\eta > 0$ such that, 
	as $N\to \infty$, 
	$$
	\Pr(\|\hat{\bs\theta}_N-{\bs\theta}_N\|\ge \varepsilon)\leq\Pr(M_N(\hat{\bs\theta}_N)-M_N({\bs\theta}_N)\ge \eta) \to 0. 
	$$ 
	Therefore, $\hat{\bs\theta}_N-{\bs\theta}_N = o_{\Pr}(1)$, i.e., Theorem \ref{thm:consistency} holds.
\end{proof}

\subsection{Z-estimation with compactness} 

\begin{theorem}\label{thm:uniform_converge_ZCP}
        Let $\phi_{i}(\theta) = (\phi_{i1}(\theta), \ldots,  \phi_{ip}(\theta))^\top \in \mathbb{R}^p$, $1\le i \le N$, and $\tilde{\phi}(\bs\theta) \in \mathbb{R}^p$ be functions of $\theta$ in a set $\Theta \subset \mathbb{R}^p$, 
        where these functions can also depend on $N$ but we make this dependence implicit for notational convenience. 
		Let 
		$(Z_{1}, Z_{2}, \ldots, Z_{N}) \in \{0,1\}^N$ be a random vector 
		whose probability of taking value $(z_1, z_2, \ldots, z_N) \in \{0,1\}^N$ is $1/\binom{N}{n}$ if $\sum_{i=1}^N z_i = n$ and zero otherwise, 
		where $n$ is a fixed constant. 
		Define 
		\begin{align*}
		    \hat{\Phi}_N(\bs\theta)= \frac{1}{n}\sum_{i=1}^N Z_{i} \phi_{i}(\bs\theta)+\tilde{\phi}(\bs\theta), 
		    \quad 
		    \Phi_N(\bs\theta)=\frac{1}{N}\sum_{i=1}^N \phi_{i}(\bs\theta)+\tilde{\phi}(\bs\theta),
		\end{align*}
		and 
		$
		\Delta(\phi_{i}, \delta)=\sup _{(\bs\theta, \bs\theta')\in \Theta^2:\|\bs\theta-\bs\theta'\|\leq\delta}
        \| \phi_{i}(\bs\theta)-\phi_{i}(\bs\theta' ) \|.
		$
		If the following conditions hold:
		\begin{itemize}
                \item[(i)] $\Theta$ is compact,
			\item[(ii)] $\sup_N \E_N \{ \Delta(\phi_{i},\delta) \} \rightarrow 0$ as $\delta \rightarrow 0$, 
			\item[(iii)]  for any $\bs\theta\in \Theta$ and $1\le k \le p$, $\Var_N\{\phi_{ik}(\bs\theta)\}$ is $o(n)$,  
		\end{itemize}   
		then $\sup_{\bs\theta\in\Theta}\|\hat{\Phi}_N(\bs\theta)-\Phi_N(\bs\theta)\| = o_{\Pr}(1)$ as $N\rightarrow \infty$. 
\end{theorem}

\begin{theorem}\label{thm:consistency_ZCP}
	Under the same setting as in Theorem \ref{thm:uniform_converge_ZCP}, if conditions (i), (ii) and (iii) in Theorem \ref{thm:uniform_converge_ZCP} holds, 
	and the following conditions hold:
	\begin{itemize}
		\item[(i)] $\bs{\theta}_N$ is the unique root of $\Phi_N(\bs{\theta})$  satisfying that for any $\varepsilon>0$, there exists $\eta>0$ such that
		$\inf_{\bs\theta\in\Theta: \|\bs\theta-\bs\theta_{N}\| \geq \varepsilon} \|\Phi_N(\bs\theta)\|>\eta$ for all $N$,
		\item[(ii)] $\hat{\bs{\theta}}_N \in \Theta$ 
		satisfies $\hat{\Phi}_N(\hat{\bs\theta}_N) = 0$, 
	\end{itemize}
	then $\hat{\bs\theta}_N-\bs\theta_N = o_{\Pr}(1)$ as $N \rightarrow \infty$. 
\end{theorem}

\begin{proof}[of Theorem \ref{thm:uniform_converge_ZCP}]
    From conditions (i), (ii) and (iii) in Theorem \ref{thm:uniform_converge_ZCP}
    and applying Theorem \ref{thm:uniform_converge} to each coordinate of $\hat{\Phi}_{N}$,  
    we can know that $\sup_{\bs\theta\in\Theta}|\hat{\Phi}_{Nk}(\bs\theta)-\Phi_{Nk}(\bs\theta)| = o_{\Pr}(1)$ as $N\rightarrow \infty$ for $1\le k \le p$, where $\hat{\Phi}_{Nk}(\bs\theta)$ and $\Phi_{Nk}(\bs\theta)$ are the $k$th coordinates of $\hat{\Phi}_{N}(\bs\theta)$ and $\Phi_{N}(\bs\theta)$, respectively. 
    These then immediately imply that     $\sup_{\bs\theta\in\Theta}\|\hat{\Phi}_N(\bs\theta)-\Phi_N(\bs\theta)\| = o_{\Pr}(1)$. Thus, Theorem \ref{thm:uniform_converge_ZCP} holds.

\end{proof} 

\begin{proof}[of Theorem \ref{thm:consistency_ZCP}]
	From Theorem \ref{thm:uniform_converge_ZCP} and 
	condition (ii) in Theorem \ref{thm:consistency_ZCP}, 
	$$
		\| \Phi_N(\hat{\bs\theta}_N) \|= \| \Phi_N(\hat{\bs\theta}_N)-\hat{\Phi}_N({\hat{\bs\theta}}_N)\|
        \le \sup_{\bs\theta\in\Theta}\|\hat{\Phi}_N(\bs\theta)-\Phi_N(\bs\theta)\|
		= o_{\Pr}(1). 
	$$
	From condition (i) in Theorem \ref{thm:consistency_ZCP}, 
	for any $\varepsilon > 0,$ 
	there exists $\eta > 0$ such that, 
	as $N\to \infty$, 
	$$
	\Pr(\|\hat{\bs\theta}_N-{\bs\theta}_N\|\ge \varepsilon)\leq\Pr(\| \Phi_N(\hat{\bs\theta}_N) \|> \eta) \to 0. 
	$$ 
	Therefore, $\hat{\bs\theta}_N-{\bs\theta}_N = o_{\Pr}(1)$, i.e., Theorem \ref{thm:consistency_ZCP} holds.
\end{proof}

\subsection{M-estimation with convexity}\label{sec:M_est_convex}

\begin{theorem}[\citet{GuoBasse21}]\label{thm:uniform_converge_MCV}
    Under the same setting as in Theorem \ref{thm:uniform_converge}, if the following conditions hold:
    \begin{itemize}
        \item[(i)] there exists $r > 0$ such that $\mathcal{B}(\bs\theta_N, r) \subseteq \Theta$ for all $N$, where $\{\bs\theta_N\}$ is a sequence from $\Theta$,
        \item[(ii)] there exists $L \geq 0$ such that $(\E_N [ \{ m_i(\bs\theta) - m_i(\bs\theta') \}^2 ])^{1/2} \leq L \|\bs\theta - \bs\theta'\|$ for all $N$ and $\bs\theta, \bs\theta' \in \mathcal{B}(\bs\theta_N, r)$,
    \end{itemize}
    then, as $N \to \infty$, 
    \begin{equation}
    \label{eq:Guo's result}
        \sup_{\bs\theta \in \Theta:\|\bs\theta-\bs\theta_N\| \leq r} \big| \{ \hat{M}_N(\theta) - M_N(\theta) \} - \{ \hat{M}_N(\theta_N) - M_N(\theta_N) \} \big| = o_{\Pr}(1).
    \end{equation}
    If, in addition, 
    \begin{itemize}
        \item[{(iii)}]
        $var_N\{m_{i}(\bs\theta_N)\}= o(n) $,
    \end{itemize}
    then, as $N \to \infty$, 
    $$
    \sup_{\bs\theta \in \Theta:\|\bs\theta-\bs\theta_N\| \leq r} \big|  \hat{M}_N(\theta) - M_N(\theta) \big| = o_{\Pr}(1).
    $$
    
\end{theorem}

\begin{theorem}[\citet{GuoBasse21}]\label{thm:consistency_MCV}
    Under the same setting as in Theorem \ref{thm:uniform_converge},if conditions (i) and (ii) in Theorem \ref{thm:uniform_converge_MCV} holds, 
	and the following conditions hold:
    \begin{itemize}
        \item[(i)] $\Theta$ is convex, 
        \item[(ii)] $m_i(\bs\theta)$ is convex for all $i$, and $\tilde{m}(\bs\theta)$ is convex, 
	\item[(iii)] 
        $\bs{\theta}_N$ is the unique minimizer of $M_N(\bs{\theta})$  satisfying that for any $\varepsilon>0$, there exists $\eta>0$ such that
		$\inf_{\bs\theta\in\Theta: \|\bs\theta-\bs\theta_{N}\| =\varepsilon} M_N(\bs\theta)>M_N\left(\bs\theta_{N}\right)+\eta$ for all $N$,
	\item[(iv)] $\hat{\bs{\theta}}_N \in \Theta$ 
		satisfies $\hat{M}_N(\hat{\bs\theta}_N)\leq \hat{M}_N(\bs{\theta}_N)$, 
	\end{itemize}
	then $\hat{\bs\theta}_N-\bs\theta_N = o_{\Pr}(1)$ as $N \rightarrow \infty$. 
\end{theorem}

\begin{proof}[of Theorem \ref{thm:uniform_converge_MCV}]
    The proof of \eqref{eq:Guo's result} can be found in the proof of Theorem 9 in \citet{GuoBasse21}. 
    Below we prove the uniform convergence of $\hat{M}_N(\theta)$. 
    By Chebyshev's inequality and the same logic as \eqref{eq:M_hat_point_wise}, we can show $|\hat{M}_{N}(\bs\theta_N)-M_{N}(\bs\theta_N)| = o_{\Pr}(1)$ as $N\rightarrow \infty$ using conditions (iii). 
    Combined with \eqref{eq:Guo's result}, we then have
    \begin{align*}
    & \quad \ \sup_{\bs\theta \in \Theta:\|\bs\theta-\bs\theta_N\| \leq r} \big|  \hat{M}_N(\theta) - M_N(\theta) \big| \\
    & \leq \sup_{\bs\theta \in \Theta:\|\bs\theta-\bs\theta_N\| \leq r} \left| \left( \hat{M}_{N}(\theta) - M_{N}(\theta) \right) - \left( \hat{M}_{N}(\theta_N) - M_{N}(\theta_N) \right) \right|
    + |\hat{M}_{N}(\bs\theta_N)-M_{N}(\bs\theta_N)| 
    \\
    & = o_{\Pr}(1).  
    \end{align*}
    Therefore, Theorem \ref{thm:uniform_converge_MCV} holds. 
\end{proof}

\begin{proof}[of Theorem \ref{thm:consistency_MCV}]
The proof can also be found in \citet{GuoBasse21}. 
Here we give a more detailed proof for clarity. 
From condition (iii), for any $\varepsilon>0$, there exists $\eta>0$ such that, for all $N$,  $\inf_{\bs\theta\in\Theta: \|\bs\theta-\bs\theta_{N}\| =\varepsilon} M_N(\bs\theta)>M_N\left(\bs\theta_{N}\right)+\eta$. 
For any $\bs\theta\in\Theta$ such that $\|\bs\theta-\bs\theta_{N}\| =\varepsilon$, we then have 
\begin{align*}
    \hat{M}_N(\theta) - \hat{M}_N(\theta_N) 
    & \ge 
    M_N(\theta)  - M_N(\theta_N) 
    - 
    \big| \big\{ \hat{M}_N(\theta) - \hat{M}_N(\theta_N) \big\} - \big\{  M_N(\theta)  - M_N(\theta_N) \big\} \big|
    \\
    & \ge \eta - \big| \big\{ \hat{M}_N(\theta) - M_N(\theta) \big\} - \big\{ \hat{M}_N(\theta_N) - M_N(\theta_N) \big\} \big|.
\end{align*}
Consequently, from Theorem \ref{thm:uniform_converge_MCV}, we have 
\begin{align*}
    \inf_{\bs\theta \in \Theta:\|\bs\theta-\bs\theta_N\| = \varepsilon} \hat{M}_N(\theta) - \hat{M}_N(\theta_N) 
    & \ge \eta - \sup_{\bs\theta \in \Theta:\|\bs\theta-\bs\theta_N\| = \varepsilon} \big| \big\{ \hat{M}_N(\theta) - M_N(\theta) \big\} -  \big\{ \hat{M}_N(\theta_N) - M_N(\theta_N) \big\} \big|
    \\
    & = \eta + o_{\Pr}(1),
\end{align*}
where the last equality is due to Theorem \ref{thm:uniform_converge_MCV}. 
This implies that, as $N\rightarrow \infty$,  
\begin{align*}
    \Pr\left(\inf_{\bs\theta \in \Theta:\|\bs\theta-\bs\theta_N\| = \varepsilon} \hat{M}_N(\theta) - \hat{M}_N(\theta_N) > \eta/2 \right) \rightarrow 1
\end{align*}

Below we prove that $\inf_{\bs\theta \in \Theta:\|\bs\theta-\bs\theta_N\| = \varepsilon} \hat{M}_N(\theta) - \hat{M}_N(\theta_N) > \eta/2$ must imply $\|\hat \theta_N - \theta_N\| \le  \varepsilon$. 
We prove this by contradiction. 
Suppose that $\inf_{\bs\theta \in \Theta:\|\bs\theta-\bs\theta_N\| = \varepsilon} \hat{M}_N(\theta) - \hat{M}_N(\theta_N) > \eta/2$ and 
$\|\hat \theta_N - \theta_N\| >  \varepsilon$. 
There must exist 
$\theta_N' = (1-\lambda) \theta_N + \lambda \hat{\theta}_N$, with $\lambda \in (0,1)$, such that  $\|\theta_N' - \theta_N\| = \varepsilon$. Since $\hat{M}_N(\bs\theta)$ is convex, 
from condition (iv), 
we have 
\begin{align*}
\hat{M}_N(\theta_N') - \hat{M}_N(\theta_N) &\leq (1-\lambda) \hat{M}_N(\theta_N) +  \lambda \hat{M}_N(\hat{\theta}_N) - \hat{M}_N(\theta_N) = \lambda\{ \hat{M}_N(\hat{\theta}_N) - \hat{M}_N(\theta_N)\} \leq 0. 
\end{align*}
However, this contradicts with our assumption that $\inf_{\bs\theta \in \Theta:\|\bs\theta-\bs\theta_N\| = \varepsilon} \hat{M}_N(\theta) - \hat{M}_N(\theta_N) > \eta/2$. 

From the above, we then have, as $N\rightarrow \infty$, 
$$
\Pr(\|\hat{\theta}_N - \theta_N\| \leq \varepsilon) \geq \Pr\left(\sup_{\theta \in \Theta: \|\theta - \theta_N\| = \varepsilon} \hat{M}_N(\theta) - \hat{M}_N(\theta_N) > \eta/2 \right) \rightarrow 1.
$$
Therefore, $\hat{\bs\theta}_N-\bs\theta_N = o_{\Pr}(1)$ as $N \rightarrow \infty$, 
i.e., Theorem \ref{thm:consistency_MCV} holds. 
\end{proof}

\subsection{Z-estimation with convexity}

\begin{theorem}\label{thm:uniform_converge_ZCV}
    Under the same setting as in Theorem \ref{thm:uniform_converge_ZCP}, if the following conditions hold:
    \begin{itemize}
        \item[(i)] there exists $r > 0$ such that $\mathcal{B}(\theta_N, r) \subseteq \Theta$ for all $N$, where $\{\bs\theta_N\}$ is a sequence from $\Theta$,
        \item[(ii)] there exists $L \geq 0$ such that $(\E_N \left[ \| \phi_i(\bs\theta) - \phi_i(\bs\theta') \|^2 \right])^{1/2} \leq L \|\bs\theta - \bs\theta'\|$ for all $N$ and $\bs\theta, \bs\theta' \in \mathcal{B}(\bs\theta_N, r)$,
        \item[(iii)] $var_N\{\phi_{ik}(\bs\theta_N)\}$ is $o(n)$ for each coordinate $1\le k \le p$,
    \end{itemize}
    then, as $N \to \infty$, 
    $$
    \sup_{\bs\theta \in \Theta:\|\bs\theta-\bs\theta_N\| \leq r} \| \hat{\Phi}_N(\theta) - \Phi_N(\theta)  \| = o_{\Pr}(1).
    $$
\end{theorem}

\begin{theorem}\label{thm:consistency_ZCV}
    Under the same setting as in Theorem \ref{thm:uniform_converge_ZCP}, if conditions (i)--(iii) in Theorem \ref{thm:uniform_converge_ZCV} holds, 
	and the following conditions hold:
    \begin{itemize}
        \item[(i)] $\Theta$ is convex, 
        \item[(ii)] For all $i$ and $ \bs\theta, \bs\theta' \in \Theta$, $\phi_i(\bs\theta)$ and $\tilde \phi(\bs\theta)$ satisfy 
        $$
        \{ \phi_i(\bs\theta) - \phi_i(\bs\theta') \}^\top(\bs\theta - \bs\theta') \geq 0 \text{ and } \{ \tilde\phi(\bs\theta) - \tilde\phi(\bs\theta') \}^\top(\bs\theta - \bs\theta') \geq 0, 
         $$
	\item[(iii)] 
        $\bs{\theta}_N$ is the unique root of $\Phi(\bs{\theta})$  satisfying that for any $\varepsilon>0$, there exists $\eta>0$ such that
		$\inf_{\bs\theta\in\Theta: \|\bs\theta-\bs\theta_{N}\| = \varepsilon} \Phi_N(\bs\theta)^\top(\bs\theta-\bs\theta_N) > \eta$ for all $N$,
	\item[(iv)] $\hat{\bs{\theta}}_N \in \Theta$ 
		satisfies $\hat{\Phi}_N(\hat{\bs\theta}_N) = 0$, 
	\end{itemize}
	then $\hat{\bs\theta}_N-\bs\theta_N = o_{\Pr}(1)$ as $N \rightarrow \infty$. 
\end{theorem}

\begin{remark}
    Condition (ii) in Theorem \ref{thm:consistency_ZCV} is motivated by the convex risk minimization in \S \ref{sec:M_est_convex}. In particular, condition (ii) in Theorem \ref{thm:consistency_ZCV} holds when $\phi_i(\bs\theta)$s and $\tilde \phi(\bs\theta)$ are derivatives of convex functions. 
\end{remark}

\begin{proof}[of Theorem \ref{thm:uniform_converge_ZCV}]
    Applying Theorem \ref{thm:uniform_converge_MCV} to each coordinate of $\phi_{i}$, we have, for $1\le k \le p$,  
    $\sup_{\bs\theta \in \Theta:\|\bs\theta-\bs\theta_N\| \leq r} |  \hat{\Phi}_{Nk}(\theta) - \Phi_{Nk}(\theta) | = o_{\Pr}(1) \text{ as } N \to \infty$, 
    where $\hat{\Phi}_{Nk}(\theta)$ and $ \Phi_{Nk}(\theta)$ are the $k$th coordinate of $\hat{\Phi}_{N}(\theta)$ and $ \Phi_{N}(\theta)$, respectively. 
    This then implies that $\sup_{\bs\theta\in\Theta:\|\bs\theta-\bs\theta_N\| \leq r}\|\hat{\Phi}_N(\bs\theta)-\Phi_N(\bs\theta)\| = o_{\Pr}(1)$.
    Therefore, Theorem \ref{thm:uniform_converge_ZCV} holds. 
\end{proof}

\begin{proof}[of Theorem \ref{thm:consistency_ZCV}]
    Define $r$ the same as in  Theorem \ref{thm:uniform_converge_ZCV}. 
    From condition (iii), for any $0< \varepsilon\le r$, there exists $\eta>0$ such that
    \begin{align*}
    \eta & < \inf_{\bs\theta \in \Theta:\|\bs\theta - \bs\theta_N\| = \varepsilon} \Phi_N(\bs\theta)^\top(\bs\theta - \bs\theta_N) \\
    & \leq  \sup_{\bs\theta \in \Theta:\|\bs\theta - \bs\theta_N\| = \varepsilon} (\Phi_N(\bs\theta) - \hat{\Phi}_N(\bs\theta))^\top(\bs\theta - \bs\theta_N) + \inf_{\bs\theta \in \Theta:\|\bs\theta - \bs\theta_N\| = \varepsilon} \hat{\Phi}_N(\bs\theta)^\top(\bs\theta - \bs\theta_N) \\
    & \leq  \sup_{\bs\theta \in \Theta:\|\bs\theta - \bs\theta_N\| = \varepsilon}\|\Phi_N(\bs\theta) - \hat{\Phi}_N(\bs\theta)\|\|\bs\theta - \bs\theta_N\| + \inf_{\bs\theta \in \Theta:\|\bs\theta - \bs\theta_N\| = \varepsilon} \hat{\Phi}_N(\bs\theta)^\top(\bs\theta - \bs\theta_N) \\
    & = o_{\Pr}(1) + \inf_{\bs\theta \in \Theta:\|\bs\theta - \bs\theta_N\| = \varepsilon} \hat{\Phi}_N(\bs\theta)^\top(\bs\theta - \bs\theta_N),
    \end{align*}
    where the second inequality follows by some algebra, the third inequality follows by the Cauchy–Schwarz inequality, and the last equality follows from Theorem \ref{thm:uniform_converge_ZCV}. 
    This then implies that, as $N\rightarrow \infty$, 
    $$
    \Pr \Big( \inf_{\bs\theta \in \Theta:\|\bs\theta - \bs\theta_N\| = \varepsilon} \hat{\Phi}_N(\theta)^\top(\theta - \theta_N) \geq {\eta}/{2} \Big) \rightarrow 1. 
    $$ 

    We now prove that  $\inf_{\bs\theta \in \Theta:\|\bs\theta - \bs\theta_N\| = \varepsilon} \hat{\Phi}_N(\theta)^\top(\theta - \theta_N) \geq {\eta}/{2}$ must imply $\|\hat{\theta}_N - \theta_N\| \le \varepsilon$. 
    We prove this by contradiction. 
    Suppose that $\inf_{\bs\theta \in \Theta:\|\bs\theta - \bs\theta_N\| = \varepsilon} \hat{\Phi}_N(\theta)^\top(\theta - \theta_N) \geq {\eta}/{2}$ and $\|\hat{\theta}_N - \theta_N\| > \varepsilon$. 
    Because $\Theta$ is convex by condition (i), there must exist $\bs\theta' \in \Theta$ such that $\|\bs\theta' - \bs\theta_N\| = \varepsilon$ and $\bs\theta' = \lambda \hat{\bs\theta}_N + (1-\lambda)\bs\theta_N$ for some $0 < \lambda < 1$. 
    From conditions (ii) and (iv), we have 
    \begin{align*}
        \hat\Phi_N(\bs\theta')^\top(\bs\theta' - \bs\theta_N) & =  
        - \frac{\lambda}{1-\lambda}\hat{\Phi}_N(\bs\theta')^\top(\bs\theta' - \hat{\bs\theta}_N) = 
        - \frac{\lambda}{1-\lambda}\{ \hat{\Phi}_N(\bs\theta')-\hat{\Phi}_N(\bs\hat\theta_N) \}^\top(\bs\theta' - \hat{\bs\theta}_N) \\
        & \leq 0.
    \end{align*}
    However, this contradics with our assumption that $\inf_{\bs\theta \in \Theta:\|\theta - \theta_N\| = \varepsilon} \hat{\Phi}_N(\theta)^\top(\theta - \theta_N) \geq {\eta}/{2}$.

    From the above, we have 
    \begin{align*}
        \Pr\big( \|\hat{\theta}_N - \theta_N\| \le \varepsilon \big)
        \ge \Pr \Big( \inf_{\bs\theta \in \Theta:\|\bs\theta - \bs\theta_N\| = \varepsilon} \hat{\Phi}_N(\theta)^\top(\theta - \theta_N) \geq {\eta}/{2} \Big) \rightarrow 1 \text{ as } N\rightarrow \infty. 
    \end{align*}
    Therefore, $\hat{\theta}_N - \theta_N = o_{\Pr}(1)$, i.e., Theorem \ref{thm:consistency_ZCV} holds. 
\end{proof} 

\section{Asymptotic Normality for M-estimation and Z-estimation under simple random sampling}\label{sec:clt_srs}

In this section, we establish the $\sqrt{n}$-consistency and asymptotic Normality of the Z-estimator under simple random sampling, which includes the M-estimator as special cases. 
Moreover, 
for conciseness, 
we impose conditions on uniform convergence of $\dot{\hat{\Phi}}_N(\bs\theta)$ and the consistency of $\hat{\theta}_N$, which can be guaranteed by conditions similar to those in \S \ref{sec:M_srs}.

\subsection{Theorems}
\begin{theorem}\label{thm:rate}
	Under the same setting as in Theorem \ref{thm:uniform_converge_ZCP}. 
	If the following conditions hold: 
	\begin{enumerate}[label=(\roman*)]
        \item 
        $\theta_N$ is the root of $\Phi_N(\theta)$, and 
		there exists $\varepsilon_0>0$ such that $\mathcal{B}({\bs\theta}_N, \varepsilon_0 ) = \{\theta: \|\theta-\theta_N\| <  \varepsilon_0\} \subset \Theta$ for all $N$, 
 
        \item
		$\hat{\Phi}_N(\hat{\bs\theta}_N)=0$, 
        and $\hat{\bs\theta}_N-\bs\theta_N = o_{\Pr}(1)$ as $N \rightarrow \infty$,

		\item for all $N$ and $1\le i \le N$, 
		$\tilde \phi({\bs\theta})$ and $\phi_i( {\bs\theta})$
        have continuous first order partial derivatives over ${\bs\theta} \in \Theta$,

        \item there exists $r_0 > 0$ such that $\sup_{\bs\theta \in \Theta:\|\bs\theta-\bs\theta_N\|\leq r_0} \| \dot{\hat{\Phi}}_N(\bs\theta)  - \dot{\Phi}_N(\bs\theta) \| = o_{\Pr}(1)$,
         
		\item
		$\sup_N \sup_{\bs\theta \in \Theta:\|\bs\theta-\bs\theta_N\|\leq\delta} \| \dot{\Phi}_N(\bs\theta)  - \dot{\Phi}_N(\bs\theta_N) \| \rightarrow 0$ 
		as $\delta \rightarrow 0$, 
        \item 
        $var_N\{\phi_{ik}(\bs\theta_N)\}= O(1) $ for $1 \leq k \leq p$,
		
		\item
		there exists $c_0>0$ such that
        $\sigma_{\min}(\dot{\Phi}_N({\bs\theta}_N)) \ge c_0$  
  for all $N$,
	\end{enumerate}
	then $\hat{\bs\theta}_N-{\bs\theta}_N=O_\Pr( n^{-1/2} )$ as $N\rightarrow \infty$.
\end{theorem}

\begin{theorem}\label{thm:clt}
	Under the same setting as in Theorem \ref{thm:uniform_converge_ZCP},  
	if
	conditions (i)--(vii) in Theorem \ref{thm:rate} hold, 
	and 
	the following conditions hold: 
	\begin{enumerate}[label=(\roman*)]
		\item $\dot{\Phi}_N({\bs\theta}_N) \rightarrow \bs{\Gamma}> 0$, and $\Cov_N\{ \phi_i({\bs\theta}_N) \} \rightarrow \bs{\Sigma}$ as $N\rightarrow \infty$, 
		\item 
		$
		\max_{1\le j \le N} \| \phi_j({\bs\theta}_N) -  \E_N\{ \phi_i({\bs\theta}_N) \} \|^2 / \min(n, N-n) \rightarrow 0
		$
		as $N\rightarrow \infty$, 
	\end{enumerate}
	then  
	\begin{align*}
		\sqrt{\frac{nN}{N-n}} \cdot ( \hat{{\bs\theta}}_N-{\bs\theta}_N ) 
		& \converged \mathcal{N}\left( 
		\bs{0}, \ 
		\bs{\Gamma}^{-1} \bs{\Sigma}  (\bs{\Gamma}^{-1})^\top
		\right). 
	\end{align*}
\end{theorem}

\subsection{Proof of the theorems}

\begin{proof}[of Theorem \ref{thm:rate}]
    Without loss of generality, we assume $\varepsilon_0$ is less than or equal to $r_0$ in condition (iv). 
	Let $\hat{\bs{\theta}}_N'$ denote the projection of $\hat{\bs{\theta}}_N$ onto the set $\overline{\mathcal{B}}(\bs{\theta}_N, \varepsilon_0/2) \equiv \{\theta: \|\theta-\theta_N\|\le \varepsilon_0/2\}$. 
	From condition (ii), $\hat{\bs{\theta}}_N - \bs{\theta}_N = o_{\Pr}(1)$, which implies that 
	$
	\Pr(\hat{\bs{\theta}}_N\ne \hat{\bs{\theta}}_N') \le \Pr(\|\hat{\bs{\theta}}_N-\bs{\theta}_N \|> \varepsilon_0/2) \rightarrow 0
	$
	as $N\rightarrow \infty$. 
	Consequently, we must have 
    $\Pr\{\hat{\Phi}_N(\hat{\bs\theta}_N') \ne 0\} = \Pr\{\hat{\Phi}_N(\hat{\bs\theta}_N') \ne \hat{\Phi}_N(\hat{\bs\theta}_N)\} \le \Pr(\hat{\bs{\theta}}_N\ne \hat{\bs{\theta}}_N')  = o(1)$
	and 
	$\hat{\bs{\theta}}_N' - \hat{\bs{\theta}}_N = o_{\Pr}(n^{-1/2})$. 
	Thus, to prove Theorem \ref{thm:rate}, it suffices to prove that $\hat{\bs{\theta}}_N' - \bs{\theta}_N = O_{\Pr}(n^{-1/2})$. 
    For descriptive convenience, 
    in the following, 
    we will simply assume that 
	$\hat{\bs{\theta}}_N \in \overline{\mathcal{B}}(\bs{\theta}_N, \varepsilon_0/2) \subset \Theta$ for all $N$. 
    Recall that we now only have $\Pr\{\hat{\Phi}_N(\hat{\bs\theta}_N) \ne 0\} = o(1)$.

	First, 
	from conditions (iii), 
	by 
	the fundamental theorem of calculus, 
	\begin{equation}\label{eq:taylor_M1}
		\hat{\Phi}_N(\hat{\bs{\theta}}_N) -  \hat{\Phi}_N(\bs{\theta}_N) = 
		\int_{0}^1 \dot{\hat{\Phi}}_N\big(\bs{\theta}_N+u(\hat{\bs{\theta}}_N-\bs{\theta}_N) \big) \text{d}u  (\hat{\bs{\theta}}_N-\bs{\theta}_N)
		\equiv 
		\bs{H} (\hat{\bs{\theta}}_N-\bs{\theta}_N), 
	\end{equation}
	where $\bs{H} \equiv \int_{0}^1 \dot{\hat{\Phi}}_N(\bs{\theta}_N+u(\hat{\bs{\theta}}_N-\bs{\theta}_N) ) \text{d}u$.
	Then 
	$
	\hat{{\bs\theta}}_N-{\bs\theta}_N = \bs{H}^{-1}
	\{ \hat{\Phi}_N(\hat{\bs\theta}_N)- \hat{\Phi}_N ({\bs\theta}_N) \},  
	$ 
	as long as $\bs{H}$ is invertible. 
	Therefore, to prove $\hat{{\bs\theta}}_N-{\bs\theta}_N = O_{\Pr}(n^{-1/2})$, it suffices to prove that 
	\begin{align}\label{eq:diff_theta_check_theta}
		\bs{H}^{-1}=O_\Pr(1), 
		\ \ \ 
		\text{and}
		\ \ \ 
		\hat{\Phi}_N (\hat{\bs\theta}_N)- \hat{\Phi}_N ({\bs\theta}_N)=O_\Pr(n^{-1/2}).
	\end{align}
	
	Second, we prove the first equality in \eqref{eq:diff_theta_check_theta}. 
    By definition, 
	\begin{align}\label{eq:bound_H}
		& \quad \ \| \bs{H} - \dot{\Phi}_N(\bs{\theta}_N) \|
        \nonumber
        \\
		& = 
		\left\| \int_{0}^1 \big\{ \dot{\hat{\Phi}}_N\big(\bs{\theta}_N+u(\hat{\bs{\theta}}_N-\bs{\theta}_N) \big) - \dot{\Phi}_N\big(\bs{\theta}_N+u(\hat{\bs{\theta}}_N-\bs{\theta}_N) \big) \big\} \text{d}u \right.
        \nonumber\\
		& \quad \ \left. + 
		\int_{0}^1 \big\{ \dot{\Phi}_N\big(\bs{\theta}_N+u(\hat{\bs{\theta}}_N-\bs{\theta}_N) \big) - \dot{\Phi}_N(\bs{\theta}_N) \big\} \text{d}u \right\|
		\nonumber\\
		& 
		\le 
		\int_{0}^1 \left\| \dot{\hat{\Phi}}_N\big(\bs{\theta}_N+u(\hat{\bs{\theta}}_N-\bs{\theta}_N) \big) - \dot{\Phi}_N\big(\bs{\theta}_N+u(\hat{\bs{\theta}}_N-\bs{\theta}_N) \big) \right\| \text{d}u
		\nonumber\\
		& \quad \  + 
		\int_{0}^1 \left\| \dot{\Phi}_N\big(\bs{\theta}_N+u(\hat{\bs{\theta}}_N-\bs{\theta}_N) \big) - \dot{\Phi}_N(\bs{\theta}_N) \right\| \text{d}u
		\nonumber\\
		& \le 
		\sup_{\bs{\theta} \in \Theta: \|\bs{\theta} - \bs{\theta}_N\| \le \|\hat{\bs{\theta}}_N - \bs{\theta}_N\|} \big\| \dot{\hat{\Phi}}_N(\bs{\theta} ) - \dot{\Phi}_N(\bs{\theta} \big) \|
		+ 
		\sup_{\bs{\theta} \in \Theta: \|\bs{\theta} - \bs{\theta}_N\| \le \|\hat{\bs{\theta}}_N - \bs{\theta}_N\|} \| \dot{\Phi}_N(\bs{\theta}) - \dot{\Phi}_N(\bs{\theta}_N) \|. 
	\end{align}
    Below we consider the two terms in \eqref{eq:bound_H} separately. 
    Recall that $\|\hat{\bs{\theta}}_N - \bs{\theta}_N\|\le \varepsilon_0/2\le r_0$. From condition (iv), we have 
    \begin{align}\label{eq:bound_H_1}
        \sup_{\bs{\theta} \in \Theta: \|\bs{\theta} - \bs{\theta}_N\| \le \|\hat{\bs{\theta}}_N - \bs{\theta}_N\|} \big\| \dot{\hat{\Phi}}_N(\bs{\theta} ) - \dot{\Phi}_N(\bs{\theta} \big) \|
        \le 
        \sup_{\bs{\theta} \in \Theta: \|\bs{\theta} - \bs{\theta}_N\| \le r_0} \big\| \dot{\hat{\Phi}}_N(\bs{\theta} ) - \dot{\Phi}_N(\bs{\theta} \big) \| = o_{\Pr}(1). 
    \end{align}
    From condition (v), for any $\varepsilon>0$, there exists $\eta$ such that 
    $\sup_{\bs\theta \in \Theta:\|\bs\theta-\bs\theta_N\|\leq\eta} \| \dot{\Phi}_N(\bs\theta)  - \dot{\Phi}_N(\bs\theta_N) \| \le \varepsilon$ for all $N$. 
    Consequently, 
    \begin{align*}
        \Pr 
        \Big( \sup_{\bs{\theta} \in \Theta: \|\bs{\theta} - \bs{\theta}_N\| \le \|\hat{\bs{\theta}}_N - \bs{\theta}_N\|} \| \dot{\Phi}_N(\bs{\theta}) - \dot{\Phi}_N(\bs{\theta}_N) \| > \varepsilon \Big)
        \le 
        \Pr (\|\hat{\bs{\theta}}_N - \bs{\theta}_N\| > \eta ) = o(1),
    \end{align*}
    where the last equality is due to condition (ii). 
    This then implies that 
    \begin{align}\label{eq:bound_H_2}
        \sup_{\bs{\theta} \in \Theta: \|\bs{\theta} - \bs{\theta}_N\| \le \|\hat{\bs{\theta}}_N - \bs{\theta}_N\|} \| \dot{\Phi}_N(\bs{\theta}) - \dot{\Phi}_N(\bs{\theta}_N) \| = o_{\Pr}(1). 
    \end{align}
    From \eqref{eq:bound_H}, \eqref{eq:bound_H_1} and \eqref{eq:bound_H_2}, 
    we have $\bs{H} - \dot{\Phi}_N(\bs{\theta}_N) = o_{\Pr}(1)$. 
	Note that condition (vii) implies $\dot{\Phi}_N({\bs\theta}_N)^{-1} =  O_\Pr (1)$. 
	We have $\dot{\Phi}_N({{\bs\theta}}_N)^{-1} \cdot \bs{H} = \bs{I}  + o_{\Pr} (1)$. 
	By the continuous mapping theorem, 
	$
	\bs{H}^{-1} \cdot \dot{\Phi}_N({{\bs\theta}}_N) = \bs{I}  + o_{\Pr} (1). 
	$
    Consequently, 
	\begin{align}\label{eq:H_inverse_limit}
		\bs{H}^{-1} - \dot{\Phi}_N({{\bs\theta}}_N)^{-1}
		& = 
		\big\{ \bs{H}^{-1}  \cdot \dot{\Phi}_N ({\bs\theta}_N) - \bs{I} \big\} \cdot \dot{\Phi}_N({\bs\theta}_N) ^{-1}
		= o_{\Pr}(1). 
	\end{align}
	Therefore, 
	$
	\bs{H}^{-1} = \dot{\Phi}_N({{\bs\theta}}_N)^{-1} + o_{\Pr}(1) = O_{\Pr}(1). 
	$
	
	Third, we prove the second equality in \eqref{eq:diff_theta_check_theta}, i.e., 
	$\hat{\Phi}_N (\hat{\bs\theta}_N)- \hat{\Phi}_N ({\bs\theta}_N)=O_\Pr(n^{-1/2})$.
	From the discussion at the beginning of the proof,  
    $\Pr\{\hat{\Phi}_N(\hat{\bs\theta}_N) \ne 0\} = o(1)$, 
    which implies that
	$ \hat{\Phi}_N (\hat{\bs\theta}_N) =  O_{\Pr}(n^{-1/2})$. 
	Thus, it suffices to prove that 
	$\hat{\Phi}_N ({\bs\theta}_N) = O_{\Pr}(n^{-1/2})$. 
	By the property of simple random sampling, 
	\begin{align*}
		\E \{ \hat{\Phi}_N ({\bs\theta}_N) \}
		& = 
		\E\left\{
		\frac{1}{n} \sum_{i=1}^N Z_i \phi_i(\bs{\theta}_N)
		\right\} + \tilde{\phi}(\bs{\theta}_N)
		= 
		\frac{1}{N} \sum_{i=1}^N \phi_i(\bs{\theta}_N) + \tilde{\phi}(\bs{\theta}_N)
		= \Phi_N ({\bs\theta}_N) = \bs{0},
	\end{align*}
	where the last equality holds by condition (i).
	Furthermore, from condition (vi), for $1\le k \le p$, 
	\begin{align*}
		\Var\big\{ \hat{\Phi}_{Nk} ({\bs\theta}_N) \big\}
		& = 
		\left( \frac{1}{n} - \frac{1}{N} \right) 
		\Var_N\{ \phi_{ik}( \bs{\theta}_N) \} = O(n^{-1}). 
	\end{align*}
	Thus, 
	by Chebyshev's inequality, 
	we must have  
	$\hat{\Phi}_N ({\bs\theta}_N) = O_\Pr(n^{-1/2}).$
	
	From the above, as $N\rightarrow \infty$,  
	\begin{align}\label{eq:theta_hat_taylor}
		\hat{{\bs\theta}}_N-{\bs\theta}_N & = 
		\bs{H}^{-1}
		\big\{ \hat{\Phi}_N(\hat{\bs\theta}_N)- \hat{\Phi}_N ({\bs\theta}_N) \big\}
		= O_{\Pr}(1) \cdot 
		O_{\Pr}(n^{-1/2}) 
		= O_\Pr(n^{-1/2}), 
	\end{align}
	i.e., Theorem \ref{thm:rate} holds.
\end{proof}

To prove Theorem \ref{thm:clt}, we need the following lemma. 

\begin{lemma}\label{lemma:clt_srs}
	Let $\{\bs{y}_1, \bs{y}_2, \ldots, \bs{y}_N\}$ be a finite population of size $N$ with $\bs{y}_i = (y_{i1}, \ldots, y_{iK})^\top \in \mathbb{R}^K$, 
	and 
	$ (Z_1, \ldots, Z_N) \in \{0,1\}^N$ be a random vector 
	whose probability of taking value $(z_1, z_2, \ldots, z_N)$ is $1/\binom{N}{n}$ if $\sum_{i=1}^N z_i = n$ and zero otherwise, where $n$ is a fixed constant. 
	Let $\bar{\bs{y}}_N= 
	\E_N(\bs{y}_i)
	= (\bar{y}_{N1}, \ldots, \bar{y}_{NK})^\top$ be the finite population average, 
	and 
	$
	\hat{\bs{y}}_N = n^{-1} \sum_{i=1}^N Z_i \bs{y}_i
	$
	be the average of units with $Z_i$'s equal 1. 
	If the sequence of finite populations satisfies that, as $N \rightarrow \infty$,
	\begin{itemize}
		\item[(i)] $\Cov_N( \bs{y_i} )$ has a limiting value $\bs{\Sigma}$,
		\item[(ii)] 
		$
		\max _{1 \leq i \leq N} 
		\left\| \bs{y}_i - \bar{\bs{y}}_N \right\|^2/\min\{n, N-n\}
		\rightarrow 0,
		$
	\end{itemize}
	then 
	$$
	\sqrt{\frac{nN}{N-n}} \cdot (\hat{\bs{y}}_N-\bar{\bs{y}}_N) \converged \mathcal{N} \left(\bs{0}, \bs{\Sigma} \right). 
	$$
\end{lemma}

\begin{proof}[of Lemma \ref{lemma:clt_srs}]
	Define pseudo potential outcomes 
	$\bs{Y}_i(1) = \bs{y}_i$ and $\bs{Y}_i(0) = \bs{0}$ for $1\le i \le N$, 
	constant coefficient matrices 
	$\bs{A}_1=\bs{I}_{K}$ and $\bs{A}_0=\bs{0}_{K\times K}$, 
	and 
	$
	\bs{\tau}_i(\bs{A}) = \bs{A}_1 \bs{Y}_i(1) + \bs{A}_0 \bs{Y}_i(0) =\bs{y}_i. 
	$ 
	Then the finite population averages of $\bs{Y}_i(1)$'s, $\bs{Y}_i(0)$'s and $\bs{\tau}_i(\bs{A})$'s are, respectively, 
	$\bar{\bs{Y}}(1) \equiv \E_N\{ \bs{Y}_i(1) \} = \bar{\bs{y}}_N$, 
	$\bar{\bs{Y}}(0) \equiv \E_N\{ \bs{Y}_i(0) \} = \bs{0}_{K \times 1}$
	and 
	$
	\bs{\tau}(\bs{A}) \equiv N^{-1} \sum_{i=1}^N \bs{\tau}_i(\bs{A}) = \bar{\bs{y}}_N. 
	$
	We further define 
	$\bar{\bs{Y}}_1 = n^{-1}\sum_{i:Z_i = 1} \bs{Y}_i(1)$ 
	and 
	$\bar{\bs{Y}}_0 = (N-n)^{-1} \sum_{i:Z_i = 0} \bs{Y}_i(0)$ as the averages of treatment and control potential outcomes among units with $Z_i$'s equal 1 and 0, respectively, 
	and 
	$
	\hat{\bs{\tau}}(\bs{A}) \equiv \bs{A}_1 \bar{\bs{Y}}_1 + \bs{A}_0 \bar{\bs{Y}}_0. 
	$
	We can verify that 
	$
	\hat{\bs{\tau}}(\bs{A}) = \bar{\bs{Y}}_1 =  \hat{\bs{y}}_N. 
	$
	Following Li and Ding (2017, Theorem 4), 
	we introduce 
	$n_1 = n$, $n_0 = N-n$, and for $1\le k \le K$, 
	\begin{align*}
		m_z(k) & \equiv \max_{1\le i\le N} \left[ \bs{A}_z \bs{Y}_i(z) - \bs{A}_z \bar{\bs{Y}}(z) \right]_{(k)}^2 
		= 
		\begin{cases}
			\max_{1\le i \le N} (y_{ik} - \bar{y}_{Nk})^2, & \text{ if } z=1, \\
			0, & \text{ if } z=0, 
		\end{cases}\\
		v_z(k) & \equiv 
        \Var_N\{ [ \bs{A}_z \bs{Y}_i(z) - \bs{A}_z \bar{\bs{Y}}(z) ]_{(k)} \}
		= 
		\begin{cases}
\Var_N(y_{ik}), & \text{ if } z=1, \\
			0, & \text{ if } z=0, 
		\end{cases}\\
		v_{\tau}(k)& \equiv
        \Var_N \{ [\bs\tau_{i}(\bs{A})-\bs{\tau}(\bs{A})]_{(k)} \}
        =
        \Var_N(y_{ik}),
	\end{align*}
	where $[\bs{c}]_{(k)}$ denotes the $k$th coordinate of a vector $\bs{c}$. 
	By definition, we then have
	\begin{align}\label{eq:cond_srs_proof}
		& \quad \ \max _{0 \leq z \leq 1} \max _{1 \leq k \leq K} 
		\frac{1}{n_{z}^{2}} 
		\frac{m_{z}(k)}{\sum_{r=1}^{2} n_{r}^{-1} v_{r}(k)-N^{-1} v_{\tau}(k)} 
		\nonumber
		\\
		& = 
		\max _{1 \leq k \leq K} \frac{1}{n^{2}} 
		\frac{\max _{1 \leq i \leq N}(y_{ik}- \bar{y}_{Nk})^{2}}{(n^{-1}-N^{-1})\Var_N(y_{ik})}
		= 
		\frac{N}{n(N-n)}  \max _{1 \leq k \leq K}
		\frac{\max _{1 \leq i \leq N}( y_{ik}- \bar{y}_{Nk} )^{2}}{\Var_N(y_{ik})}
		\nonumber
		\\
		& \le \frac{2}{\min \{n, N-n\}} \max _{1 \leq k \leq K}
		\frac{\max _{1 \leq i \leq N}( y_{ik}- \bar{y}_{Nk} )^{2}}{\Var_N(y_{ik})}, 
	\end{align}
	where the last equality holds because $N/\{n(N-n)\} = n^{-1} + (N-n)^{-1} \le 2/\min\{n, N-n\}$. 
	Below we consider two cases, depending on whether $\bs{\Sigma}$ is positive definite.

	First, we consider the case where all the diagonal elements of $\bs\Sigma$ is positive. 
	Let $\hat{y}_{Nk}$, $\hat{\tau}_{(k)}(\bs{A})$ and $\tau_{(k)}(\bs{A})$ 
	be the $k$th coordinates of $\hat{\bs{y}}_N,$ $\hat{\bs{\tau}}(\bs{A})$ and $\bs{\tau}(\bs{A})$, respectively, for $1\le k \le K$.  
	From condition (i) in Lemma \ref{lemma:clt_srs}, 
	as $N\rightarrow \infty$, 
	the correlation matrix of $\hat{\bs{\tau}}(\bs{A})$ will converge to $\diag(\bs\Sigma)^{-1/2} \cdot \bs\Sigma \cdot \diag(\bs\Sigma)^{-1/2}$. 
	From condition (ii) in Lemma \ref{lemma:clt_srs}, \eqref{eq:cond_srs_proof} and the fact that $\Var_N(y_{ik}) \rightarrow \Sigma_{kk} > 0$,  
	we can know that the quantity in \eqref{eq:cond_srs_proof} converges to zero as $N\rightarrow \infty$. 
	From Li and Ding (2017, Theorem 4) and by definition, we have
	\begin{align*}
		& \quad \ \left(\frac{\hat{\tau}_{(1)}(\bs{A})-\tau_{(1)}(\bs{A})}{\Var^{1/2}\{\hat{\tau}_{(1)}(\bs{A})\} }, 
		\ldots, 
		\frac{\hat{\tau}_{(K)}(\bs{A})-\tau_{(K)}(\bs{A})}{\Var^{1/2}\{\hat{\tau}_{(K)}(\bs{A})\} }
		\right)
		= \left(
		\frac{\hat{y}_{N1}-\bar{y}_{N1}}{\Var^{1/2}( \hat{y}_{N1})}, 
		\ldots,
		\frac{\hat{y}_{NK}-\bar{y}_{NK}}{\Var^{1/2}(\hat{y}_{NK})}
		\right)
		\\
		& \stackrel{d}{\longrightarrow} 
		\mathcal{N}\left( \bs{0}_{K \times 1}, \ \ \diag(\bs\Sigma)^{-1/2} \cdot \bs\Sigma \cdot \diag(\bs\Sigma)^{-1/2} \right).
	\end{align*}
	By the property of simple random sampling, 
	$
	\Var ( \hat{y}_{Nk} ) 
	= 
	( n^{-1} - N^{-1}) \Var_N( y_{ik} )
	$
	for $1\le k \le K$. 
	From condition (i) in Lemma \ref{lemma:clt_srs}, this implies that 
	\begin{align*}
		\frac{nN}{N-n}\Var\left(\hat{y}_{Nk}\right) = \Var_N\left( y_{ik} \right) \rightarrow \Sigma_{kk}, 
		\qquad (1\le k \le K)
	\end{align*}
	where $\Sigma_{kk}$ is the $k$th diagonal element of $\bs{\Sigma}$. 
	By Slutsky's theorem, we then have 
	\begin{align*}
		\sqrt{\frac{nN}{N-n}} \cdot \left( \hat{\bs{y}}_N - \bar{\bs{y}}_N \right)
		& = 
		\sqrt{\frac{nN}{N-n}} \cdot 
		\diag\left( \Cov\left(\hat{\bs{y}}_N\right) \right)^{1/2}
		\left(
		\frac{\hat{y}_{N1}-\bar{y}_{N1}}{\Var^{1/2}( \hat{y}_{N1})}, 
		\ldots,
		\frac{\hat{y}_{NK}-\bar{y}_{NK}}{\Var^{1/2}(\hat{y}_{NK})}
		\right)^\top\\
		& \converged 
		\mathcal{N}\left( \bs{0}_{K\times 1}, \ \bs{\Sigma} \right). 
	\end{align*}
	
	Second, we consider the case where some diagonal elements of $\bs{\Sigma}$ are zeros. 
	Let $\mathcal{D} = \{k: \Sigma_{kk} > 0, 1\le k \le K\}$, 
	$\mathcal{D}^c = \{1, 2, \ldots, K\} \setminus \mathcal{D}$, 
	and 
	$\bs{\Sigma}_{\mathcal{D}\mathcal{D}}$, $\bs{\Sigma}_{\mathcal{D}^c\mathcal{D}^c}$
	and 
	$\bs{\Sigma}_{\mathcal{D}\mathcal{D}^c} = \bs{\Sigma}_{\mathcal{D}^c\mathcal{D}}^\top$ 
	be the submatrices of $\bs{\Sigma}$ with indices in 
	$\mathcal{D} \times \mathcal{D}$, $\mathcal{D}^c \times \mathcal{D}^c$ 
	and 
	$\mathcal{D} \times \mathcal{D}^c$. 
	Note that by definition, for any $1\leq k\leq K$ and $k' \in \mathcal{D}^c$, 
	$$
	\Sigma_{kk'}^2 = \lim_{N \rightarrow \infty} \Cov_N^2(y_{ik}, y_{ik'}) 
	\le 
	\lim_{N\rightarrow \infty} \Var_N(y_{ik}) \Var_N(y_{ik'}) = \Sigma_{kk} \Sigma_{k'k'} = 0. 
	$$
	Thus, $\bs{\Sigma}_{\mathcal{D}\mathcal{D}^c} = \bs{0}$ and $\bs{\Sigma}_{\mathcal{D}^c\mathcal{D}^c} = \bs{0}$. 
	Let 
	$\hat{\bs{y}}_{N \mathcal{D}}$ and $\bar{\bs{y}}_{N \mathcal{D}}$ be the subvectors of $\hat{\bs{y}}_{N}$ and $\bar{\bs{y}}_{N}$ with indices in $\mathcal{D}$, 
	and 
	$\hat{\bs{y}}_{N \mathcal{D}^c}$ and $\bar{\bs{y}}_{N \mathcal{D}^c}$ be the subvectors of $\hat{\bs{y}}_{N}$ and $\bar{\bs{y}}_{N}$ with indices in $\mathcal{D}^c$. 
	From the previous discussion, 
	we can know that 
	\begin{align*}
		\sqrt{\frac{nN}{N-n}} \cdot \left( \hat{\bs{y}}_{N\mathcal{D}} - \bar{\bs{y}}_{N\mathcal{D}} \right)
		& \converged 
		\mathcal{N}\left( \bs{0}_{K\times 1}, \ \bs{\Sigma}_{\mathcal{D}\mathcal{D}} \right). 
	\end{align*}
	By Chebyshev's inequality, for any $k \in \mathcal{D}^c$, 
	\begin{align*}
		\sqrt{\frac{nN}{N-n}} \cdot \left( \hat{{y}}_{Nk} - \bar{{y}}_{Nk} \right)
		& = 
		\sqrt{\frac{nN}{N-n}} \cdot
		O_{\Pr}\left( 
		\Var^{1/2}
		(
		\hat{{y}}_{Nk} 
		)
		\right)
		= 
		O_{\Pr} \left\{ \Var_N^{1/2}({y}_{ik}) \right\}
		= o_{\Pr}(1). 
	\end{align*}
	These then imply that 
	\begin{align*}
		\sqrt{\frac{nN}{N-n}} \cdot 
		\begin{pmatrix}
			\hat{\bs{y}}_{N\mathcal{D}} - \bar{\bs{y}}_{N\mathcal{D}}
			\\
			\hat{\bs{y}}_{N\mathcal{D}^c} - \bar{\bs{y}}_{N\mathcal{D}^c}
		\end{pmatrix}
		\converged
		\mathcal{N}
		\left(
		\bs{0}, \ 
		\begin{pmatrix}
			\bs{\Sigma}_{\mathcal{D}\mathcal{D}} & \bs{0} \\
			\bs{0} & \bs{0}
		\end{pmatrix}
		\right)
		\sim 
		\mathcal{N}
		\left(
		\bs{0}, \ 
		\begin{pmatrix}
			\bs{\Sigma}_{\mathcal{D}\mathcal{D}} & \bs{\Sigma}_{\mathcal{D}\mathcal{D}^c}  \\
			\bs{\Sigma}_{\mathcal{D}^c\mathcal{D}}  & \bs{\Sigma}_{\mathcal{D}^c\mathcal{D}^c} 
		\end{pmatrix}
		\right).
	\end{align*}
	By reordering the indices, we then have 
	\begin{align*}
		\sqrt{\frac{nN}{N-n}} \cdot \left( \hat{\bs{y}}_N - \bar{\bs{y}}_N \right)
		& \converged 
		\mathcal{N}\left( \bs{0}_{K\times 1}, \ \bs{\Sigma} \right). 
	\end{align*}
	
	From the above, Lemma \ref{lemma:clt_srs} holds. 
\end{proof}

\begin{proof}[of Theorem \ref{thm:clt}]
	By the same logic as  the proof of Theorem \ref{thm:rate}, 
	without loss of generality, we can assume that 
	$\hat{\bs{\theta}}_N \in \overline{\mathcal{B}}(\bs{\theta}_N, \varepsilon_0/2) \equiv \{\theta: \|\theta-\theta_N\|\le \varepsilon_0/2\} \subset \Theta$ for all $N$. 
	From \eqref{eq:taylor_M1}, \eqref{eq:H_inverse_limit} and \eqref{eq:theta_hat_taylor} in the proof of Theorem \ref{thm:rate},  
	\begin{align*}
		\hat{{\bs\theta}}_N-{\bs\theta}_N & = 
		\bs{H}^{-1}\hat{\Phi}_N(\hat{\bs\theta}_N) -  \bs{H}^{-1}\hat{\Phi}_N ({\bs\theta}_N)
		\\
		& =
		\bs{H}^{-1}
		\hat{\Phi}_N(\hat{\bs\theta}_N)
		-  \left[ \bs{H}^{-1} -  \dot{\Phi}_N({\bs\theta}_N)^{-1} \right]\hat{\Phi}_N ({\bs\theta}_N)
		- \dot{\Phi}_N({\bs\theta}_N)^{-1} \hat{\Phi}_N ({\bs\theta}_N)
		\\
        & = 
        - \dot{\Phi}_N({\bs\theta}_N)^{-1} \hat{\Phi}_N ({\bs\theta}_N) + \bs{H}^{-1}
		\hat{\Phi}_N(\hat{\bs\theta}_N) +o_{\Pr}(  1 ) \cdot \hat{\Phi}_N ({\bs\theta}_N).
	\end{align*}
	From conditions (i) and (ii) and 
	by Lemma \ref{lemma:clt_srs}, 
	we have
	\begin{align*}
		\sqrt{\frac{nN}{N-n}} \cdot \hat{\Phi}_N({\bs\theta}_N) 
		& = \sqrt{\frac{nN}{N-n}} \cdot \big\{ \hat{\Phi}_N({\bs\theta}_N) - \Phi_N({\bs\theta}_N) \big\}
		\\
		& = 
		\sqrt{\frac{nN}{N-n}} \cdot \Big\{ \frac{1}{n}\sum_{i=1}^N Z_i \phi_i(\bs{\theta}_N) - \frac{1}{N}\sum_{i=1}^N \phi_i(\bs{\theta}_N) \Big\}
		\\
		& 
		\converged \mathcal{N}\left(\bs{0}, \ 
		\bs{\Sigma} \right). 
	\end{align*}
    From the discussion at the beginning of the proof of Theorem \ref{thm:rate},  
    $\Pr\{\hat{\Phi}_N(\hat{\bs\theta}_N) \ne 0\} = o(1)$, 
    which then implies that
	$ \sqrt{{nN}/(N-n)} \cdot \bs{H}^{-1} \hat{\Phi}_N(\hat{\bs\theta}_N) =  o_{\Pr}(1)$.
	From condition (i), we have $\dot{\Phi}_N({\bs\theta}_N) \rightarrow \bs{\Gamma} > 0$. 
	From the above and by Slutsky's theorem,
	\begin{align*}
		\sqrt{\frac{nN}{N-n}} \cdot ( \hat{{\bs\theta}}_N-{\bs\theta}_N ) 
		& \converged \mathcal{N}\left( 
		\bs{0}, \ 
		\bs{\Gamma}^{-1} \bs{\Sigma}  (\bs{\Gamma}^{-1})^\top. 
		\right),
	\end{align*}
	Therefore, Theorem \ref{thm:clt} holds. 
\end{proof}

\section{
Uniform convergence and consistency for 
M-estimation and Z-estimation under completely randomized experiments}\label{sec:M_cre}

In this section, we consider M-estimation and Z-estimation under the CRE. 
We establish the uniform convergence of the empirical risk function and the estimating equation, and the consistency of the M-estimator and Z-estimator. 
Similar to \S \ref{sec:M_srs}, we will consider  two types of sufficient conditions, based on either the compactness of the parameter space or the convexity of the risk function. 
The results in this section are obtained using theorems in \S \ref{sec:M_srs}. 
In parallel to \S \ref{sec:M_srs}, we present the results in four subsections, and, within each subsection, we first present the theorems and then their proofs. 

\subsection{M-estimation with compactness}

\begin{theorem}\label{thm:uniform_converge_sp}
	Consider a finite population of $N$ units, where each unit $i$ has two potential outcomes $(Y_i(1), Y_i(0))$ and a pretreatment covariate vector $\bs{X}_i$. 
	Let 
	$(Z_1, Z_2, \ldots, Z_N) \in \{0,1\}^N$ be a random vector 
	whose probability of taking value $(z_1, z_2, \ldots, z_N)$ is $n_1!n_0!/N!$ if $\sum_{i=1}^N z_i = n_1$ and zero otherwise, 
	where $n_1$ and $n_0 = N-n$ are two fixed constants. 
    We assume both $r_1 \equiv n_1/N$ and $r_0 \equiv n_0/N$ have positive limits.
	Let $\loss_1(\bs{y}, \bs{x}; \bs{\theta})$ and $\loss_0(\bs{y}, \bs{x}; \bs{\theta})$ be two general loss functions, 
	with $\theta$ in a set $\Theta \subset \mathbb{R}^p$.
    For each unit $i$, define $\loss_{zi}(\theta) = \loss_z(Y_i(z), \bs X_i;{\bs\theta})$ for $z=0,1$, and $\loss_{1\text{-}0, i}(\bs{\theta}) = 
        \loss_{1i}(\bs{\theta}) - \loss_{0i}(\bs{\theta})$.
    Define further 
    $$
	\Delta( \loss_{1\text{-}0,i}, \delta)=\sup _{(\bs\theta, \bs\theta') \in \Theta^2:\|\bs\theta-\bs\theta'\| \leq  \delta}| \loss_{1\text{-}0, i}({\bs\theta})
	-
	\loss_{1\text{-}0, i}({\bs\theta}') 
	|. 
	$$
    Let 
    \begin{align*}
	\hat{M}_N (\bs{\theta})
	& \equiv 
    \E_N\{
    \loss_{Z_i} (Y_i, \bs{X}_i; \bs{\theta})
    \}
	=
	r_1 \E_{N}^1 \{\loss_1(Y_i, \bs X_i;{\bs\theta})\}
	+
	r_0 \E_{N}^0 \{\loss_0(Y_i, \bs X_i;{\bs\theta})\}
	\\
    & = \frac{1}{n_1}\sum_{i=1}^N Z_i r_1 \loss_{1\text{-}0, i}(\bs{\theta}) + \E_N\{ \loss_{0i} ({\bs\theta}) \}, 
    \end{align*}
    and 
    \begin{align}\label{eq:M_population}
    M_N (\bs{\theta})
	& \equiv
	r_1 \E_N\{\loss_1(Y_i(1), \bs X_i;{\bs\theta})\}
	+r_0 \E_N\{\loss_0 (Y_i(0), \bs X_i; {\bs\theta})\}\\
    & = 
    \frac{1}{N} \sum_{i=1}^N r_1 \loss_{1\text{-}0, i}(\bs{\theta}) + \E_N\{ \loss_{0i} ({\bs\theta}) \}, 
    \nonumber
    \end{align}
    where the equivalent forms follow from some algebra. 	
	If the following conditions hold as $N\rightarrow \infty$: 
	\begin{enumerate}[label=(\roman*)]
            \item $\Theta$ is compact,
		\item $\sup_N \E_N\{ \Delta( \loss_{1\text{-}0,i}, \delta) \} \rightarrow 0$ as $\delta \rightarrow 0$, 
		\item for any ${\bs\theta}\in \Theta$, $ \Var_N
		\left\{ \loss_{1\text{-}0,i}(\bs{\theta})
		\right\}=o(N)$,
	\end{enumerate}
	then $\sup_{\bs{\theta} \in \Theta} |\hat{M}_N(\bs{\theta}) - M_N(\bs{\theta})| = o_{\Pr}(1)$ as $N\rightarrow \infty$. 
\end{theorem}

\begin{theorem}\label{thm:consist_sp}
	Under the same setting as in Theorem \ref{thm:uniform_converge_sp}, 
	if conditions (i)--(iii) in Theorem \ref{thm:uniform_converge_sp} hold, 
	and the following conditions hold:  
	\begin{itemize}
		\item[(i)]  $M_N(\bs{\theta})$ has a unique minimizer $\bs{\theta}_N$ satisfying that for any $\varepsilon>0$, there exists $\eta>0$ such that
		$\inf_{\bs\theta\in\Theta: \|\bs\theta, \bs\theta_{N}\| \geq \varepsilon} M_N(\bs\theta)>M_N\left(\bs\theta_{N}\right)+\eta$ for all $N$,
		\item[(ii)] $\hat{\bs\theta}_N\in\Theta$ satisfies  $\hat{M}_N(\hat{\bs\theta}_N)\leq \hat{M}_N({\bs\theta}_N)$, 
	\end{itemize}
	then $\hat{\bs\theta}_N-{\bs\theta}_N = o_{\Pr}(1)$ as $N \rightarrow \infty$.
\end{theorem}

\begin{proof}[of Theorem \ref{thm:uniform_converge_sp}]
	We prove Theorem \ref{thm:uniform_converge_sp} using Theorem \ref{thm:uniform_converge}. 
	Following the notation in Theorem \ref{thm:uniform_converge}, 
	we define $n = n_1$,
	$m_i({\bs\theta})=r_1  \loss_{1\text{-}0, i}(\bs{\theta})$,  and $\tilde m(\bs\theta)= \E_N\{ \loss_{0i} ({\bs\theta}) \}$. 
	Then
	$
	\Delta(m_i, \delta) 
    =r_1\Delta(\loss_{1\text{-}0,i}, \delta). 
	$
	Below we show that the conditions (ii) and (iii) in Theorem \ref{thm:uniform_converge} hold. 
	
	First, 
	because 
	$\E_N\{ \Delta(m_i, \delta) \}= r_1 \E_N\{ \Delta(\loss_{1\text{-}0,i}, \delta) \}$, 
	condition (ii) in Theorem \ref{thm:uniform_converge} follows immediately from condition (ii) in Theorem \ref{thm:uniform_converge_sp}. 
	Second, because for any $\bs{\theta}\in \Theta$, 
	\begin{align*}
		n^{-1} \Var_{N} \left\{ m_i(\bs\theta) \right\}
		& = 
		n_1^{-1} \Var_{N} \left\{ r_1  \loss_{1\text{-}0, i}(\bs{\theta}) \right\}
		= 
		r_1 \cdot N^{-1} \Var_{N} \left\{ \loss_{1\text{-}0, i}(\bs{\theta}) \right\},  
	\end{align*}
	condition (iii) in Theorem \ref{thm:uniform_converge} follows immediately from 
	condition (iii) in Theorem \ref{thm:uniform_converge_sp}. 
	
	From the above and Theorem \ref{thm:uniform_converge}, 
	$\sup_{\bs{\theta} \in \Theta} |\hat{M}_N(\bs{\theta}) - M_N(\bs{\theta})| = o_{\Pr}(1)$, i.e., Theorem \ref{thm:uniform_converge_sp} holds. 
\end{proof}

\begin{proof}[of Theorem \ref{thm:consist_sp}]
	We prove Theorem \ref{thm:consist_sp} using Theorem \ref{thm:consistency}. 
	From Theorem \ref{thm:uniform_converge_sp} and following the notation in the proof of Theorem \ref{thm:uniform_converge_sp}, 
	we can know that conditions (i)--(iii) in Theorem \ref{thm:uniform_converge} hold. 
	It suffices to verify conditions (i) and (ii) in Theorem \ref{thm:consistency}.  
	Note that conditions (i) and (ii) in Theorem \ref{thm:consistency} follow immediately from conditions (i) and (ii) in Theorem \ref{thm:consist_sp}.  
	Therefore, 
	$\hat{\bs{\theta}}_N - \bs{\theta}_N = o_{\Pr}(1)$, i.e., Theorem \ref{thm:consist_sp} holds. 
\end{proof}

\subsection{Z-estimation with compactness} 

\begin{theorem}\label{thm:uniform_converge_ZCP_sp}
	Consider a finite population of $N$ units, where each unit $i$ has two potential outcomes $(Y_i(1), Y_i(0))$ and a pretreatment covariate vector $\bs{X}_i$. 
	Let 
	$(Z_1, Z_2, \ldots, Z_N) \in \{0,1\}^N$ be a random vector 
	whose probability of taking value $(z_1, z_2, \ldots, z_N)$ is $n_1!n_0!/N!$ if $\sum_{i=1}^N z_i = n_1$ and zero otherwise, 
	where $n_1$ and $n_0 = N-n$ are two fixed constants. 
        We assume both $r_1 \equiv n_1/N$ and $r_0 \equiv n_0/N$ have positive limits.
	Let $\psi_1(\bs{y}, \bs{x}; \bs{\theta})\in \mathbb{R}^p$ and $\psi_0(\bs{y}, \bs{x}; \bs{\theta})\in \mathbb{R}^p$ be two general functions, 
	with $\theta$ in a set $\Theta \subset \mathbb{R}^p$.  
    For each unit $i$, define $\psi_{zi}(\theta) = \psi_z(Y_i(z), \bs X_i;{\bs\theta})$ for $z=0,1$, and $\psi_{1\text{-}0, i}(\bs{\theta}) = 
    \psi_{1i}(\bs{\theta}) - \psi_{0i}(\bs{\theta})$.
    Define further 
    $$
	\Delta( \psi_{1\text{-}0,i}, \delta)=\sup _{(\bs\theta, \bs\theta') \in \Theta^2:\|\bs\theta-\bs\theta'\| \leq  \delta}
    \| \psi_{1\text{-}0, i}({\bs\theta})
	-
	\psi_{1\text{-}0, i}({\bs\theta}') 
	\|. 
	$$
    Let 
    \begin{align*}
	\hat{\Psi}_N (\bs{\theta})
	& \equiv 
    \E_N\{
    \psi_{Z_i} (Y_i, \bs{X}_i; \bs{\theta})
    \}
	=
	r_1 \E_{N}^1 \{\psi_1(Y_i, \bs X_i;{\bs\theta})\}
	+
	r_0 \E_{N}^0 \{\psi_0(Y_i, \bs X_i;{\bs\theta})\}
	\\
    & 
    = \frac{1}{n_1}\sum_{i=1}^N Z_i r_1 \psi_{1\text{-}0, i}(\bs{\theta}) + \E_N\{ \psi_{0i} ({\bs\theta}) \}, 
    \end{align*}
    and 
    \begin{align*}
    \Psi_N (\bs{\theta})
	& \equiv
	r_1 \E_N\{\psi_1(Y_i(1), \bs X_i;{\bs\theta})\}
	+r_0 \E_N\{\psi_0 (Y_i(0), \bs X_i; {\bs\theta})\}\\
    & = 
    \frac{1}{N} \sum_{i=1}^N r_1 \psi_{1\text{-}0, i}(\bs{\theta}) + \E_N\{ \psi_{0i} ({\bs\theta}) \}, 
    \end{align*}
    where the equivalent forms follow from some algebra. 
    	If the following conditions hold as $N\rightarrow \infty$: 
    	\begin{enumerate}[label=(\roman*)]
                \item $\Theta$ is compact,
    		\item $\sup_N \E_N\{ \Delta( \psi_{1\text{-}0,i}, \delta) \} \rightarrow 0$ as $\delta \rightarrow 0$, 
    		\item for any ${\bs\theta}\in \Theta$ and $1\le k \le p$, $ \Var_N
    		\left\{ \psi_{1\text{-}0,ik}(\bs{\theta})
    		\right\}=o(N)$, 
            where $ \psi_{1\text{-}0,ik}(\bs{\theta})$ denotes the $k$th coordinate of $ \psi_{1\text{-}0,i}(\bs{\theta})$, 
    	\end{enumerate}
    	then $\sup_{\bs\theta\in\Theta}\|\hat{\Psi}_N(\bs\theta)-\Psi_N(\bs\theta)\| = o_{\Pr}(1)$ as $N\rightarrow \infty$.  
\end{theorem}

\begin{theorem}\label{thm:consist_ZCP_sp}
	Under the same setting as in Theorem \ref{thm:uniform_converge_ZCP_sp}, if conditions (i)--(iii) in Theorem \ref{thm:uniform_converge_ZCP_sp} holds, 
	and the following conditions hold:
	\begin{itemize}
		\item[(i)] $\bs{\theta}_N$ is the unique root of $\Psi_N(\bs{\theta})$  satisfying that for any $\varepsilon>0$, there exists $\eta>0$ such that
		$\inf_{\bs\theta\in\Theta: \|\bs\theta-\bs\theta_{N}\| \geq \varepsilon} \|\Psi_N(\bs\theta)\|>\eta$ for all $N$,
		\item[(ii)] $\hat{\bs{\theta}}_N \in \Theta$ 
		satisfies $\hat{\Psi}_N(\hat{\bs\theta}_N) = 0$, 
	\end{itemize}
	then $\hat{\bs\theta}_N-\bs\theta_N = o_{\Pr}(1)$ as $N \rightarrow \infty$. 
\end{theorem}

\begin{proof}[of Theorem \ref{thm:uniform_converge_ZCP_sp}]
	We prove Theorem \ref{thm:uniform_converge_ZCP_sp} using Theorem \ref{thm:uniform_converge_ZCP}. 
	Following the notation in Theorem \ref{thm:uniform_converge_ZCP}, 
	we define $n = n_1$,
	$\phi_i({\bs\theta})=r_1  \psi_{1\text{-}0, i}(\bs{\theta})$,  and $\tilde \phi(\bs\theta)= \E_N\{ \psi_{0i} ({\bs\theta}) \}$. 
	Then
	$
	\Delta(\phi_i, \delta) 
    =r_1\Delta(\psi_{1\text{-}0,i}, \delta). 
	$
	Below we show that conditions (i) and (ii) in Theorem \ref{thm:uniform_converge_ZCP} hold. 
	
	First, 
	because 
	$\E_N\{ \Delta(\phi_i, \delta) \}= r_1 \E_N\{ \Delta(\psi_{1\text{-}0,i}, \delta) \}$, 
	condition (ii) in Theorem \ref{thm:uniform_converge_ZCP} follows immediately from condition (ii) in Theorem \ref{thm:uniform_converge_ZCP_sp}. 
	Second, because for any $\bs{\theta}\in \Theta$ and $1\le k \le p$, 
	\begin{align*}
		n^{-1} \Var_{N} \left\{ \phi_{ik}(\bs\theta) \right\}
		& = 
		n_1^{-1} \Var_{N} \left\{ r_1  \psi_{1\text{-}0, ik}(\bs{\theta}) \right\}
		= 
		r_1 \cdot N^{-1} \Var_{N} \left\{ \psi_{1\text{-}0, ik}(\bs{\theta}) \right\},  
	\end{align*}
	condition (iii) in Theorem \ref{thm:uniform_converge_ZCP} follows immediately from 
	condition (iii) in Theorem \ref{thm:uniform_converge_ZCP_sp}. 
	
	From the above and Theorem \ref{thm:uniform_converge_ZCP}, 
	$\sup_{\bs{\theta} \in \Theta} \|\hat{\Psi}_N(\bs{\theta}) - \Psi_N(\bs{\theta})\| = o_{\Pr}(1)$, i.e., Theorem \ref{thm:uniform_converge_ZCP_sp} holds. 
\end{proof}

\begin{proof}[of Theorem \ref{thm:consist_ZCP_sp}]
	We prove Theorem \ref{thm:consist_ZCP_sp} using Theorem \ref{thm:consistency_ZCP}. 
	From Theorem \ref{thm:uniform_converge_ZCP_sp} and following the notation in the proof of Theorem \ref{thm:uniform_converge_ZCP_sp}, 
	we can know that conditions (i)--(iii) in Theorem \ref{thm:uniform_converge_ZCP} hold. 
	It suffices to verify conditions (i) and (ii) in Theorem \ref{thm:consistency_ZCP}.  
	Note that conditions (i) and (ii) in Theorem \ref{thm:consistency_ZCP} follow immediately from conditions (i) and (ii) in Theorem \ref{thm:consist_ZCP_sp}.  
	Therefore, 
	$\hat{\bs{\theta}}_N - \bs{\theta}_N = o_{\Pr}(1)$, i.e., Theorem \ref{thm:consist_ZCP_sp} holds. 
\end{proof}

\subsection{M-estimation with convexity}

\begin{theorem}\label{thm:uniform_converge_MCV_sp}
    Under the same setting as in Theorem \ref{thm:uniform_converge_sp}, if the following conditions hold:
    \begin{itemize}
        \item[(i)] there exists $r > 0$ such that $\mathcal{B}(\bs\theta_N, r) \subseteq \Theta$ for all $N$, where $\{\bs\theta_N\}$ is a sequence from $\Theta$,
        \item[(ii)] there exists $L \geq 0$ such that $(\E_N [ \{ \loss_{1\text{-}0,i}(\bs\theta) - \loss_{1\text{-}0,i}(\bs\theta') \}^2 ])^{1/2} \leq L \|\bs\theta - \bs\theta'\|$ for all $N$ and $\bs\theta, \bs\theta' \in \mathcal{B}(\bs\theta_N, r)$,
    \end{itemize}
    then, as $N \to \infty$, 
    \begin{equation*}
        \sup_{\bs\theta \in \Theta:\|\bs\theta-\bs\theta_N\| \leq r} \big| \{ \hat{M}_N(\theta) - M_N(\theta) \} - \{ \hat{M}_N(\theta_N) - M_N(\theta_N) \} \big| = o_{\Pr}(1).
    \end{equation*}
    If, in addition, 
    \begin{itemize}
        \item[{(iii)}]
        $var_N\{\loss_{1\text{-}0,i}(\bs\theta_N)\}= o(N) $,
    \end{itemize}
    then, as $N \to \infty$, 
    $$
    \sup_{\bs\theta \in \Theta:\|\bs\theta-\bs\theta_N\| \leq r} \big|  \hat{M}_N(\theta) - M_N(\theta) \big| = o_{\Pr}(1).
    $$
    
\end{theorem}

\begin{theorem}\label{thm:consist_MCV_sp}
    Under the same setting as in Theorem \ref{thm:uniform_converge_sp},if conditions (i) and (ii) in Theorem \ref{thm:uniform_converge_MCV_sp} holds, 
	and the following conditions hold:
    \begin{itemize}
        \item[(i)] $\Theta$ is convex
        \item[(ii)] $\loss_{1\text{-}0,i}(\bs\theta)$, $\loss_{0,i}(\bs\theta)$ are convex for all $i$,
	\item[(iii)] 
        $\bs{\theta}_N$ is the unique minimizer of $M_N(\bs{\theta})$  satisfying that for any $\varepsilon>0$, there exists $\eta>0$ such that
		$\inf_{\bs\theta\in\Theta: \|\bs\theta-\bs\theta_{N}\| =\varepsilon} M_N(\bs\theta)>M_N\left(\bs\theta_{N}\right)+\eta$ for all $N$,
	\item[(iv)] $\hat{\bs{\theta}}_N \in \Theta$ 
		satisfies $\hat{M}_N(\hat{\bs\theta}_N)\leq \hat{M}_N(\bs{\theta}_N)$, 
	\end{itemize}
	then $\hat{\bs\theta}_N-\bs\theta_N = o_{\Pr}(1)$ as $N \rightarrow \infty$. 
\end{theorem}

\begin{proof}[of Theorem \ref{thm:uniform_converge_MCV_sp}]

	We prove Theorem \ref{thm:uniform_converge_MCV_sp} using Theorem \ref{thm:uniform_converge_MCV}. 
	Following the notation in Theorem \ref{thm:uniform_converge_MCV}, 
	we define $n = n_1$,
	$m_i({\bs\theta})=r_1  \loss_{1\text{-}0, i}(\bs{\theta})$,  and $\tilde m(\bs\theta)= \E_N\{ \loss_{0i} ({\bs\theta}) \}$. 
	Below we show conditions (ii) and (iii) in Theorem \ref{thm:uniform_converge_MCV} can be implied by conditions (ii) and (iii) in  Theorem \ref{thm:uniform_converge_MCV_sp}, respectively. 
	
    First,
	because 
	$\E_N[\{ m_i(\bs\theta) - m_i(\bs\theta') \}^2]= r_1^2  \E_N[\{ l_{1\text{-}0,i}(\bs\theta) - l_{1\text{-}0,i}(\bs\theta') \}^2]$, 
	condition (ii) in Theorem \ref{thm:uniform_converge_MCV} follows immediately from condition (iii) in Theorem \ref{thm:uniform_converge_MCV_sp}. 
	Second, because
	\begin{align*}
		n^{-1} \Var_{N} \left\{ m_i(\bs\theta_N) \right\}
		& = 
		n_1^{-1} \Var_{N} \left\{ r_1  \loss_{1\text{-}0, i}(\bs{\theta_N}) \right\}
		= 
		r_1 \cdot N^{-1} \Var_{N} \left\{ \loss_{1\text{-}0, i}(\bs{\theta}) \right\},  
	\end{align*}
	condition (iii) in Theorem \ref{thm:uniform_converge_MCV} follows immediately from 
	condition (iii) in Theorem \ref{thm:uniform_converge_MCV_sp}. 
	
	From the above and Theorem \ref{thm:uniform_converge_MCV}, 
	Theorem \ref{thm:uniform_converge_MCV_sp} holds. 
\end{proof}

\begin{proof}[of Theorem \ref{thm:consist_MCV_sp}]
	We prove Theorem \ref{thm:consist_MCV_sp} using Theorem \ref{thm:consistency_MCV}. 
	From Theorem \ref{thm:uniform_converge_MCV_sp} and following the notation in the proof of Theorem \ref{thm:uniform_converge_MCV_sp}, 
	we can know that conditions (i) and (ii) in Theorem \ref{thm:uniform_converge_MCV} hold. 
	It suffices to verify conditions (i)-(iv) in Theorem \ref{thm:consistency_MCV}.  
	Note that conditions (i)-(iv) in Theorem \ref{thm:consistency_MCV} follow immediately from condition (i)-(iv) in Theorem \ref{thm:consist_MCV_sp}.  
	Therefore, 
	$\hat{\bs{\theta}}_N - \bs{\theta}_N = o_{\Pr}(1)$, i.e., Theorem \ref{thm:consist_MCV_sp} holds. 
\end{proof}

\subsection{Z-estimation with convexity}

\begin{theorem}\label{thm:uniform_converge_ZCV_sp}
    Under the same setting as in Theorem \ref{thm:uniform_converge_ZCP_sp}, if the following conditions hold:
    \begin{itemize}
        \item[(i)] there exists $r > 0$ such that $\mathcal{B}(\theta_N, r) \subseteq \Theta$ for all $N$, where $\{\bs\theta_N\}$ is a sequence from $\Theta$,
        \item[(ii)] there exists $L \geq 0$ such that $(\E_N \left[ \| \psi_{1\text{-}0,i}(\bs\theta) - \psi_{1\text{-}0,i}(\bs\theta') \|^2 \right])^{1/2} \leq L \|\bs\theta - \bs\theta'\|$ for all $N$ and $\bs\theta, \bs\theta' \in \mathcal{B}(\bs\theta_N, r)$,
        \item[(iii)] $var_N\{\psi_{1\text{-}0,ik}(\bs\theta_N)\}$ is $o(N)$ for each coordinate k,
    \end{itemize}
    then 
    $$
    \sup_{\bs\theta \in \Theta:\|\bs\theta-\bs\theta_N\| \leq r} \| \hat{\Psi}_N(\theta) - \Psi_N(\theta)  \| = o_{\Pr}(1) \text{ as } N \to \infty.
    $$
\end{theorem}

\begin{theorem}\label{thm:consist_ZCV_sp}
    Under the same setting as in Theorem \ref{thm:uniform_converge_ZCP_sp},if conditions (i)--(iii) in Theorem \ref{thm:uniform_converge_ZCV_sp} holds, 
	and the following conditions hold:
    \begin{itemize}
        \item[(i)] $\Theta$ is convex, 
        \item[(ii)] For all $i$ and $\bs\theta, \bs\theta' \in \Theta$, $\psi_{1\text{-}0,i}(\bs\theta)$ and $\psi_{0i}(\bs\theta)$ satisfy 
        $$
        \{\psi_{1\text{-}0,i}(\bs\theta) - \psi_{1\text{-}0,i}(\bs\theta')\}^\top(\bs\theta - \bs\theta') \geq 0 \text{ and } 
        \{ \psi_{0i}(\bs\theta) - \psi_{0i}(\bs\theta') \}^\top(\bs\theta - \bs\theta') \geq 0, 
         $$
	\item[(iii)] 
        $\bs{\theta}_N$ is the unique root of $\Psi(\bs{\theta})$  satisfying that for any $\varepsilon>0$, there exists $\eta>0$ such that
		$\inf_{\bs\theta\in\Theta: \|\bs\theta-\bs\theta_{N}\| = \varepsilon} \Psi_N(\bs\theta)^\top(\bs\theta-\bs\theta_N) > \eta$ for all $N$,
	\item[(iv)] $\hat{\bs{\theta}}_N \in \Theta$ 
		satisfies $\hat{\Psi}_N(\hat{\bs\theta}_N) = 0$, 
	\end{itemize}
	then $\hat{\bs\theta}_N-\bs\theta_N = o_{\Pr}(1)$ as $N \rightarrow \infty$. 
\end{theorem}

\begin{proof}[of Theorem \ref{thm:uniform_converge_ZCV_sp}]
	We prove Theorem \ref{thm:uniform_converge_ZCV_sp} using Theorem \ref{thm:uniform_converge_ZCV}. 
	Following the notation in Theorem \ref{thm:uniform_converge_ZCV}, 
	we define $n = n_1$,
	$\phi_i({\bs\theta})=r_1  \psi_{1\text{-}0, i}(\bs{\theta})$,  and $\tilde \phi(\bs\theta)= \E_N\{ \psi_{0i} ({\bs\theta}) \}$. 
	Below we show that conditions (ii) and (iii) in Theorem \ref{thm:uniform_converge_ZCV} hold. 
	
	First, 
	because 
	$\E_N[\{ \phi_i(\bs\theta) - \phi_i(\bs\theta') \}^2]= r_1^2  \E_N[\{ \psi_{1\text{-}0,i}(\bs\theta) - \psi_{1\text{-}0,i}(\bs\theta') \}^2]$, 
	condition (ii) in Theorem \ref{thm:uniform_converge_ZCV} follows immediately from condition (ii) in Theorem \ref{thm:uniform_converge_ZCV_sp}. 
	Second, because, for any $1\le k \le p$, 
	\begin{align*}
		n^{-1} \Var_{N} \left\{ \phi_{ik}(\bs\theta_N) \right\}
		& = 
		n_1^{-1} \Var_{N} \left\{ r_1  \psi_{1\text{-}0, ik}(\bs{\theta_N}) \right\}
		= 
		r_1 \cdot N^{-1} \Var_{N} \left\{ \psi_{1\text{-}0, ik}(\bs{\theta_N}) \right\},  
	\end{align*}
	condition (iii) in Theorem \ref{thm:uniform_converge_ZCV} follows immediately from 
	condition (iii) in Theorem \ref{thm:uniform_converge_ZCV_sp}. 
	
	From the above and Theorem \ref{thm:uniform_converge_ZCV}, 
	$\sup_{\bs{\theta} \in \Theta} \|\hat{\Psi}_N(\bs{\theta}) - \Psi_N(\bs{\theta})\| = o_{\Pr}(1)$, i.e., Theorem \ref{thm:uniform_converge_ZCV_sp} holds. 
\end{proof}

\begin{proof}[of Theorem \ref{thm:consist_ZCV_sp}]
	We prove Theorem \ref{thm:consist_ZCV_sp} using Theorem \ref{thm:consistency_ZCV}. 
	From Theorem \ref{thm:uniform_converge_ZCV_sp} and following the notation in the proof of Theorem \ref{thm:uniform_converge_ZCV_sp}, 
	we can know that conditions (i)--(iii) in Theorem \ref{thm:uniform_converge_ZCV} hold. 
	It suffices to verify conditions (i)-(iv) in Theorem \ref{thm:consistency_ZCV}.  
	Note that conditions (i)-(iv) in Theorem \ref{thm:consistency_ZCV} follow immediately from condition (i)-(iv) in Theorem \ref{thm:consist_ZCV_sp}.  
	Therefore, 
	$\hat{\bs{\theta}}_N - \bs{\theta}_N = o_{\Pr}(1)$, i.e., Theorem \ref{thm:consist_ZCV_sp} holds. 
\end{proof}

\section{Asymptotic Normality for Z-estimation under completely randomized experiments}\label{sec:clt_cre}

In this section, we establish the $\sqrt{N}$-consistency and asymptotic Normality of the Z-estimator under the CRE , which includes the M-estimator as special cases. 
For conciseness, 
we impose conditions on uniform convergence of $\dot{\hat{\Psi}}_N(\bs\theta)$ and the consistency of $\hat{\theta}_N$, which can be guaranteed by conditions similar to those in \S \ref{sec:M_cre}.
The results in this section are obtained using theorems in \S \ref{sec:clt_srs}. 

\subsection{Theorems}

\begin{theorem}\label{thm:rate_sp}
	Under the same setting as in Theorem \ref{thm:uniform_converge_ZCP_sp}, 
	for $z=0,1$ and $1\le i\le N$, we further define $\psi_{zi}(\theta)$ and $\Delta(\dot{\psi}_{zi}, \delta)$ the same as that in \S \ref{sec:fp_asym}. 
	If conditions (i)--(iii) in Theorem \ref{thm:uniform_converge_sp} hold, 
	conditions (i) and (ii) in Theorem \ref{thm:consist_sp} hold, 
	and the following conditions hold:  
 	\begin{enumerate}[label=(\roman*)]
        \item 
        $\theta_N$ is the root of $\Psi_N(\theta)$, and 
		there exists $\varepsilon_0>0$ such that $\mathcal{B}({\bs\theta}_N, \varepsilon_0 ) = \{\theta: \|\theta-\theta_N\| <  \varepsilon_0\} \subset \Theta$ for all $N$, 
 
        \item
		$\hat{\Psi}_N(\hat{\bs\theta}_N)=0$, 
        and $\hat{\bs\theta}_N-\bs\theta_N = o_{\Pr}(1)$ as $N \rightarrow \infty$,

		\item for all $N$ and $1\le i \le N$ and $z=0,1$, 
		$\psi_{zi}({\bs\theta})$  have continuous first order partial derivatives over ${\bs\theta} \in \Theta$,

        \item there exists $r_0 > 0$ such that $\sup_{\bs\theta \in \Theta:\|\bs\theta-\bs\theta_N\|\leq r_0} \| \dot{\hat{\Psi}}_N(\bs\theta)  - \dot{\Psi}_N(\bs\theta) \| = o_{\Pr}(1)$,
         
		\item
		$\sup_N \sup_{\bs\theta \in \Theta:\|\bs\theta-\bs\theta_N\|\leq\delta} \| \dot{\Psi}_N(\bs\theta)  - \dot{\Psi}_N(\bs\theta_N) \| \rightarrow 0$ 
		as $\delta \rightarrow 0$, 
        \item 
        $var_N\{\psi_{1\text{-}0,ik}(\bs\theta_N)\}= O(1) $ for $1 \leq k \leq p$,
		
		\item
		there exists $c_0>0$ such that 
        $\sigma_{\min}(\dot{\Psi}_N({\bs\theta}_N)) \ge c_0$
        for all $N$,
	\end{enumerate}
	then $\hat{\bs\theta}_N-{\bs\theta}_N=O_\Pr( N^{-1/2} )$ as $N\rightarrow \infty$.
\end{theorem}

\begin{theorem}\label{thm:clt_sp}
	Under the same setting as in Theorem \ref{thm:uniform_converge_ZCP_sp}, 
	if  conditions (i)--(vii) in Theorem \ref{thm:rate_sp} hold, 
	and the following conditions hold as $N\rightarrow \infty$: 
	\begin{itemize}
		\item[(i)] $
        \dot{\Psi}_N(\theta_N)
        \rightarrow \bs{\Gamma}$ with $\sigma_{\min}(\bs{\Gamma}) > 0$, 
		and 
		$
		\Cov_N \{ \psi_{1\text{-}0,i}(\bs{\theta}_N) \} \rightarrow \bs{\Sigma}, 
		$
		\item[(ii)] 
		$
		\max_{1\le j \le N} \| \psi_{1\text{-}0, j}(\bs{\theta}_N) - \E_N\{\psi_{1\text{-}0, i}(\bs{\theta}_N)\} \|^2 
		/N \rightarrow 0, 
		$
	\end{itemize}
	then 
	$
	\sqrt{N}(\hat{\bs\theta}_N-{\bs\theta}_N) \converged \mathcal{N}( \bs{0},  \tilde{r}_1\tilde{r}_0\bs{\Gamma}^{-1} \bs{\Sigma} (\bs{\Gamma}^{-1})^\top )
	$, 
	where $\tilde{r}_z$ is the limit of $r_z$ for $z=0, 1$.
\end{theorem}

\subsection{Proofs of the theorems}

\begin{proof}[of Theorem \ref{thm:rate_sp}]
	We prove Theorem \ref{thm:rate_sp} using Theorem \ref{thm:rate}. 
    Following the notation in Theorem \ref{thm:rate}, 
	we define $n = n_1$,
	$\phi_i({\bs\theta})=r_1  \psi_{1\text{-}0, i}(\bs{\theta})$,  and $\tilde \phi(\bs\theta)= \E_N\{ \psi_{0i} ({\bs\theta}) \}$.  
	We verify conditions (i)--(vii) in Theorem \ref{thm:rate}. 
	Note that 
	conditions (i)--(v) and (vii) in Theorem \ref{thm:rate} follow immediately from 
	conditions (i)--(v) and (vii) in Theorem \ref{thm:rate_sp}. 
	Below we consider only conditions (vi) in Theorem \ref{thm:rate}. 
	By definition, for $1\le k\le p$, 
	$
	\Var_N\{ {\phi}_{ik}(\bs{\theta}_N) \} = r_1^2 \Var_N\{  {\psi}_{1\text{-}0, ik}(\bs{\theta}_N) \}
	$. Condition (vi) in Theorem \ref{thm:rate} then follows from condition (vi) in Theorem \ref{thm:rate_sp}. 
	
	From the above and Theorem \ref{thm:rate}, 
	$\hat{\bs{\theta}}_N - \bs{\theta}_N = O_{\Pr}(n_1^{-1/2})$. 
	Because both $r_1 = n_1/N$ and $r_0 = n_0/N$ have positive limits under the setting of Theorem \ref{thm:uniform_converge_ZCP_sp}, we then have 
	$\hat{\bs{\theta}}_N - \bs{\theta}_N = O_{\Pr}(N^{-1/2})$, 
	i.e., Theorem \ref{thm:rate_sp} holds. 
\end{proof}

\begin{proof}[of Theorem \ref{thm:clt_sp}]
	We prove Theorem \ref{thm:clt_sp} using Theorem \ref{thm:clt}. 
	From Theorem \ref{thm:rate_sp} and following the notation in the proof of Theorem \ref{thm:rate_sp}, 
	we can know that 
	conditions (i)--(vii) in Theorem \ref{thm:rate} hold.
	Below we verify conditions (i) and (ii) in Theorem \ref{thm:clt}. 
	
	First, conditions (i) in Theorem \ref{thm:clt} follows immediately from condition (i) in Theorem \ref{thm:clt_sp}. 
	In particular, 
	$
	\Cov_N\{ \phi_i(\bs{\theta}_N) \}
	= 
	r_1^2 \Cov_N\{ \psi_{1\text{-}0, i}(\bs{\theta}_N) \} 
	\rightarrow \tilde{r}_1^{2}\bs{\Sigma}. 
	$
	Second, 
	because 
	$
	\| \phi_j({\bs\theta}_N) - \E_N\{\phi_i({\bs\theta}_N)\} \| 
	= 
	r_1 \| \psi_{1\text{-}0, j}(\bs{\theta}_N) - \E_N\{ \psi_{1\text{-}0, i}(\bs{\theta}_N)\} \| 
	$, 
	condition (ii) in Theorem \ref{thm:clt} follows from condition (ii) in Theorem \ref{thm:clt_sp}.
	
	From the above and Theorem \ref{thm:clt}, we have 
	\begin{align*}
		\sqrt{  n_1N/n_0 } \cdot 
		( \hat{\bs{\theta}}_N - \bs{\theta}_N )
		\converged 
		\mathcal{N}
		\left( 
		\bs{0}, \tilde{r}_1^2\bs{\Gamma}^{-1}  \bs{\Sigma} (\bs{\Gamma}^{-1})^{\top}
		\right)
	\end{align*}
	By Slutsky's theorem, this further implies that 
	\begin{align*}
		\sqrt{N}( \hat{\bs{\theta}}_n - \bs{\theta}_N )
		=
		\sqrt{ n_0/n_1 } \cdot \sqrt{  n_1N/n_0 } \cdot 
		( \hat{\bs{\theta}}_n - \bs{\theta}_N )
		\converged 
		\mathcal{N}
		\left( 
		\bs{0}, \tilde{r}_1\tilde{r}_0\bs{\Gamma}^{-1} \bs{\Sigma} (\bs{\Gamma}^{-1})^\top
		\right). 
	\end{align*}
	Therefore, Theorem \ref{thm:clt_sp} holds. 
\end{proof}

\section{Technical lemmas for variance estimation under completely randomized experiments}\label{sec:var_cre}

\subsection{Lemmas}

\begin{lemma}\label{lemma:consist_of_var_e}
    For $1\le i \le N$, 
    let $f_{ai}(\theta)\in \mathbb{R}$ and $f_{bi}(\theta)\in \mathbb{R}$ be functions of $\theta\in \Theta$, where these functions can depend on $N$ but we make such dependence implicit for notational simplicity. 
	Let 
	$(Z_1, Z_2, \ldots, Z_N) \in \{0,1\}^N$ be a random vector 
	whose probability of taking value $(z_1, z_2, \ldots, z_N)$ is $1/\binom{N}{n}$ if $\sum_{i=1}^N z_i = n$ and zero otherwise, 
	where $n$ is a fixed constant. 
    Let $\{\theta_N\}$ be a sequence from $\Theta$. 
	If following conditions hold as $N\rightarrow \infty$: 
	\begin{itemize}
		\item[(i)] the finite population variances and covariance of $f_{ai}(\bs{\theta}_N )$ and $f_{bi}(\bs{\theta}_N)$ have limits, 
		\item[(ii)] $n/N$ has a limiting value in $(0,1)$, 
		\item[(iii)] $\max_{1\leq j\leq N}[f_{aj}(\bs{\theta}_N )-\E_N\{f_{ai}(\bs{\theta}_N )\}]^2$ and $\max_{1\leq j\leq N}[f_{bj}(\bs{\theta}_N )-\E_N\{f_{bi}(\bs{\theta}_N )\}]^2$ are of order $o(N)$,  
	\end{itemize}
	then 
	$
	\Cov_{N}^1 \{f_{ai}(\bs{\theta}_N),f_{bi}(\bs{\theta}_N) \}
	= \Cov_N\{f_{ai}(\bs{\theta}_N),f_{bi}(\bs{\theta}_N) \} + o_\Pr(1). 
	$
\end{lemma}

\begin{lemma}\label{lemma:consist_of_var_e_est_para}
    Under the same setting as in Lemma \ref{lemma:consist_of_var_e}, define 
    $$
		\Delta(f_{ai}, \bs \theta_N, \delta)=\sup _{\bs\theta\in \Theta:\|\bs\theta-\bs\theta_N\|\leq\delta}|
		f_{ai}(\bs\theta)-f_{ai}(\bs\theta' )|,
    $$
    and define analogously $\Delta(f_{ai}, \bs \theta_N, \delta)$. 
    If conditions (i)--(iii) in Lemma \ref{lemma:consist_of_var_e} hold, 
    and the following conditions hold: 
    \begin{enumerate}[label=(\roman*)]
        \item as $\delta \rightarrow 0$, $\sup_N \E_N\{\Delta^2(f_{ai}, \bs \theta_N, \delta)\} \rightarrow 0$, and $\sup_N \E_N\{\Delta^2(f_{bi}, \bs \theta_N, \delta)\} \rightarrow 0$,  
        \item $\hat{\theta}_N - \theta_N = o_{\Pr}(1)$ as $N\rightarrow \infty$, 
    \end{enumerate}
    then 
    $
	\Cov_{N}^1 \{f_{ai}(\hat{\bs{\theta}}_N), f_{bi}(\hat{\bs{\theta}}_N) \}
	= \Cov_N\{f_{ai}(\bs{\theta}_N),f_{bi}(\bs{\theta}_N) \} + o_\Pr(1). 
	$
\end{lemma}

\subsection{Proof of the lemmas}

\begin{proof}[of Lemma \ref{lemma:consist_of_var_e}]
	Lemma \ref{lemma:consist_of_var_e} follows immediately from \citet[][Lemma A3]{fpclt2017}. 
\end{proof}

\begin{proof}[of Lemma \ref{lemma:consist_of_var_e_est_para}]
    For $1\le j \le N$ with $Z_j=1$, define
    \begin{align*}
        \xi_{aj} & = 
        [ f_{aj}(\hat{\bs{\theta}}_N) - \E_N^1 \{f_{ai}(\hat{\bs{\theta}}_N)\} ]
        - [ f_{aj}(\bs{\theta}_N) - \E_N^1 \{f_{ai}(\bs{\theta}_N)\} ]
        \\
        & = 
        \{ f_{aj}(\hat{\bs{\theta}}_N) - f_{aj}(\bs{\theta}_N) \}
        - \E_N^1 \{f_{ai}(\hat{\bs{\theta}}_N) - f_{ai}(\bs{\theta}_N)\},  
    \end{align*}
    and define analogously $\xi_{bj}$. 
    We can then write $\Cov_{N}^1 \{f_{ai}(\hat{\bs{\theta}}_N), f_{bi}(\hat{\bs{\theta}}_N) \}$ as 
    \begin{align*}
        & \quad \ \Cov_{N}^1 \{f_{ai}(\hat{\bs{\theta}}_N), f_{bi}(\hat{\bs{\theta}}_N) \}
        \\
        & = 
        \frac{1}{n-1}
        \sum_{j:Z_j=1} 
        [f_{aj}(\bs{\theta}_N) - \E_N^1 \{f_{ai}(\bs{\theta}_N)\} + \xi_{aj}]
        [f_{bj}(\bs{\theta}_N) - \E_N^1 \{f_{bi}(\bs{\theta}_N)\} + \xi_{bj}]
        \\
        & = \Cov_{N}^1 \{f_{ai}({\bs{\theta}}_N), f_{bi}({\bs{\theta}}_N) \} 
        + 
        \frac{1}{n-1}
        \sum_{j:Z_j=1} 
        \xi_{aj}
        [f_{bj}(\bs{\theta}_N) - \E_N^1 \{f_{bi}(\bs{\theta}_N)\}]
        \\
        & \quad \ +
        \frac{1}{n-1} \sum_{j:Z_j=1} 
        \xi_{bj} [f_{aj}(\bs{\theta}_N) - \E_N^1 \{f_{ai}(\bs{\theta}_N)\} ]
        + 
        \frac{1}{n-1} \sum_{j:Z_j=1} \xi_{aj} \xi_{bj}. 
    \end{align*}
    By the Cauchy–Schwarz inequality, we have 
    \begin{align*}
        \Big( \frac{1}{n-1} \sum_{j:Z_j=1} 
        \xi_{bj} [f_{aj}(\bs{\theta}_N) - \E_N^1 \{f_{ai}(\bs{\theta}_N)\} ] \Big)^2 
        & \le 
        \frac{n}{n-1} \cdot\E_N^1 (\xi_{bi}^2) \cdot
        \Var_N^1 \{ f_{ai}( \bs{\theta}_N) \}, \\
        \Big( \frac{1}{n-1} \sum_{j:Z_j=1} 
        \xi_{aj} [f_{bj}(\bs{\theta}_N) - \E_N^1 \{f_{bi}(\bs{\theta}_N)\} ] \Big)^2 
        & \le 
        \frac{n}{n-1} \cdot \E_N^1 (\xi_{ai}^2) \cdot
        \Var_N^1 \{ f_{bi}( \bs{\theta}_N) \}, \\
        \Big( \frac{1}{n-1} \sum_{j:Z_j=1} \xi_{aj} \xi_{bj} \Big)^2 
        & \le 
        \frac{n^2}{(n-1)^2} \cdot  \E_N^1 (\xi_{ai}^2) \cdot \E_N^1 (\xi_{bi}^2)
    \end{align*}
    From Lemma \ref{lemma:consist_of_var_e}, we can know that
    $\Var_N^1 \{ f_{ai}( \bs{\theta}_N) \} = \Var_N \{ f_{ai}( \bs{\theta}_N) \} + o_\Pr(1)$, 
    $\Var_N^1 \{ f_{bi}( \bs{\theta}_N) \} = \Var_N \{ f_{bi}( \bs{\theta}_N) \} + o_\Pr(1)$,
    and 
    $
	\Cov_{N}^1 \{f_{ai}(\bs{\theta}_N),f_{bi}(\bs{\theta}_N) \}
	= \Cov_N\{f_{ai}(\bs{\theta}_N),f_{bi}(\bs{\theta}_N) \} + o_\Pr(1).
	$
	Therefore, to prove Lemma \ref{lemma:consist_of_var_e_est_para}, it suffices to prove that $\E_N^1 (\xi_{ai}^2) = o_{\Pr}(1)$ and $\E_N^1 (\xi_{bi}^2) = o_{\Pr}(1)$. 
    By symmetry, below we prove only $\E_N^1 (\xi_{ai}^2) = o_{\Pr}(1)$. 
    
    Let $\zeta_N = \| \hat{\theta}_N - \theta_N \|$. 
    By definition, 
    \begin{align*}
        |\xi_{aj}| \le 
        |f_{aj}(\hat{\bs{\theta}}_N) - f_{aj}(\bs{\theta}_N) |
        + \E_N^1 \{|f_{ai}(\hat{\bs{\theta}}_N) - f_{ai}(\bs{\theta}_N)|\} 
        \le 
        \Delta(f_{aj}, \bs \theta_N, \zeta_N) + \E_N^1 \{\Delta(f_{ai}, \bs \theta_N, \zeta_N)\}. 
    \end{align*}
    This then implies that $\xi_{aj}^2 \le 2 \Delta^2(f_{aj}, \bs \theta_N, \zeta_N) + 2 [\E_N^1 \{\Delta(f_{ai}, \bs \theta_N, \zeta_N)\}]^2$, and consequently  
    \begin{align}\label{eq:E_N1_xi_square}
        \E_N^1(\xi_{ai}^2)
        & \le 2 \E_N^1 \{\Delta^2(f_{ai}, \bs \theta_N, \zeta_N)\} + 2 [\E_N^1 \{\Delta(f_{ai}, \bs \theta_N, \zeta_N)\}]^2 
        \le 4 \E_N^1 \{\Delta^2(f_{ai}, \bs \theta_N, \zeta_N)\}. 
    \end{align}
    From condition (i), 
    for any $\varepsilon>0$ and $\eta \in (0,1)$, there exists $\delta>0$ such that 
    $\E_N\{\Delta^2(f_{ai}, \bs \theta_N, \delta)\} < \eta \varepsilon$ for all $N$. 
    We then have 
    \begin{align*}
        \Pr\big( \E_N^1 \{\Delta^2(f_{ai}, \bs \theta_N, \zeta_N)\}> \varepsilon \big)
        & \le 
        \Pr\big( \E_N^1 \{\Delta^2(f_{ai}, \bs \theta_N, \zeta_N)\}> \varepsilon, \zeta_N\le \delta \big) 
        +
        \Pr(\zeta_N > \delta)
        \\
        & \le \Pr\big( \E_N^1 \{\Delta^2(f_{ai}, \bs \theta_N, \delta)\}> \varepsilon \big) 
        +
        \Pr(\| \hat{\theta}_N - \theta_N \| > \delta)\\
        & 
        \le \varepsilon^{-1} \cdot \E_N \{\Delta^2(f_{ai}, \bs \theta_N, \delta)\} +
        \Pr(\| \hat{\theta}_N - \theta_N \| > \delta)
        \\
        & \le \eta + \Pr(\| \hat{\theta}_N - \theta_N \| > \delta),
    \end{align*}
    where the second last inequality holds due to the Markov inequality and the fact that $\E[\E_N^1 \{\Delta^2(f_{ai}, \bs \theta_N, \delta)\}] = \E_N \{\Delta^2(f_{ai}, \bs \theta_N, \delta)\}$, implied by the property of simple random sampling. 
    From condition (ii), 
    letting $N\rightarrow \infty$, 
    $
        \limsup_{N\rightarrow \infty}\Pr\big( \E_N^1 \{\Delta^2(f_{ai}, \bs \theta_N, \zeta_N)\}> \varepsilon \big) \le \eta. 
    $
    Because $\eta$ can be arbitrarily small, we must have $\E_N^1 \{\Delta^2(f_{ai}, \theta_N, \zeta_N)\} = o_{\Pr}(1)$. 
    From \eqref{eq:E_N1_xi_square}, this further implies that $\E_N^1(\xi_{ai}^2) = o_{\Pr}(1)$. 
    
    From the above, Lemma \ref{lemma:consist_of_var_e_est_para} holds. 
\end{proof}

\section{Proofs of Theorems \ref{thm:theta_clt} and \ref{thm:m_est_variance_estimate} and Corollary \ref{cor:wald_ci_coverage}}\label{sec:proof_thm12_cor1}

\begin{remark}\label{rmk:converge_dist}
 {\rev

    We first give the technical details for the comment before Theorem \ref{thm:theta_clt}. 
    Consider any two sequences of random vectors or distributions $A_N$ and $B_N$. 
    Below we show that, if the L\'evy--Prokhorov distance between the probability measures induced by $A_N$ and $B_N$ converges to zero as $N \rightarrow \infty$, 
    and $B_N$ follows mean-zero Gaussian distribution with the eigenvalues of
    its covariance matrix uniformly bounded away from $0$ and $\infty$ for all $N$, 
    then the Kolmogorov distance between the distributions of $A_N$ and $B_N$ converges to zero as $N \rightarrow \infty$.   

    Consider any given subsequence $\{N_s\}$.
    By the Bolzano–Weierstrass theorem, there exists a further subsequence $\{N_t\} \subset \{N_s\}$ such that 
    $\Cov(B_{N_t})$ has a nonsingular and finite limit as $t \rightarrow \infty$.  
    By the property of L\'evy--Prokhorov distance \citep[see, e.g.,][Theorem 4.2]{vG2003}, we know that, as $t\rightarrow\infty$, both $A_{N_t}$ and $B_{N_t}$ converge weakly to the Gaussian distribution with mean zero and covariance matrix $\lim_{t\rightarrow\infty}\Cov(B_{N_t})$. 
    From the multivariate P\'olya's theorem \citep{Chandra1989}, we know that the Kolmogorov distance between the distributions of $A_{N_t}$ and $B_{N_t}$ converges to zero as $t \rightarrow \infty$. 
    From the above, the Kolmogorov distance between the distributions of $A_{N}$ and $B_{N}$ must converge to zero as $N\rightarrow\infty$. 
    \hfill $\Box$}
\end{remark}

\begin{proof}[of Theorem \ref{thm:theta_clt}]
    Let $\lpdist_N = \lpdist(\sqrt{N}( \hat{\bs{\theta}}_N - \bs{\theta}_N), \mathcal{N}(\bs{0}, \bs{\Sigma}_N))$  
    denote the 
    L\'evy--Prokhorov distance
    between the probability measures induced by $\sqrt{N}( \hat{\bs{\theta}}_N - \bs{\theta}_N)$ and $\mathcal{N}(\bs{0}, \bs{\Sigma}_N)$. 
    To prove that 
    $\lpdist_N$ 
    converges to zero as $N\rightarrow \infty$, it suffices to prove that for any subsequence $\{N_s\}$, there exists a further subsequence $\{N_{t}\}\subset \{N_s\}$ such that $\lpdist_{N_{t}}$ converges to zero.

    Consider any given subsequence $\{N_{s}\}$. 
    From Conditions \ref{cond:stable_limits}(i) and (ii), 
    by the Bolzano–Weierstrass theorem,
    there exists a further subsequence $\{N_{t}\} \subset \{N_{s}\}$ 
    such that $r_1$, $r_0$, $\Cov_{N_{t}} \{ \psi_{1\text{-}0,i}(\bs{\theta}_{N_{t}}) \}$
    and  $\dot{\Psi}_{N_{t}}(\theta_{N_{t}})$ have limits. 
    Let $\tilde{r}_1, \tilde{r}_0$, $\bs{\Sigma}$ and $\bs{\Gamma}$ denote the 
    limit of $r_1, r_0$, $\Cov_{N_{t}} \{ \psi_{1\text{-}0,i}(\bs{\theta}_{N_{t}}) \}$
    and  $\dot{\Psi}_{N_{t}}(\theta_{N_{t}})$, respectively, as $t\rightarrow \infty$. 
    Below we prove that, along the subsequence $\{N_{t}\}$, 
    both $\sqrt{N_t}( \hat{\bs{\theta}}_{N_t} - \bs{\theta}_{N_t})$ and $\mathcal{N}(\bs{0}, \bs{\Sigma}_{N_{t}})$ converge weakly to the same distribution $\mathcal{N}(\bs{0}, \bs{\Sigma}_0)$, where $  \bs{\Sigma}_0 = \tilde{r}_1 \tilde{r}_0\bs{\Gamma}^{-1} \bs{\Sigma} (\bs{\Gamma}^{-1})^\top$.  We prove this using Theorem \ref{thm:clt_sp}.
    It suffices to verify all conditions in Theorem \ref{thm:clt_sp} along the subsequence $\{N_{t}\}$, which include conditions (i)--(vii) in Theorem \ref{thm:rate_sp} and conditions (i) and (ii) in Theorem \ref{thm:clt_sp}. 
	We can verify that 
	\begin{itemize}%
		\item[(a)] Condition \ref{cond:population_minimum} implies condition (i) in Theorem \ref{thm:rate_sp}; 

            \item[(b)] Condition      \ref{cond:empirical_minimum} implies conditions (ii) and (iv) in Theorem \ref{thm:rate_sp};

		\item[(c)] Condition \ref{cond:continuous} implies conditions (iii) and (v) in Theorem \ref{thm:rate_sp};
		
		\item[(d)] Condition \ref{cond:stable_limits} and the property of the subsequence imply conditions (vi) and (vii) in Theorem \ref{thm:rate_sp}
		and conditions (i) and (ii) in Theorem \ref{thm:clt_sp}.
        \end{itemize}
    Most of these are obvious. Below we only verify condition (v) in Theorem \ref{thm:rate_sp} using Condition \ref{cond:continuous}. 
    For any $\delta>0$, any $N$ and any $\theta$ satisfying $\|\theta - \theta_N\|\le \delta$, 
    \begin{align*}
        \| \dot{\Psi}_N(\bs\theta)  - \dot{\Psi}_N(\bs\theta_N) \| 
        & \le  
        \| 
        r_1 \E_N\{\dot\psi_{1i}(\theta) - \dot\psi_{1i}(\theta_N)\} 
        + 
        r_0 \E_N\{\dot\psi_{0i}(\theta) - \dot\psi_{0i}(\theta_N)\} 
        \|\\
        & \le 
        \E_N\{\| \dot\psi_{1i}(\theta) - \dot\psi_{1i}(\theta_N)\|\}
        + 
        \E_N\{\|\dot\psi_{0i}(\theta) - \dot\psi_{0i}(\theta_N) \|\} 
        \\
        & \le \E_N \{\Delta(\dot\psi_{1i}, \theta_N, \delta)\} + \E_N \{\Delta(\dot\psi_{0i}, \theta_N, \delta)\}.
    \end{align*}
    
    Thus, 
    $$
    \sup_N \sup_{\bs\theta \in \Theta:\|\bs\theta-\bs\theta_N\|\leq\delta} \| \dot{\Psi}_N(\bs\theta)  - \dot{\Psi}_N(\bs\theta_N) \| \le \sup_N \E_N \{\Delta(\dot\psi_{1i}, \theta_N, \delta)\} + \sup_N \E_N \{\Delta(\dot\psi_{0i}, \theta_N, \delta)\},$$ 
    which converges to $0$ as $\delta \rightarrow 0$ under Condition \ref{cond:continuous}. 
    By Theorem \ref{thm:clt_sp}, 
    we then have $\sqrt{N_t}( \hat{\bs{\theta}}_{N_t} - \bs{\theta}_{N_t})$
    converges weakly to $\mathcal{N}(\bs{0}, \bs{\Sigma}_0)$. 
    By the construction of the subsequence $\{N_t\}$, we obviously have $\bs{\Sigma}_{N_{t}} \rightarrow \bs{\Sigma}_0$, and thus $\mathcal{N}(\bs{0}, \bs{\Sigma}_{N_{t}})$ converges weakly to $\mathcal{N}(\bs{0}, \bs{\Sigma}_0)$. 
    These then imply that $\pi_{N_t} \rightarrow 0$ as $t\rightarrow 0$. 
 
	From the above, Theorem \ref{thm:theta_clt} holds. 
\end{proof}

\begin{proof}[of Theorem \ref{thm:m_est_variance_estimate}]    
    By similar logic as the proof of Theorem \ref{thm:theta_clt}, 
    it suffices to show that $\hat{\bs{\Sigma}}_N = \tilde{\bs{\Sigma}}_N  + o_{\Pr}(1)$ along any subsequence $\{N_t\}$ such that $r_z$, $\Cov_N\{\psi_{zi}(\bs{\theta}_N)\}$ , $\Cov_N\{\psi_{1i}(\bs{\theta}_N),\psi_{0i}(\bs{\theta}_N)\}$ and $\dot \Psi_N(\bs{\theta}_N)$ have limits, $z=0,1$.  Note that all these quantities or their entries are bounded under Condition \ref{cond:stable_limits}.
 
        First, we prove that $\dot{\hat{\Psi}}_N(\hat{\bs{\theta}}_N) = \dot \Psi_N(\bs{\theta}_N) + o_{\Pr}(1)$. 
	Note that 
	\begin{align}\label{eq:var_est_ddot_M}
		&\big\| \dot{\hat{\Psi}}_N(\hat{\bs{\theta}}_N) - \dot \Psi_N(\bs{\theta}_N) \big\| \nonumber \\
		 \le &
		\big\| \dot{\hat{\Psi}}_N(\hat{\bs{\theta}}_N) - \dot \Psi_N(\hat{\bs{\theta}}_N) \big\|
		+ 
		\big\| \dot \Psi_N(\hat{\bs{\theta}}_N) - \dot \Psi_N(\bs{\theta}_N) \big\|
		\nonumber
		\\
		\le &
		\sup_{\bs{\theta} \in \Theta: \|\bs{\theta} - \bs{\theta}_N\| \leq \|\hat{\bs{\theta}}_N - \bs{\theta}_N\|} \big\| \dot{\hat{\Psi}}_N(\bs{\theta}) - \dot \Psi_N(\bs{\theta}) \big\| + \sup_{\bs{\theta} \in \Theta: \|\bs{\theta} - \bs{\theta}_N\| \leq \|\hat{\bs{\theta}}_N - \bs{\theta}_N\|} \big\| \dot \Psi_N(\bs{\theta}) - \dot \Psi_N(\bs{\theta}_N) \big\|. 
	\end{align}
	By the same logic as the second part in the proof of Theorem \ref{thm:rate}, in particular the proof for the two terms in \eqref{eq:bound_H},  
    both terms in \eqref{eq:var_est_ddot_M} are of order $o_{\Pr}(1)$ under Conditions  \ref{cond:continuous} and \ref{cond:empirical_minimum}. 
	Thus, $\dot{\hat{\Psi}}_N(\hat{\bs{\theta}}_N) - \dot \Psi_N(\bs{\theta}_N) = o_{\Pr}(1)$. 
	
	Second, we prove that $\Cov_{N}^z \{\psi_{z,i}(\hat{\bs{\theta}}_N) \}= \Cov_{N}\{\psi_{z,i}(\bs{\theta}_N) \} + o_{\Pr}(1)$ for $z=0,1$ using Lemma \ref{lemma:consist_of_var_e_est_para}. 
	It suffices to show that 
	conditions (i)--(iii) in Lemma \ref{lemma:consist_of_var_e} and 
	conditions (i)--(ii) in Lemma \ref{lemma:consist_of_var_e_est_para} hold for $f_{ai}( \bs{\theta}) =  \psi_{z,ik}(\bs{\theta})$ and $f_{bi}(\bs{\theta}) = \psi_{z',il}(\bs{\theta})$ for any $0\le z,z' \le 1$ and $1\le k, l \le p$. 
	We can verify that 
	Condition \ref{cond:stable_limits} and the property of the subsequence imply conditions (i)-(iii) in Lemma \ref{lemma:consist_of_var_e}, 
	Condition \ref{cond:continuous} implies condition (i) in Lemma \ref{lemma:consist_of_var_e_est_para}, 
	and 
	Condition \ref{cond:empirical_minimum} implies condition (ii) in Lemma \ref{lemma:consist_of_var_e_est_para}. 
	
	From the above, we can know that
	\begin{align*}
		\hat{\bs{\Sigma}}_N 
		& = r_1 r_0 \{\dot{\hat{\Psi}}_N(\hat{\bs{\theta}}_N)\}^{-1} 
		\big(
		r_1^{-1} {\Cov}_N^1\{\psi_{1i}(\hat{\bs{\theta}}_N) \} + 
		r_0^{-1} {\Cov}_N^0 \{\psi_{0i}(\hat{\bs{\theta}}_N) \}
		\big)
		\{\dot{\hat{\Psi}}_N(\hat{\bs{\theta}}_N)^\top\}^{-1}
		\\
		& = 
		r_1 r_0 \{\dot{\Psi}_N(\bs{\theta}_N)\}^{-1} 
		\Big(
		r_1^{-1} \Cov_{N}\{\psi_{1i}(\bs{\theta}_N) \} + 
		r_0^{-1} \Cov_{N}\{\psi_{0i}(\bs{\theta}_N) \}
		\Big)
		\{\dot{\Psi}_N(\bs{\theta}_N)^\top\}^{-1} + o_{\Pr}(1)
		\\
		& = \tilde{\bs{\Sigma}}_N + o_{\Pr}(1), 
	\end{align*}
	where $\tilde{\bs{\Sigma}}_N \equiv r_1 r_0 \{\dot{\Psi}_N(\bs{\theta}_N)\}^{-1} 
		(
		r_1^{-1} \Cov_{N}\{\psi_{1i}(\bs{\theta}_N) \} + 
		r_0^{-1} \Cov_{N}\{\psi_{0i}(\bs{\theta}_N) \}
		)
		\{\dot{\Psi}_N(\bs{\theta}_N)\}^{-1}$. 
	Because 
	\begin{align*}
		& \quad \ r_1^{-1} \Cov_{N}\{\psi_{1i}(\bs{\theta}_N) \} + 
		r_0^{-1} \Cov_{N}\{\psi_{0i}(\bs{\theta}_N) \} - 
		\Cov_N\{ \psi_{1\text{-}0, i}(\bs{\theta}_N) \}
		\\
		& = 
		\frac{r_0}{r_1} \Cov_{N}\{\psi_{1i}(\bs{\theta}_N) \} + \frac{r_1}{r_0} \Cov_{N}\{\psi_{0i}(\bs{\theta}_N) \}  
		+ \Cov_N\{ \psi_{1i}(\bs{\theta}_N), \psi_{0i}(\bs{\theta}_N)\}
		+ \Cov_N\{  \psi_{0i}(\bs{\theta}_N), \psi_{1i}(\bs{\theta}_N)\}
		\\
		& = 
		\Cov_N\left\{ \sqrt{\frac{r_0}{r_1}} \  \psi_{1i}(\bs{\theta}_N) + \sqrt{\frac{r_1}{r_0}} \ \psi_{0i}(\bs{\theta}_N) \right\} \ge 0,
	\end{align*}
	we must have $\tilde{\bs{\Sigma}}_N \ge \bs{\Sigma}_N$. 
	Therefore, Theorem \ref{thm:m_est_variance_estimate} holds. 
\end{proof}

{\rev
\begin{proof}[of Corollary \ref{cor:wald_ci_coverage}]  

By similar logic as the proof of Theorem \ref{thm:theta_clt}, 
it suffices to show that for any $v \in \mathbb{R}^{p\times m}$ and $\alpha \in (0,1)$, 
\begin{align*}
\liminf_{t \to \infty} \Pr \big[ N_t \{ v^\top(\hat{\bs{\theta}}_{N_t} - \bs{\theta}_{N_t}) \}^\top   (v^\top \hat{\Sigma}_{N_t} v)^{-1} \{ v^\top(\hat{\bs{\theta}}_{N_t} - \bs{\theta}_{N_t}) \} 
\le \chi^2_{m, 1-\alpha} ) \big] \geq 1 - \alpha,
\end{align*}
along any subsequence $\{N_t\}$ such that $\Sigma_N$ and $\tilde{\Sigma}_N$ have limits $\Sigma_\infty$ and $\tilde{\Sigma}_\infty$, respectively. 
Because $\tilde{\Sigma}_N \geq \Sigma_N$ for all $N$ and $\sigma_{\min}(\tilde{\Sigma}_N)$ is bounded away from zero, we can know that $\tilde{\Sigma}_\infty \geq \Sigma_\infty$ and $\tilde{\Sigma}_\infty$ is positive-definite.

From Theorem \ref{thm:theta_clt} and Theorem \ref{thm:m_est_variance_estimate}, we can know that, as $t\rightarrow \infty$, 
\begin{align*}
\sqrt{N_t} v^\top(\hat{\bs{\theta}}_{N_t} - \bs{\theta}_{N_t}) \converged \mathcal{N}(0, v^\top \Sigma_\infty v) \quad \text{and} \quad v^\top \hat{\Sigma}_{N_t} v \convergep v^\top \tilde{\Sigma}_\infty v.
\end{align*}
By Slutsky's theorem, we then have, as $t\rightarrow \infty$, 
\begin{align*}
N_t \{ v^\top(\hat{\bs{\theta}}_{N_t} - \bs{\theta}_{N_t}) \}^\top   (v^\top \hat{\Sigma}_{N_t} v)^{-1} \{ v^\top(\hat{\bs{\theta}}_{N_t} - \bs{\theta}_{N_t}) \}
& \converged \epsilon^\top (v^\top \Sigma_\infty v)^{1/2} (v^\top \tilde{\Sigma}_\infty v)^{-1} (v^\top \Sigma_\infty v)^{1/2} \epsilon \leq \epsilon^T\epsilon,
\end{align*}
where $(v^\top \Sigma_\infty v)^{1/2}$ denotes the  positive semidefinite square root of $v^\top \tilde{\Sigma}_\infty v$, 
$\epsilon$ is an $m$ dimensional standard Gaussian random vector, 
and the last inequality holds due to the fact that $\tilde{\Sigma}_\infty \geq \Sigma_\infty$.
Note that
the distribution function of $\epsilon^\top (v^\top \Sigma_\infty v)^{1/2} (v^\top \tilde{\Sigma}_\infty v)^{-1} (v^\top \Sigma_\infty v)^{1/2} \epsilon$ must be continuous at  
$\chi^2_{m, 1-\alpha}>0$.  
Thus, the coverage probability of the confidence set satisfies
\begin{align*}
&\quad \ \liminf_{t \to \infty} \Pr \big[ N_t \{ v^\top(\hat{\bs{\theta}}_{N_t} - \bs{\theta}_{N_t}) \}^\top   (v^\top \hat{\Sigma}_{N_t} v)^{-1} \{ v^\top(\hat{\bs{\theta}}_{N_t} - \bs{\theta}_{N_t}) \} 
\le \chi^2_{m, 1-\alpha} ) \big]  \\
& = \Pr \{ \epsilon^\top (v^\top \Sigma_\infty v)^{\frac{1}{2}} (v^\top \tilde{\Sigma}_\infty v)^{-1} (v^\top \Sigma_\infty v)^{\frac{1}{2}} \epsilon 
\le \chi^2_{m, 1-\alpha} ) \} \geq \Pr \{ \epsilon^\top  \epsilon \le \chi^2_{m, 1-\alpha} ) \} = 1 - \alpha.
\end{align*}

From the above, Corollary \ref{cor:wald_ci_coverage} holds. 
\end{proof}
}

\section{Proofs of Theorems \ref{thm:clt_mae},  \ref{thm:MA_ci} and \ref{thm:ma_optimal}}\label{sec:proof_thm345}

\subsection{Proof of Theorem \ref{thm:clt_mae}}
To prove Theorem \ref{thm:clt_mae}, we need the following three lemmas. 
\begin{lemma}\label{lemma:average_true_hat}
    Let $m_1(\theta), m_2(\theta), \ldots, m_N(\theta)$ be functions of $\theta\in \Theta$, 
    where the functions can also depend on $N$ but we make this dependence implicit for notational convenience, 
    and 
    $(Z_1, Z_2, \ldots, Z_N) \in \{0,1\}^N$ be a random vector 
	whose probability of taking value $(z_1, z_2, \ldots, z_N)$ is $n_1!n_0!/N!$ if $\sum_{i=1}^N z_i = n$ and zero otherwise, 
	where $n$ is a fixed constant.
	Define 
    $$
		\Delta(m_{i}, \theta_N, \delta)=\sup _{\bs\theta \in \Theta:\|\bs\theta-\bs\theta_N \|\leq\delta}|
		m_{i}(\bs\theta)-m_{i}(\bs\theta_N)|.
	$$
	If the following conditions hold: 
	\begin{enumerate}[label=(\roman*)]
	    \item $\sup_N \E_N \{\Delta(m_i, \theta_N, \delta)\} \rightarrow 0$ as $\delta \rightarrow 0$,  	
	    
	    \item $\hat{\theta}_N - \theta_N = o_{\Pr}(1)$, 
	\end{enumerate}
	then $\E_N^1 \{m_i(\hat{\theta}_N)\} - \E_N^1 \{m_i(\theta_N)\} = o_{\Pr}(1)$ and 
	$\E_N \{m_i(\hat{\theta}_N)\} - \E_N \{m_i(\theta_N)\} = o_{\Pr}(1)$. 
\end{lemma}

\begin{proof}[of Lemma \ref{lemma:average_true_hat}]
    By definition, 
    \begin{align*}
        |\E_N^1 \{m_i(\hat{\theta}_N)\} - \E_N^1 \{m_i(\theta_N)\}|
        & \le 
        \E_N^1 \{ |m_i(\hat{\theta}_N) - m_i(\theta_N)| \}
        \le 
        \E_N^1 \{ \Delta(m_i, \theta_N, \|\hat{\theta}_N - \theta_N\|) \},
    \end{align*}
    and analogously, 
    \begin{align*}
        |\E_N \{m_i(\hat{\theta}_N)\} - \E_N \{m_i(\theta_N)\}|
        & \le 
        \E_N \{ |m_i(\hat{\theta}_N) - m_i(\theta_N)| \}
        \le 
        \E_N \{ \Delta(m_i, \theta_N, \|\hat{\theta}_N - \theta_N\|) \}. 
    \end{align*}
    From condition (i), for any $\varepsilon>0$ and $\eta \in (0,1)$, there exists $\delta > 0$ such that $\E_N \{\Delta(m_i, \theta_N, \delta)\} < \eta \varepsilon$ for all $N$. We then have 
    \begin{align*}
        & \quad \ \Pr\big( 
        |\E_N^1 \{m_i(\hat{\theta}_N)\} - \E_N^1 \{m_i(\theta_N)\}| > \varepsilon
        \big)
        \\
        & \le 
        \Pr\big( 
        |\E_N^1 \{m_i(\hat{\theta}_N)\} - \E_N^1 \{m_i(\theta_N)\}| > \varepsilon, \| \hat{\theta}_N - \theta_N \| \le \delta
        \big) + 
        \Pr\big( \| \hat{\theta}_N - \theta_N \| > \delta \big)\\
        & \le 
        \Pr\big( 
        \E_N^1 \{ \Delta(m_i, \theta_N, \|\hat{\theta}_N - \theta_N\|) \} > \varepsilon, \| \hat{\theta}_N - \theta_N \| \le \delta
        \big) + 
        \Pr\big( \| \hat{\theta}_N - \theta_N \| > \delta \big)\\
        & \le 
        \Pr\big( 
        \E_N^1 \{ \Delta(m_i, \theta_N, \delta) \} > \varepsilon
        \big) + 
        \Pr\big( \| \hat{\theta}_N - \theta_N \| > \delta \big)
        \\
        & \le 
        \varepsilon^{-1} \E_N \{ \Delta(m_i, \theta_N, \delta) \} + 
        \Pr\big( \| \hat{\theta}_N - \theta_N \| > \delta \big)
        \\
        & \le \eta + \Pr\big( \| \hat{\theta}_N - \theta_N \| > \delta \big), 
    \end{align*}
    where the second last inequality holds due to the Markov inequality, 
    and 
    \begin{align*}
        \Pr\big( |\E_N \{m_i(\hat{\theta}_N)\} - \E_N \{m_i(\theta_N)\}| > \varepsilon \big)
        & \le 
        \Pr\big( \E_N \{ \Delta(m_i, \theta_N, \|\hat{\theta}_N - \theta_N\|) \} > \varepsilon \big)\\
        & \le \Pr\big( \| \hat{\theta}_N - \theta_N \| > \delta \big). 
    \end{align*}
    From condition (ii) and noting that $\eta$ can be arbitrarily small, we can derive that, 
    as $N\rightarrow \infty$, both 
    $\Pr( 
        |\E_N^1 \{m_i(\hat{\theta}_N)\} - \E_N^1 \{m_i(\theta_N)\}| > \varepsilon
    )$
    and 
    $ \Pr( |\E_N \{m_i(\hat{\theta}_N)\} - \E_N \{m_i(\theta_N)\}| > \varepsilon )$ 
    converge to zero. 
    Therefore, we derive Lemma \ref{lemma:average_true_hat}.
\end{proof}

\begin{lemma}\label{lemma:maa_1}
	Let $\{Y_i(1), Y_i(0), \bs{X}_i\}$ be a finite population of $N$ units, and
	$(Z_1, Z_2, \ldots, Z_N) \in \{0,1\}^N$ be a random vector 
	whose probability of taking value $(z_1, z_2, \ldots, z_N)$ is $n_1!n_0!/N!$ if $\sum_{i=1}^N z_i = n_1$ and zero otherwise, 
	where $n_1$ and $n_0 = N-n_1$ are fixed constants.
	For any $\bs{x}$, 
	let $h_1(\bs{x}, \bs{\theta})\in \mathbb{R}$ and $h_1(\bs{x}, \bs{\theta}) \in \mathbb{R}$ be two functions of $\bs{\theta}\in \Theta \subset \mathbb{R}^p$. 
	If the following conditions hold: 
	\begin{enumerate}[label=(\roman*)]
		\item $n_1/N$ and $n_0/N$ have positive limits as $N\rightarrow \infty$, 
		\item there exists $\varepsilon_0>0$ such that $\mathcal{B}(\bs{\theta}_N, \varepsilon_0) \subset \Theta$ for all $N$, 
		\item for any $N$ and $1\le i \le N$, 
		$h_1(\bs{X}_i, \bs{\theta})$ and $h_1(\bs{X}_i, \bs{\theta})$, viewed as functions of $\bs{\theta}$, are continuously differentiable over $\bs{\theta}\in \Theta$, 
		\item $\sup_N \E_N\{\Delta(\dot{h}_{zi}, \theta_N, \delta)\} \rightarrow 0$ as $\delta \rightarrow 0$, where 
		$$
		\Delta(\dot{h}_{zi}, \theta_N, \delta) = \sup_{\bs{\theta}\in \Theta: \|\bs{\theta}-\bs{\theta}_N \|\le \delta} 
		\| \dot{h}_z(\bs{X}_i; \bs{\theta}) - \dot{h}_z(\bs{X}_i; \bs{\theta}_N) \|,
		$$
		\item for $z=0,1$, $N^{-1} \Cov_N\{ \dot{h}_z(\bs{X}_i; \bs{\theta}_N) \} = o(1)$, 
		\item $\hat{\bs\theta}_N-{\bs\theta}_N=O_{\Pr}(N^{-1/2})$, 
	\end{enumerate}
	then, for $z=0,1$, 
	\begin{align*}
		{\E}_N^z \{h_z(\bs X_i;\hat{\bs{\theta}}_N)\}- \E_N\{ h_z(\bs X_i;\hat{\bs\theta}_N) \}
		& = 
		{\E}_N^z\{ h_z(\bs X_i;{\bs\theta}_N) \} - \E_N\{ h_z(\bs X_i;{\bs\theta}_N) \} + 
		o_{\Pr}(N^{-1/2}). 
	\end{align*}
\end{lemma}
\begin{proof}[of Lemma \ref{lemma:maa_1}]
	By the same logic as  the proof of Theorem \ref{thm:rate}, 
    we can assume that 
	$\hat{\bs{\theta}}_N \in \overline{\mathcal{B}}(\bs{\theta}_N, \varepsilon_0/2) \equiv \{\theta: \|\theta-\theta_N\|\le \varepsilon_0/2\} \subset \Theta$ for all $N$. 
	By symmetry, we prove only the case with $z=1$, since the proof for the other case with $z=0$ is almost the same. 
	By the mean value theorem, 
	there exists $\tilde {\bs\theta}_N$ between $\hat{\bs\theta}_N$ and ${\bs\theta}_N$ such that 
	\begin{align*}
		& \quad \ 
		\big[ {\E}_N^1 \{h_1(\bs X_i;\hat{\bs{\theta}}_N)\}- \E_N\{ h_1(\bs X_i;\hat{\bs\theta}_N) \}  \big]
		- 
		\big[ {\E}_N^1 \{ h_1(\bs X_i;{\bs\theta}_N) \} - \E_N\{ h_1(\bs X_i;{\bs\theta}_N) \} \big]
		\\
		& = 
		( \hat{\bs\theta}_N - \bs\theta_N )^\top 
		\big[ {\E}_N^1 \{\dot{h}_1(\bs X_i;\tilde{\bs{\theta}}_N)\}- \E_N\{ \dot{h}_1(\bs X_i;\tilde{\bs\theta}_N) \}  \big] . 
	\end{align*}
	Because $\hat{\bs\theta}_N - \bs\theta_N = O_{\Pr}(N^{-1/2})$, 
    we must have $\tilde{\bs\theta}_N - \bs\theta_N = O_{\Pr}(N^{-1/2}) = o_{\Pr}(1)$. 
    Let $\nabla_k {h}_1(\bs X_i;{\bs{\theta}})$ denote the partial derivative of ${h}_1(\bs X_i;{\bs{\theta}})$ over the $k$th coordinate of $\theta$.     
    From condition (iv) 
    and applying Lemma \ref{lemma:average_true_hat}, we then have, for $1\le k \le p$, 
    \begin{align*}
        & \quad \ \E_N^1 \{ \nabla_k {h}_1(\bs X_i;\tilde{\bs{\theta}}_N)\}- \E_N\{ \nabla_k {h}_1(\bs X_i;\tilde{\bs\theta}_N) \}\\
        & = 
        \E_N^1 \{ \nabla_k {h}_1(\bs X_i; \bs{\theta}_N)\}- \E_N\{ \nabla_k {h}_1(\bs X_i; \bs\theta_N) \} + o_{\Pr}(1)
        \\
        & = 
        \sqrt{\left(\frac{1}{n} - \frac{1}{N}\right) \Var_N\{\nabla_k {h}_1(X_i; \theta_N)\}} \cdot O_{\Pr}(1) + o_{\Pr}(1)
        = o_{\Pr}(1), 
    \end{align*}
    where the second last equality follows from Chebyshev's inequality and the property of simple random sampling, 
    and the last equality follows from conditions (i) and (v). 
    From the above, we have
    \begin{align*}
		& \quad \ 
		\big[ {\E}_N^1 \{h_1(\bs X_i;\hat{\bs{\theta}}_N)\}- \E_N\{ h_1(\bs X_i;\hat{\bs\theta}_N) \}  \big]
		- 
		\big[ {\E}_N^1 \{ h_1(\bs X_i;{\bs\theta}_N) \} - \E_N\{ h_1(\bs X_i;{\bs\theta}_N) \} \big]
		\\
		& = 
		( \hat{\bs\theta}_N - \bs\theta_N )^\top 
		\big[ {\E}_N^1 \{\dot{h}_1(\bs X_i;\tilde{\bs{\theta}}_N)\}- \E_N\{ \dot{h}_1(\bs X_i;\tilde{\bs\theta}_N) \}  \big] 
		= 
		O_{\Pr}(N^{-1/2}) \cdot o_{\Pr}(1) = o_{\Pr}(N^{-1/2}). 
	\end{align*}
    Therefore, Lemma \ref{lemma:maa_1} holds.  
\end{proof}

\begin{lemma}\label{lemma:maa_2}
	Consider the same setting as Lemma \ref{lemma:maa_1}, 
	and define $Y_i(z;\theta)$, $Y_i(\theta)$ and $\bar{Y}_z(\theta)$ the same as in \S \ref{sec:MA_consist}. 
	If conditions (i)--(vi) in Lemma \ref{lemma:maa_1} hold, 
	and the following conditions hold as $N\rightarrow\infty$: 
	\begin{enumerate}[label = (\roman*)]
		\item the finite population variances and covariance for $Y_i(1;\theta_N)$ and $Y_i(0;\theta_N)$ have limits, 
		\item for $z=0,1$, $N^{-1}\max_{1\le j \le N} [ Y_j(z;\theta_N)-\E_N\{Y_i(z;\theta_N)\} ]^2 \rightarrow 0$, 
	\end{enumerate}
	then
	\begin{align*}
	    \sqrt{N}
		\begin{pmatrix}
			\bar{Y}_1(\hat{\theta}_N) - \bar{Y}(1)\\
			\bar{Y}_0(\hat{\theta}_N) - \bar{Y}(0)
		\end{pmatrix}
		& = \sqrt{N}
		\begin{pmatrix}
			{\E}_N^1 \{ Y_i(1)-h_1(\bs X_i;\hat{\bs{\theta}}_N ) \}+ \E_N\{h_1(\bs X_i;\hat{\bs\theta}_N)\} - \bar{Y}(1)\\
			{\E}_N^0 \{ Y_i(0)-h_0(\bs X_i;\hat{\bs{\theta}}_N ) \}+ \E_N\{h_0(\bs X_i;\hat{\bs\theta}_N)\} - \bar{Y}(0)
		\end{pmatrix}
		\\
		& \converged
		\mathcal{N}
		\left( 
		\bs{0}, \ 
		\begin{pmatrix}
			( \tilde{r}_1^{-1} - 1 )\Omega_{11} & - \Omega_{10}\\
			- \Omega_{01} & ( \tilde{r}_0^{-1} - 1 ) \Omega_{00}
		\end{pmatrix}
		\right), 
	\end{align*}
	where $\tilde{r}_1$ and $\tilde{r}_0$ are the limits of $r_1$ and $r_0$, 
	and $\Omega_{11}$, $\Omega_{00}$ and $\Omega_{10} = \Omega_{01}$ are the limits of the finite population variances of $Y_i(1) - h_1(\bs{X}_i; \bs{\theta}_N)$ and $Y_i(0) - h_0(\bs{X}_i; \bs{\theta}_N)$ and their finite population covariance. 
\end{lemma}
\begin{proof}[of Lemma \ref{lemma:maa_2}]
	From Lemma \ref{lemma:maa_1}, 
	for $z=0, 1$, we have 
	\begin{align}\label{eq:MA_e_eq}
		\bar{Y}_z(\hat{\theta}_N) - \bar{Y}(z) & =
		{\E}_N^z \{ Y_i(z)-h_z(\bs X_i;\hat{\bs{\theta}}_N ) \}+\E_N\{ h_z(\bs X_i;\hat{\bs\theta}_N) \} - \bar{Y}(z)
		\nonumber
		\\
		& = 
		{\E}_N^z \{Y_i(z)\} - 
		\big[ {\E}_N^z\{h_z(\bs X_i;\hat{\bs{\theta}}_N)\}- \E_N\{ h_z(\bs X_i;\hat{\bs\theta}_N) \} \big] - \bar{Y}(z)
		\nonumber
		\\
		& = {\E}_N^z \{Y_i(z)\} - 
		\big[ {\E}_N^z \{h_z(\bs X_i;\bs{\theta}_N)\}- \E_N\{ h_z(\bs X_i;\bs\theta_N) \} \big] - \bar{Y}(z) + o_{\Pr}(N^{-1/2})
		\nonumber
		\\
		& = 
		{\E}_N^z \{Y_i(z) - h_z(\bs X_i;\bs{\theta}_N)\} +    \E_N\{ h_z(\bs X_i;\bs\theta_N) \}
		- \bar{Y}(z) + o_{\Pr}(N^{-1/2})
		\nonumber
		\\
		& = {\E}_N^z \{Y_i(z, \bs{\theta}_N)\} - \E_{N} \{Y_i(z, \bs{\theta}_N)\}  + o_{\Pr}(N^{-1/2}).
	\end{align}
	where the last equality follows from the definition of $Y_{i}(z; \bs{\theta})$. 
	Below we prove Lemma \ref{lemma:maa_2} using Theorem 5 in \citet{fpclt2017}. 
	
	Following the notation in \citet{fpclt2017}, 
	define matrices
	$\bs{A}_1 = (1, 0)^\top$ and $\bs{A}_0 = (0, 1)^\top$, 
	and 
	$\hat{\bs{\delta}}(\bs{A}) = 
	\bs{A}_1 \E_N^1\{ Y_i(1; \bs{\theta}_N) \} + 
	\bs{A}_0 \E_N^0\{ Y_i(0; \bs{\theta}_N) \} = (\E_N^1\{ Y_i(1; \bs{\theta}_N) \}, \E_N^0\{ Y_i(0; \bs{\theta}_N))^\top$. 
	We can verify that, 
	$\hat{\bs{\delta}}(\bs{A})$ has mean 
	$$
	\bs{A}_1\E_N\{Y_i(1; \bs{\theta}_N)\} + \bs{A}_0 \E_N\{Y_i(0; \bs{\theta}_N)\} = (\E_N\{Y_i(1; \bs{\theta}_N)\}, \E_N\{Y_i(0; \bs{\theta}_N)\})^\top, 
	$$
	and covariance matrix 
	\begin{align*}
		& \quad \ \frac{1}{n_1}\bs{A}_1 \Var_N\{Y_i(1; \bs{\theta}_N)\} \bs{A}_1^\top + 
		\frac{1}{n_0}\bs{A}_0 \Var_N\{Y_i(0; \bs{\theta}_N)\} \bs{A}_0^\top - 
		\frac{1}{N} \Cov_N\left\{ (Y_i(1; \bs{\theta}_N), Y_i(0; \bs{\theta}_N))^\top \right\}
		\\
		& = 
		\begin{pmatrix}
			( n_1^{-1} - N^{-1} ) \Var_N\{Y_i(1; \bs{\theta}_N)\} & - N^{-1} \Cov_N\{Y_i(1; \bs{\theta}_N), Y_i(0; \bs{\theta}_N)\}\\
			- N^{-1} \Cov_N\{Y_i(1; \bs{\theta}_N), Y_i(0; \bs{\theta}_N)\} & ( n_0^{-1} - N^{-1} ) \Var_N\{Y_i(0; \bs{\theta}_N)\}
		\end{pmatrix}. 
	\end{align*}
	From condition (i) in Lemma \ref{lemma:maa_1},  conditions (i) and (ii) in Lemma \ref{lemma:maa_2} and the finite population central limit theorem \citep[][Theorem 5]{fpclt2017}, 
	$\hat{\bs{\delta}}(\bs{A})$ has the following asymptotic distribution as $N\rightarrow \infty$: 
	\begin{align*}%
		\sqrt{N}
		\begin{pmatrix}
			{\E}_N^1 \{Y_i(1; \bs{\theta}_N)\} - \E_N\{Y_i(1; \bs{\theta}_N)\}\\
			{\E}_N^0 \{Y_i(0; \bs{\theta}_N)\} - \E_N\{Y_i(0; \bs{\theta}_N)\}
		\end{pmatrix}
		\converged
		\mathcal{N}
		\left( 
		\bs{0}, \ 
		\begin{pmatrix}
			( \tilde{r}_1^{-1} - 1 )\Omega_{11} & - \Omega_{10}\\
			- \Omega_{01} & ( \tilde{r}_0^{-1} - 1 ) \Omega_{00}
		\end{pmatrix}
		\right).
	\end{align*}
	By \eqref{eq:MA_e_eq} and Slutsky's theorem, 
	as $N\rightarrow \infty$, 
	$\sqrt{N}(\bar{Y}_1(\hat{\theta}_N) - \bar{Y}(1), \bar{Y}_0(\hat{\theta}_N) - \bar{Y}(0))^\top$ will converge to the same asymptotic distribution. 
	Therefore, Lemma \ref{lemma:maa_2} holds. 
\end{proof}

\begin{proof}[of Theorem \ref{thm:clt_mae}]
    By the same logic as the proof of Theorem \ref{thm:theta_clt}, it suffices to show that
    $\sqrt{N}(\hat{\tau}_{g\MA} - \tau_g)$ 
    and $\mathcal{N}(0, V_{g\MA})$ converge weakly to the same limit 
    along any subsequence $\{N_t\}$ such that $r_z$, 
    $\bar{Y}(z)$, $\Var_N\{Y_i(z;\theta_N)\}$ and 
    $\Cov_N\{Y_i(1;\theta_N), Y_i(0;\theta_N)\}$    
    have limits, $z=0,1$.  Note that all these quantities or their entries are bounded under Condition \ref{cond:ma_stable}.

	Along the subsequence $\{N_t\}$, define 
	$\bs{\xi}_{N_t} = (\bar{Y}(1), \bar{Y}(0))^\top$, 
	$\bs{\xi}_{\infty} \equiv (\mu_1, \mu_0)^\top = \lim_{t\rightarrow \infty} \bs{\xi}_{N_t}$, 
	$\hat{\bs{\xi}}_{N_t} = (\bar{Y}_1(\hat{\theta}_{N_t}), \bar{Y}_0(\hat{\theta}_{N_t}))^\top$, 
	and $G((\xi_1, \xi_2)^\top) = g(\xi_1) - g(\xi_2)$. 
    Similar to Lemma \ref{lemma:maa_2},
	let $\tilde{r}_0,$ $\tilde{r}_1,$ $\Omega_{11}, \Omega_{00}$ and $\Omega_{10}=\Omega_{01}$ 
    be the limits of $r_1$, $r_0,$
    $\Var_N\{Y_i(1;\theta_N)\}$, $\Var_N\{Y_i(0;\theta_N)\}$ and 
    $\Cov_N\{Y_i(1;\theta_N), Y_i(0;\theta_N)\}$, respectively, as $t\rightarrow \infty$. 
    Define further 
	$$
	\tilde{\bs{\Omega}} = 
	\begin{pmatrix}
		( \tilde{r}_1^{-1} - 1 )\Omega_{11} & - \Omega_{10}\\
		- \Omega_{01} & ( \tilde{r}_0^{-1} - 1 ) \Omega_{00}
	\end{pmatrix}. 
	$$

	First, we prove  $
	\sqrt{N} (\hat{\bs{\xi}}_{N_t} - \bs{\xi}_{N_t}) \converged \mathcal{N}(\bs{0}, \tilde{\bs{\Omega}})
	$
	using Lemmas \ref{lemma:maa_1} and \ref{lemma:maa_2}. 
	It suffices to verify the conditions in Lemmas \ref{lemma:maa_1} and \ref{lemma:maa_2}. 
	We can verify that 
	Condition \ref{cond:ma_continuous} implies conditions (iii) and (iv) in Lemma \ref{lemma:maa_1}, 
	and 
	Condition \ref{cond:ma_stable} and the property of the subsequence $\{N_t\}$ imply conditions (i), (ii) and (v) in Lemma \ref{lemma:maa_1} 
	and conditions (i) and (ii) in Lemma \ref{lemma:maa_2}.
	Besides, condition (vi) in Lemma \ref{lemma:maa_1} holds by the condition in Theorem \ref{thm:clt_mae}.

    Second, we show the asymptotic distribution of $\sqrt{{N_t}}(\hat{\tau}_{g\MA} - \tau_g) = \sqrt{{N_t}}\{G(\hat{\xi}_{N_t}) - G(\xi_{N_t})\}$ along the subsequence $\{N_t\}$. 
	By the mean value theorem, there exists $\tilde{\bs{\xi}}_{N_t}$ between $\hat{\bs{\xi}}_{N_t}$ and $\bs{\xi}_{N_t}$ such that, 
	\begin{align}\label{eq:mean_val_G_MA}
		G(\hat{\bs{\xi}}_{N_t}) - G(\bs{\xi}_{N_t})
		= 
		\dot G(\tilde{\bs{\xi}}_{N_t})^\top (\hat{\bs{\xi}}_{N_t} - \bs{\xi}_{N_t}). 
	\end{align}
	By the asymptotic Normality of $\hat{\bs{\xi}}_{N_t}$, 
	$$
	\| \tilde{\bs{\xi}}_{N_t} - \bs{\xi}_\infty \| 
	\le 
	\| \tilde{\bs{\xi}}_{N_t} - \bs{\xi}_{N_t} \| + \| \bs{\xi}_{N_t} - \bs{\xi}_\infty \| 
	\le 
	\| \hat{\bs{\xi}}_{N_t} - \bs{\xi}_{N_t} \| + \| \bs{\xi}_{N_t} - \bs{\xi}_\infty \| 
	= o_{\Pr}(1).
	$$
	By the continuous differentiability of $g$ in Condition \ref{cond:ma_continuous} and the continuous mapping theorem, 
	$\dot G(\tilde{\bs{\xi}}_{N_t}) = \dot G(\bs{\xi}_\infty) + o_{\Pr}(1)$. 
	From \eqref{eq:mean_val_G_MA}, the asymptotic Normality of $\hat{\bs{\xi}}_N$ and Slutsky's theorem, we then have 
	\begin{align}\label{eq:mae_asymp_proof1}
		\sqrt{{N_t}}(\hat{\tau}_{g\MA} - \tau_g)
		& = \sqrt{{N_t}}\{ G(\hat{\bs{\xi}}_{N_t}) - G(\bs{\xi}_{N_t}) \}
		= 
		\{ \dot G(\bs{\xi}_\infty) + o_{\Pr}(1) \}^\top \cdot \sqrt{{N_t}} (\hat{\bs{\xi}}_{N_t} - \bs{\xi}_{N_t})
		\nonumber
		\\
		& =
		\dot G(\bs{\xi}_\infty)^\top \sqrt{{N_t}} (\hat{\bs{\xi}}_{N_t} - \bs{\xi}_{N_t})
		+ 
		o_{\Pr}(1)
		\converged 
		\mathcal{N}\left( 0, \dot G(\bs{\xi}_\infty)^\top \tilde{\bs{\Omega}} \dot G(\bs{\xi}_\infty) \right).
	\end{align}
	By definition, 
	the asymptotic variance 
	has the following equivalent forms: 
	\begin{align*}%
		& \quad \ \dot G(\bs{\xi}_\infty)^\top \tilde{\bs{\Omega}} \dot G(\bs{\xi}_\infty)
		\nonumber
		\\
		& = 
		\begin{pmatrix}
			\dot g(\mu_1) &- \dot g(\mu_0)
		\end{pmatrix}
		\begin{pmatrix}
			( \tilde{r}_1^{-1} - 1 )\Omega_{11} & - \Omega_{10}\\
			- \Omega_{01} & ( \tilde{r}_0^{-1} - 1 ) \Omega_{00}
		\end{pmatrix}
		\begin{pmatrix}
			\dot g(\mu_1)\\-\dot g(\mu_0)
		\end{pmatrix}
		\nonumber
        \\
		& = 
		\tilde{r}_1^{-1} \{\dot g(\mu_1)\}^2  \Omega_{11} 
		+ 
		\tilde{r}_0^{-1} \{\dot g(\mu_0)\}^2  \Omega_{00} 
		- 
		\left[
		\{\dot g(\mu_1)\}^2  \Omega_{11}  + \{\dot g(\mu_0)\}^2  \Omega_{00} - 2 \dot g(\mu_1) \dot g(\mu_0) \Omega_{10}
		\right].
	\end{align*}

    Third, by the property of the subsequence $\{N_t\}$, 
    $V_{g\MA}$ converges to $\dot G(\bs{\xi}_\infty)^\top \tilde{\bs{\Omega}} \dot G(\bs{\xi}_\infty)$ as $t\rightarrow \infty$.

	From the above, 
	we can then derive Theorem \ref{thm:clt_mae}. 
\end{proof}

\subsection{Proof of Theorem \ref{thm:MA_ci}}

\begin{proof}[of Theorem \ref{thm:MA_ci}]
    By similar logic as the proof of Theorem \ref{thm:theta_clt}, it suffices to show that $\hat{V}_{g\MA} = \tilde{V}_{g\MA} + o_{\Pr}(1)$  along any subsequence $\{N_t\}$ such that $r_z$, 
    $\bar{Y}(z)$, $\Var_N\{Y_i(z;\theta_N)\}$ and 
    $\Cov_N\{Y_i(1;\theta_N), Y_i(0;\theta_N)\}$    
    have limits, $z=0,1$. Note that all these quantities or their entries are bounded under Condition \ref{cond:ma_stable}.

    First, we prove that $g(\bar{Y}_z(\hat{\theta}_{N_t})) = g(\bar{Y}(z)) + o_{\Pr}(1)$ for $z=0,1$. From the proof of Theorem \ref{thm:clt_mae}, for $z=0,1$,  
    $\bar{Y}_z(\hat{\theta}_{N_t}) = \bar{Y}(z) + o_{\Pr}(1)$. 
    By the continuity of $g$ in Condition \ref{cond:ma_continuous} and the fact that $\bar{Y}(z)$ has a limit along the subsequence $\{N_t\}$, we can derive that $g(\bar{Y}_z(\hat{\theta}_{N_t})) = g(\bar{Y}(z)) + o_{\Pr}(1)$ as $t\rightarrow \infty$.

    Second, we prove that $\Var_{N_t}^z\{Y_i(\hat{\theta}_{N_t})\} = \Var_{N_t}\{Y_i(z;\theta_{N_t})\} + o_{\Pr}(1)$ for $z=0,1$. 
    We consider only the case with $z=1$, since the proof for the other case with $z=0$ is almost the same. 
    Define $n=n_1$, $f_{ai}(\theta) = f_{bi}(\theta) = Y_i(1) - h_1(X_i, \theta)$. 
    We then have $\Var_{N_t}^1\{Y_i(\hat{\theta}_{N_t})\} = \Cov_{N_t}^1\{f_{ai}(\hat{\theta}_{N_t}), f_{bi}(\hat{\theta}_{N_t})\}$
    and 
    $\Var_{N_t}\{Y_i(1;\theta_{N_t})\} = \Cov_{N_t}\{f_{ai}(\theta_{N_t}), f_{bi}(\theta_{N_t})\}$. 
    Thus, 
    from Lemma \ref{lemma:consist_of_var_e_est_para}, 
    to prove that $\Var_{N_t}^1\{Y_i(\hat{\theta}_{N_t})\} = \Var_{N_t}\{Y_i(1;\theta_{N_t})\} + o_{\Pr}(1)$, 
    it suffices to verify the conditions in Lemmas \ref{lemma:consist_of_var_e} and  \ref{lemma:consist_of_var_e_est_para}. 
    By definition, conditions (i)-(iii) in Lemma \ref{lemma:consist_of_var_e} follow from Condition \ref{cond:ma_stable}  and the property of the subsequence $\{N_t\}$, 
    condition (i) in Lemma \ref{lemma:consist_of_var_e_est_para} follows from Condition \ref{cond:ma_continuous}, 
    and condition (ii) in Lemma \ref{lemma:consist_of_var_e_est_para} follows from the condition in Theorem \ref{thm:MA_ci}. 
    
    From the above, we can then derive that $\hat{V}_{g\MA} = \tilde{V}_{g\MA} + o_{\Pr}(1)$ along the subsequence $\{N_t\}$. 
    Therefore, Theorem \ref{thm:MA_ci} holds.
\end{proof}

\subsection{Proof of Theorem \ref{thm:ma_optimal}}
\begin{proof}[of Theorem \ref{thm:ma_optimal}]
	Since the adjustment functions $h_1(\bs{x}; \bs{\theta})$ and $h_0(\bs{x}; \bs{\theta})$ depend on disjoint subsets of the parameter $\bs{\theta}$ and both of them have intercept terms, 
	for descriptive simplicity, 
	we decompose $\bs\theta$ to into $(\alpha_1, \beta_1)$ and $(\alpha_0, \beta_0)$, and will write  $h_z(x;\theta)$ equivalently as $\alpha_z + \tilde{h}_z(x; \beta_z)$ for $z=0,1$.
	Besides, we use $\alpha_{N1}, \beta_{N1}, \alpha_{N0}$ and $\beta_{N0}$ to denote the corresponding subvectors of $\theta_N$ from minimizing $M_N(\theta)$ with the squared losses.

	For $z=0,1$, by definition, we can know that $\theta_N$ minimizes $\E_N [ \{Y_i(z) - h_z(X_i;\theta)\}^2 ]$ over all $\theta$. 
	Besides, since $h_z$ contains an intercept term, we must have 
	$\E_N \{Y_i(z) - h_z(X_i;\theta)\} = 0$. 
	Thus, for $z=0,1$ and any $\theta$, we have 
	\begin{align*}
	    \Var_N\{Y_i(z;\theta)\}
	    & = \Var_N\{Y_i(z) - \tilde{h}_z(X_i; \beta_z)\}
	    = 
	    \frac{N}{N-1}
	    \min_{\alpha_z} \E_N \big[ \{Y_i(z) - \alpha_z - \tilde{h}_z(X_i; \beta_z)\}^2 \big] 
	    \\
	    & \ge \frac{N}{N-1} \min_{\alpha_z, \beta_z} \E_N \big[ \{Y_i(z) - \alpha_z - \tilde{h}_z(X_i; \beta_z)\}^2 \big] 
	    \\
	    & = \frac{N}{N-1} \min_{\theta} \E_N \big[ \{Y_i(z) - h_z(X_i;\theta)\}^2 \big] 
	    = \frac{N}{N-1} \E_N \big[ \{Y_i(z) - h_z(X_i;\theta_N)\}^2 \big]
	    \\ 
	    & = 
	    \Var_N\{Y_i(z;\theta_N)\}. 
	\end{align*}
	This immediately implies that $\theta_N$ from minimizing $M_N(\theta_N)$ with squared losses must also minimize the probability limit $\tilde{V}_{g\MA}$ in \eqref{eq:Vtilde_gMA} of the conservative variance estimator. 
	
	From the above, Theorem \ref{thm:ma_optimal} holds. 
\end{proof}

\section{Proofs of Theorems \ref{thm:clt_mbe}, \ref{thm:clt_ibe} \& \ref{thm:clt_ai} and Corollaries \ref{cor:glm_est_eqn}, \ref{cor:model_imp_consist} \& \ref{cor:equi_i_a} }\label{sec:proof_thm_cor_A}

\subsection{Proof of Corollary \ref{cor:glm_est_eqn}}

\begin{proof}[of Corollary \ref{cor:glm_est_eqn}]
    For $z=0,1$ and by the density form in \eqref{eq:glm}, the partial derivative of $\loss_z(y,x;\theta) = -\log f_z(y,x;\theta)$ over $\alpha_z$ has the following equivalent forms: 
    \begin{align*}
        \frac{\partial}{\partial\alpha_z} \loss_z(y,x;\theta)
        & = 
        \frac{\partial}{\partial\alpha_z}
        \left\{ - \frac{y (\alpha_z + \bs{\beta}_z^\top \bs{x}) - b_z (\alpha_z + \bs{\beta}_z^\top \bs{x})}{a_z(\bs{\phi}_z)} - c_z(y, \bs{\phi}_z) \right\}
        = 
        - \frac{y  - \dot{b}_z(\alpha_z + \bs{\beta}_z^\top \bs{x})}{a_z(\bs{\phi}_z)} 
        \\
        & = 
        -1/a_z(\bs{\phi}_z) \cdot \{y  - h_z(x;\theta)\}. 
    \end{align*}
    This then implies that $M_N(\theta)$ defined as in  \eqref{eq:M_population} satisfies: 
    \begin{align*}
        \frac{\partial}{\partial\alpha_z} M_N(\theta) 
        & = r_z \E_N\Big\{ \frac{\partial}{\partial\alpha_z} \loss_z(Y_i(z),X_i;\theta) \Big\} + r_{1-z} \E_N\Big\{ \frac{\partial}{\partial\alpha_z} \loss_{1-z}(Y_i(1-z),X_i;\theta) \Big\}\\
        & = -r_z/a_z(\bs{\phi}_z) \cdot \E_N\{Y_i(z)  - h_z(X_i;\theta)\}.  
    \end{align*}
    Under Conditions \ref{cond:population_minimum}--\ref{cond:empirical_minimum}, $\Psi_N(\theta_N) = 0$.
    Because $\partial M_N(\theta)/{\partial\alpha_z}$ is a coordinate of $\Psi_N(\theta_N)$, we must have 
    $\E_N\{Y_i(z)  - h_z(X_i;\theta_N)\} = 0$ for $z=0,1$. 
    Therefore, for $z=0,1$, 
    $\E_N\{Y_i(z)\} = \E_N\{h_z(X_i;\theta_N)\} = \E_N \{ \dot{b}_z(\alpha_{Nz} + \bs{\beta}_{Nz}^\top \bs{X}_i) \}$, 
    i.e., Corollary \ref{cor:glm_est_eqn} holds. 
\end{proof}

\subsection{Proofs of Theorems \ref{thm:clt_mbe} and \ref{thm:clt_ibe}}

\begin{proof}[of Theorem \ref{thm:clt_mbe}]
    By the same logic as the proof of Theorem \ref{thm:theta_clt}, it suffices to show that
    $\sqrt{N}( \hat{\tau}_{g\MB}-\tau_{g\MB})$ 
    and $\mathcal{N}
		( 
		0,
		\dot G_{N\MB}(\bs{\theta}_N)^\top \bs{\Sigma}_N \dot G_{N\MB}(\bs{\theta}_N)
		)$ converge weakly to the same limit 
    along any subsequence $\{N_t\}$ such that $\Sigma_N$ and 
    $\dot G_{N\MB}(\bs{\theta}_N)$   
    have limits.  Note that these quantities or their entries are bounded under the condition in Theorem \ref{thm:clt_mbe} and Condition \ref{cond:gh_MB}(iv).
    Furthermore, 
	by the same logic as the proof of Theorem \ref{thm:rate}, 
    we can assume that 
	$\hat{\bs{\theta}}_{N_t} \in \overline{\mathcal{B}}(\bs{\theta}_{N_t}, \varepsilon_0/2) \equiv \{\theta: \|\theta-\theta_N\|\le \varepsilon_0/2\} \subset \Theta$ for all $N_t$.

	By the mean value theorem, 
	there exists $\tilde{{\bs\theta}}_{N_t}$ between $\hat{\bs\theta}_{N_t}$ and ${\bs\theta}_{N_t}$ such that
	$
	G_{{N_t}\MB}(\hat{{\bs\theta}}_{N_t}) - G_{{N_t}\MB}(\bs{\theta}_{N_t})
	=
	\dot G_{{N_t}\MB}(\tilde{{\bs\theta}}_{N_t})^\top (\hat{\bs\theta}_{N_t}-{\bs\theta}_{N_t}). 
	$
	By the asymptotic Normality of $\hat{\bs\theta}_{N_t}$ and Condition \ref{cond:gh_MB}(iii), 
	for any $\varepsilon>0$, there exists $\delta>0$ such that 
	\begin{align*}
		\Pr( \| \dot G_{{N_t}\MB}(\tilde{{\bs\theta}}_{N_t}) - \dot G_{{N_t}\MB}(\bs{\theta}_{N_t}) \| > \varepsilon )
		& \le 
		\Pr( \| \tilde{{\bs\theta}}_{N_t} - \bs{\theta}_{N_t} \| > \delta ) 
		\le 
		\Pr( \| \hat{{\bs\theta}}_{N_t} - \bs{\theta}_{N_t} \| > \delta ), 
	\end{align*}
	which converges to zero as $N \rightarrow \infty$. 
    Thus, 
	\begin{align*}
		G_{{N_t}\MB}(\hat{{\bs\theta}}_{N_t}) - G_{{N_t} \MB}(\bs{\theta}_{N_t})
		& =
		 \{ \dot G_{{N_t}\MB}({\bs\theta}_{N_t})+o_\Pr(1)\}^\top (\hat{\bs\theta}_{N_t}-{\bs\theta}_{N_t})
		\\
        & = 
		\dot G_{{N_t}\MB}({\bs\theta}_{N_t})^\top (\hat{\bs\theta}_{N_t} - {\bs\theta}_{N_t})  + o_\Pr({N_t}^{-1/2}),
	\end{align*}
	where the last equality holds because $\hat{\bs\theta}_{N_t} - {\bs\theta}_{N_t} = O_\Pr({N_t}^{-1/2})$. 
	From the asymptotic Normality of $\hat{\bs\theta}_{N_t}$, the choice of the subsequence $\{N_t\}$, and Slutsky's theorem, 
	we can derive that 
	\begin{align*}
		\sqrt{{N_t}} \left( \hat{\tau}_{g\MB} -  \tau_{g\MB} \right)
		& = 
		\sqrt{{N_t}} \left\{ G_{{N_t}\MB}(\hat{{\bs\theta}}_{N_t}) - G_{{N_t}\MB}(\bs{\theta}_{N_t}) \right\}
		= \dot G_{{N_t}\MB}({\bs\theta}_{N_t})^\top \cdot \sqrt{{N_t}}(\hat{\bs\theta}_{N_t} - {\bs\theta}_{N_t}) + o_\Pr(1)
		\\
		& \ \dot\sim \  \mathcal{N}\left( 0, \ \dot G_{{N_t}\MB}({\bs\theta}_{N_t})^\top \bs\Sigma_{N_t} \dot G_{{N_t}\MB}({\bs\theta}_{N_t}) \right). 
	\end{align*}
	
     From the above, Theorem \ref{thm:clt_mbe} holds.  
\end{proof}

\begin{proof}[of Theorem \ref{thm:clt_ibe}]
	Note that $\hat{\tau}_{g\MI} - \tau_{g\MI} = G_{N\MI}(\hat{\bs{\theta}}_N) - G_{N\MI}(\bs{\theta}_N)$. 
	The proof of Theorem \ref{thm:clt_ibe} is the same as that of Theorem \ref{thm:clt_mbe}, and is thus omitted here. 
\end{proof}

\subsection{Proof of Corollary \ref{cor:model_imp_consist}}

\begin{proof}[of Corollary \ref{cor:model_imp_consist}]
	From Conditions \ref{cond:population_minimum}--\ref{cond:empirical_minimum} and Theorem \ref{thm:theta_clt}, 
	$\sqrt{N}(\hat{\bs{\theta}}_N - \bs{\theta}_N) \lpc \mathcal{N}(\bs{0}, \bs{\Sigma}_N)$, with $\sigma_{\max}(\Sigma_N)$ being bounded and 
	\begin{align*}
		\bs{\Sigma}_N = r_1 r_0 \{\ddot M_N(\bs{\theta_N})\}^{-1} \Cov_N\{ \dot \loss_{1\text{-}0, i}(\bs{\theta}_N) \}  \{\ddot M_N(\bs{\theta_N})\}^{-1}. 
	\end{align*}
	From Condition \ref{cond:MI_g_h} and Theorem \ref{thm:clt_ibe}, 
	\begin{align*}
		\sqrt{N}(\hat{\tau}_{g\MI}-\tau_{g\MI}) 
		& \lpc 
		\mathcal{N}\left( 
		0, \ 
		\dot G_{N\MI}(\bs{\theta}_N)^\top 
		\bs{\Sigma}_N \dot G_{N\MI}(\bs{\theta}_N)
		\right)\\
		& \sim 
		\mathcal{N}\left( 
		0, \ 
		r_1 r_0\dot G_{N\MI}(\bs{\theta}_N)^\top 
		\{\ddot M_N(\bs{\theta_N})\}^{-1} \Cov_N\{ \dot \loss_{1\text{-}0, i}(\bs{\theta}_N) \}  \{\ddot M_N(\bs{\theta_N})\}^{-1} \dot G_{N\MI}(\bs{\theta}_N)
		\right)\\
		& \sim 
		\mathcal{N}\left( 
		0, V_{g\MI}
		\right), 
	\end{align*}
	where 
	$
	\tau_{g\MI} 
	= 
	g( \E_N\{ h_1(\bs{X}_i; \bs{\theta}_N) \} ) - g( \E_N\{ h_0(\bs{X}_i; \bs{\theta}_N )\} ).
	$
	From Condition \ref{cond:consis_MI}, for $z=0,1$, 
	\begin{align*}
		& \quad \ \bar{Y}(z) - \E_N\{ h_z(\bs{X}_i; \bs{\theta}_N) \}
		= 
		\E_N \left\{
		Y_i(z) - h_z(\bs{X}_i; \bs{\theta}_N)
		\right\}
		\\
		& =
		r_z^{-1} \E_N \big\{
		r_z q_z(\bs{\theta}_N)^\top \psi_z(Y_i(z), \bs{X}_i; \bs{\theta}_N)
		+
		r_{1-z} q_z(\bs{\theta}_N)^\top \psi_{1-z}(Y_i(1-z), \bs{X}_i; \bs{\theta}_N)
		\big\}
		\\
		& = 
		r_z^{-1} q_z(\bs{\theta}_N)^\top \E_N \big\{
		r_z  \psi_z(Y_i(z), \bs{X}_i; \bs{\theta}_N)
		+
		r_{1-z} \psi_{1-z}(Y_i(1-z), \bs{X}_i; \bs{\theta}_N)
		\big\}
		\\
		& = 
		r_z^{-1} q_z(\bs{\theta}_N)^\top \Psi_N(\bs{\theta}_N) = 0. 
	\end{align*}
	Thus, $\tau_{g\MI}$ has the following equivalent forms:
	\begin{align*}
		\tau_{g, \MI} 
		= 
		g( \E_N\{ h_1(\bs{X}_i; \bs{\theta}_N) \} ) - g( \E_N\{ h_0(\bs{X}_i; \bs{\theta}_N )\} )
		= 
		g(\bar{Y}(1)) - g(\bar{Y}(0)) = \tau_g.
	\end{align*}
	From the above, Corollary \ref{cor:model_imp_consist} holds. 
\end{proof}

\subsection{Proof of Corollary \ref{cor:equi_i_a}}

\begin{proof}[of Corollary \ref{cor:equi_i_a}]
    First, we prove that $\hat{\tau}_{g\MI} - \hat{\tau}_{g\MA} = o_{\Pr}(N^{-1/2})$. 
	By Condition \ref{cond:consis_MI}, for $z=0,1$ and any $\theta\in \Theta$, 
	\begin{align*}
		& \quad \ \E_N^z \{ Y_i(z)-h_z(\bs{X}_i;\bs{\theta}) \} =
		\bs{q}_z(\bs{\theta})^\top \E_N^z  \{ \psi_z(Y_i(z), \bs{X}_i; \bs{\theta}) \}
		\\
		& = r_z^{-1}\bs{q}_z(\bs{\theta})^\top \cdot 
		\big[ r_z \E_N^z  \{ \psi_z(Y_i(z), \bs{X}_i; \bs{\theta}) \}
		+
		r_{1-z} \E_N^{1-z}  \{ \psi_{1-z}(Y_i(1-z), \bs{X}_i; \bs{\theta}) \}
		\big]
		\\
		&= r_z^{-1} \bs{q}_z(\bs{\theta})^\top {\hat \Psi}_N(\bs{\theta}).
	\end{align*}
    From Condition \ref{cond:empirical_minimum}, 
    $\hat \Psi_N(\bs{\theta}_N) = 0$, which immediately implies that 
    $\E_N^z \{ Y_i(z)-h_z(\bs{X}_i;\hat{\bs{\theta}}_N) \} = 0$ for $z=0,1$.     
    By definition, this further implies that, for $z=0,1$, 
    \begin{align*}
        \bar{Y}_z(\hat{\theta}_N)
        & = \E_N^z \{ Y_i(z)-h_z(\bs{X}_i;\hat{\bs{\theta}}_N) \} + 
        \E_N\{ h_z(\bs{X}_i;\hat{\bs{\theta}}_N) \} 
        = \E_N\{ h_z(\bs{X}_i;\hat{\bs{\theta}}_N) \}. 
    \end{align*}
    Consequently, $\hat{\tau}_{g\MI}=\hat{\tau}_{g\MA}$.

	Second, we prove that $V_{g\MI} = V_{g\MA}$. 
    From Condition \ref{cond:consis_MI} and by the same logic as the proof of Corollary \ref{cor:model_imp_consist}, for $z=0,1$, 
    $\E_N\{ Y_i(z) - h_z(\bs{X}_i;\bs{\theta})\} = r_z^{-1} q_z(\theta)^\top \Psi_N(\theta)$. 
	Taking derivative on both sides, we then have
	\begin{align*}
		-\E_N\{ \dot h_z(\bs{X}_i;\bs{\theta})\} =r_z^{-1} \dot \Psi_N(\bs\theta)^\top \bs{q}_z(\bs{\theta}) +r_z^{-1}\dot{\bs{q}}_z(\bs{\theta})^\top \Psi_N(\bs\theta)
	\end{align*}
    Because $\Psi_N(\bs\theta_N)=0$, we then have 
	$
		\E_N\{ \dot h_z(\bs{X}_i;\bs{\theta}_N) \}=-r_z^{-1} \dot \Psi_N(\bs{\theta}_N)^\top \bs{q}_z(\bs{\theta}).  
	$
	By the definition of $G_{N\MI}(\theta)$ and from Corollary \ref{cor:model_imp_consist}, we then have 
	\begin{align*}
		\dot G_{N\MI}(\bs{\theta}_N)&= \dot g(\E_N\{h_1(\bs X_i;\bs {\theta}_N)\})\E_N\{\dot h_1(\bs X_i;\bs {\theta}_N)\}-\dot g(\E_N\{h_0(\bs X_i;\bs {\theta}_N)\})\E_N\{\dot h_0(\bs X_i;\bs {\theta}_N)\}\\
		& = 
		\dot g(\bar{Y}(1))\E_N\{\dot h_1(\bs X_i;\bs {\theta}_N)\}-\dot g(\bar{Y}(0))\E_N\{\dot h_0(\bs X_i;\bs {\theta}_N)\}
		\\
		&= -r_1^{-1} \dot g(\bar{Y}(1)) \dot \Psi_N(\bs{\theta}_N)^\top \bs{q}_1(\bs{\theta}) 
		+ r_0^{-1} \dot g(\bar{Y}(0))  \dot \Psi_N(\bs{\theta}_N)^\top \bs{q}_0(\bs{\theta})\\
		& = \dot \Psi_N(\bs{\theta}_N)^\top \big\{ r_0^{-1} \dot g(\bar{Y}(0)) \bs{q}_0(\bs{\theta})-r_1^{-1} \dot g(\bar{Y}(1)) \bs{q}_1(\bs{\theta}) \big\}\\
        & = 
        \dot \Psi_N(\bs{\theta}_N)^\top 
        \begin{pmatrix}
        q_1(\theta_N) & q_0(\theta_N)
        \end{pmatrix}
        \begin{pmatrix}
        -r_1^{-1} \dot g(\bar{Y}(1))\\
        r_0^{-1} \dot g(\bar{Y}(0))
        \end{pmatrix}.
	\end{align*}
	We can then simplify $V_{g\MI}$ in Corollary \ref{cor:model_imp_consist} as 
	\begin{align*}
		& \quad \ V_{g\MI} 
        \\ 
        & = r_1 r_0 \dot G_{N\MI}(\bs{\theta}_N) ^\top 
		\{\dot \Psi_N(\bs{\theta_N})\}^{-1} \Cov_N\left\{ \psi_{1i}(\bs{\theta}_N) - \psi_{0i}(\bs{\theta}_N) \right\}  \{\dot \Psi_N(\bs{\theta_N})^\top \}^{-1}  \dot G_{N\MI}(\bs{\theta}_N)\\
		&=r_1 r_0 
		\begin{pmatrix} 
		-r_1^{-1} \dot g(\bar{Y}(1)) \\
		 r_0^{-1} \dot g(\bar{Y}(0))
		\end{pmatrix}^\top
		\begin{pmatrix}
		\bs{q}_1(\bs{\theta}_N)^\top\\\bs{q}_0(\bs{\theta}_N)^\top
		\end{pmatrix} 
		\Cov_N\{ \psi_{1i}(\bs{\theta}_N) - \psi_{0i}(\bs{\theta}_N)\}
		\begin{pmatrix}
		\bs{q}_1(\bs{\theta}_N)&\bs{q}_0(\bs{\theta}_N)
		\end{pmatrix}
		\begin{pmatrix} 
		- r_1^{-1} \dot g(\bar{Y}(1))  \\ r_0^{-1} \dot g(\bar{Y}(0))
		\end{pmatrix}\\
		&=r_1 r_0 
		\begin{pmatrix} 
		- r_1^{-1} \dot g(\bar{Y}(1)) \\
		r_0^{-1} \dot g(\bar{Y}(0))
		\end{pmatrix}^\top
		\Cov_N\left\{ 
		\begin{pmatrix}
		\bs{q}_1(\bs{\theta}_N)^\top \{\psi_{1i}(\bs{\theta}_N) - \psi_{0i}(\bs{\theta}_N)\} \\
		\bs{q}_0(\bs{\theta}_N)^\top \{\psi_{1i}(\bs{\theta}_N) - \psi_{0i}(\bs{\theta}_N)\}
		\end{pmatrix} 
		\right\}
		\begin{pmatrix} 
		- r_1^{-1} \dot g(\bar{Y}(1))  \\ r_0^{-1} \dot g(\bar{Y}(0))
		\end{pmatrix}
	\end{align*}
	By Condition \ref{cond:consis_MI}, for $1\le j\le N$, $\bs{q}_1(\bs{\theta}_N)^\top \{\psi_{1j}(\bs{\theta}_N) - \psi_{0j}(\bs{\theta}_N)\}=Y_j(1; \bs{\theta}_N) - \E_N\{h_1(X_i;\theta_N)\}$, and 
	$\bs{q}_0(\bs{\theta}_N)^\top  \{\psi_{1j}(\bs{\theta}_N) - \psi_{0j}(\bs{\theta}_N)\}=-Y_j(0; \bs{\theta}_N) + \E_N\{h_0(X_i;\theta_N)\}$. 
	We can then further simplify $V_{g\MI}$ as
	\begin{align*}
	    V_{g\MI} & = 
	    r_1 r_0 
		\begin{pmatrix} 
		- r_1^{-1} \dot g(\bar{Y}(1)) \\
		r_0^{-1} \dot g(\bar{Y}(0))
		\end{pmatrix}^\top
		\Cov_N\left\{ 
		\begin{pmatrix}
		Y_i(1;\theta_N)\\
		-Y_i(0;\theta_N)
		\end{pmatrix} 
		\right\}
		\begin{pmatrix} 
		- r_1^{-1} \dot g(\bar{Y}(1))  \\ r_0^{-1} \dot g(\bar{Y}(0))
		\end{pmatrix}\\
		& =
		r_1 r_0 
		\Var_N\left\{ 
		 r_1^{-1} \dot g(\bar{Y}(1)) Y_i(1;\theta_N) 
		 + r_0^{-1} \dot g(\bar{Y}(0)) Y_i(0;\theta_N)
		\right\}\\
		& = 
		\frac{r_0}{r_1} \Var_N\left\{ 
		 r_1^{-1} \dot g(\bar{Y}(1)) Y_i(1;\theta_N) 
		\right\}
		+
		\frac{r_1}{r_0} \Var_N\left\{ 
	    r_0^{-1} \dot g(\bar{Y}(0)) Y_i(0;\theta_N)
		\right\}
		\\
		& \quad + 
		2 \Cov_N\{ \dot g(\bar{Y}(1)) Y_i(1;\theta_N) ,  \dot g(\bar{Y}(0)) Y_i(0;\theta_N) \}
		\\
		& = r_1^{-1} \Var_N\{ \dot g(\bar{Y}(1))Y_i(1; \bs{\theta}_N) \} + 
		r_0^{-1} \Var_N\{ \dot g(\bar{Y}(0))Y_i(0; \bs{\theta}_N) \} 
		\nonumber
		\\
		& \quad \ 
		- 
		\Var_N\{ \dot g(\bar{Y}(1))Y_i(1; \bs{\theta}_N) - \dot g(\bar{Y}(0))Y_i(0; \bs{\theta}_N)  \} \\
		& = V_{g\MA}. 
	\end{align*}
	
    From the above, Corollary \ref{cor:equi_i_a} holds. 
\end{proof}

\subsection{Proof of Theorem \ref{thm:clt_ai}}\label{app:ai}

\begin{proof}[of Theorem \ref{thm:clt_ai}]

By the same logic as the proof of Theorem \ref{thm:theta_clt}, it suffices consider any subsequence such that $r_1, r_0, \bar{Y}(1), \bar{Y}(0)$ and the finite population variances and covariances for $Y_i(z)$, $h_{j1}^\MI(X_i; \theta_{Nj}^\MI)$ and $h_{j0}^\MI(X_i; \theta_{Nj}^\MI)$ $(z=0,1; 1 \le j \le J)$ have limits. Note that these quantities are all bounded under the conditions in Theorem \ref{thm:clt_ai}. 
For notational simplicity, we will simply use $\{N\}$ to denote such a subsequence.

Recall that $\hat{\theta}_N^{\MA}$ is the minimizer of the empirical risk function in \eqref{eq:M_hat} with square losses and transformed covariates $h_{j0}^\MI (X_i; \hat{\theta}^\MI_{Nj})$ and $h_{j1}^\MI (X_i; \hat{\theta}^\MI_{Nj})$, where the parameters $\hat{\theta}^\MI_{Nj}$'s are estimated from the data. 
Let $\tilde{\theta}_N^{\MA}$ be the minimizer of the empirical risk function in \eqref{eq:M_hat} with square losses and transformed covariates $h_{j0}^\MI (X_i; \theta^\MI_{Nj})$ and $h_{j1}^\MI (X_i; \theta^\MI_{Nj})$, where the parameters are fixed and are the probability limits of the corresponding estimated ones. 
We further use $\theta_N^{\MA}$ to denote the minimizer of the finite population risk function in \eqref{eq:M_population} with transformed covariates $h_{j0}^\MI (X_i; \theta^\MI_{Nj})$ and $h_{j1}^\MI (X_i; \theta^\MI_{Nj})$. 
For simplicity, we will use $\hat{\alpha}_{Nz}, \hat{\beta}_{Nzj}$ and $\hat{\gamma}_{Nzj}$ to denote elements of $\hat{\theta}_N^{\MA}$, 
$\tilde{\alpha}_{Nz}, \tilde{\beta}_{Nzj}$ and $\tilde{\gamma}_{Nzj}$ to denote elements of $\tilde{\theta}_N^{\MA}$, 
and 
${\alpha}_{Nz}, {\beta}_{Nzj}$ and ${\gamma}_{Nzj}$ to denote elements of ${\theta}_N^{\MA}$.

First, we prove that, for $z =0,1$,   
$\Var_N^z\{Y_i\} = \Var_N\{Y_i(z)\} + o_{\Pr}(1)$, 
\begin{align}\label{eq:AI_samp_var_est}
    \Var_N^z\{h_{j0}^\MI(X_i; \hat{\theta}_{Nj}^\MI)\} & = \Var_N\{h_{j0}^\MI(X_i; {\theta}_{Nj}^\MI)\} + o_{\Pr}(1), 
    \nonumber
    \\
    \Var_N^z\{h_{j1}^\MI(X_i; \hat{\theta}_{Nj}^\MI)\} & = \Var_N\{h_{j1}^\MI(X_i; {\theta}_{Nj}^\MI)\} + o_{\Pr}(1), 
    \nonumber
    \\
    \Cov_N^z\{Y_i, h_{j0}^\MI(X_i; \hat{\theta}_{Nj}^\MI)\} & = \Cov_N\{Y_i(z), h_{j0}^\MI(X_i; {\theta}_{Nj}^\MI)\} + o_{\Pr}(1), 
    \nonumber
    \\
    \Cov_N^z\{Y_i, h_{j1}^\MI(X_i; \hat{\theta}_{Nj}^\MI)\} & = \Cov_N\{Y_i(z), h_{j1}^\MI(X_i; {\theta}_{Nj}^\MI)\} + o_{\Pr}(1),  
    \nonumber
    \\
    \Cov_N^z\{h_{j0}^\MI(X_i; \hat{\theta}_{Nj}^\MI), h_{j1}^\MI(X_i; \hat{\theta}_{Nj}^\MI)\}
    & =
    \Cov_N\{h_{j0}^\MI(X_i; {\theta}_{Nj}^\MI), h_{j1}^\MI(X_i; {\theta}_{Nj}^\MI)\} + o_{\Pr}(1), 
\end{align}
and 
\begin{align}\label{eq:AI_samp_var_true}
    \Var_N^z\{h_{j0}^\MI(X_i; {\theta}_{Nj}^\MI)\} & = \Var_N\{h_{j0}^\MI(X_i; {\theta}_{Nj}^\MI)\} + o_{\Pr}(1), 
    \nonumber
    \\
    \Var_N^z\{h_{j1}^\MI(X_i; {\theta}_{Nj}^\MI)\} & = \Var_N\{h_{j1}^\MI(X_i; {\theta}_{Nj}^\MI)\} + o_{\Pr}(1), 
    \nonumber
    \\
    \Cov_N^z\{Y_i, h_{j0}^\MI(X_i; {\theta}_{Nj}^\MI)\} & = \Cov_N\{Y_i(z), h_{j0}^\MI(X_i; {\theta}_{Nj}^\MI)\} + o_{\Pr}(1), 
    \nonumber
    \\
    \Cov_N^z\{Y_i, h_{j1}^\MI(X_i; {\theta}_{Nj}^\MI)\} & = \Cov_N\{Y_i(z), h_{j1}^\MI(X_i; {\theta}_{Nj}^\MI)\} + o_{\Pr}(1),  
    \nonumber
    \\
    \Cov_N^z\{h_{j0}^\MI(X_i; {\theta}_{Nj}^\MI), h_{j1}^\MI(X_i; {\theta}_{Nj}^\MI)\}
    & =
    \Cov_N\{h_{j0}^\MI(X_i; {\theta}_{Nj}^\MI), h_{j1}^\MI(X_i; {\theta}_{Nj}^\MI)\} + o_{\Pr}(1),
\end{align}
Under Condition \ref{cond:impute_adjust} and given that $\hat{\theta}_{Nj}^\MI - {\theta}_{Nj}^\MI = O_{\Pr}(N^{-1/2}) = o_{\Pr}(1)$, these follows immediately from Lemmas \ref{lemma:consist_of_var_e} and \ref{lemma:consist_of_var_e_est_para}. 

Second, we prove that, for $z=0,1$ and $1\le j\le J$, $\hat{\beta}_{Nzj} = {\beta}_{Nzj} + o_{\Pr}$, $\hat{\gamma}_{Nzj} = \gamma_{Nzj} + o_{\Pr}(1)$, 
$\tilde{\beta}_{Nzj} = {\beta}_{Nzj} + o_{\Pr}$, $\tilde{\gamma}_{Nzj} = \gamma_{Nzj} + o_{\Pr}(1)$. 
These follows immediately from the forms of least squares estimators and the discussion before. 

Third, we prove that, for $z=0,1$ and $1\le j \le J$, 
\begin{align*}
    \E_N^z \{h_{j0}^{\MI}(X_i, \hat{\theta}_{Nj}^\MI)\} - \E_N\{ h_z(X_i; \hat{\theta}_{Nj}^\MI) \} = 
    \E_N^z \{h_{j0}^{\MI}(X_i, {\theta}_{Nj}^\MI)\} - \E_N\{ h_z(X_i; {\theta}_{Nj}^\MI) \} + o_{\Pr}(N^{-1/2}).
\end{align*}
Under Condition \ref{cond:impute_adjust} and given that $\hat{\theta}_{Nj}^\MI - {\theta}_{Nj}^\MI = O_{\Pr}(N^{-1/2})$, 
these follow from Lemma \ref{lemma:maa_1}. 

Fourth, we prove that, for $z=0,1$, 
\begin{align}\label{eq:h_z_ENz_EN}
    \E_N^z\{h_z^\MA(X_i; \hat{\theta}_N^\MA)\} - \E_N\{h_z^\MA(X_i; \hat{\theta}_N^\MA)\}
    & = 
    \E_N^z\{\tilde{h}_z^\MA(X_i; {\theta}_N^\MA)\} - \E_N\{\tilde{h}_z^\MA(X_i; {\theta}_N^\MA)\} + o_{\Pr}(N^{-1/2}), 
    \nonumber
    \\
    \E_N^z\{\tilde{h}_z^\MA(X_i; \tilde{\theta}_N^\MA)\} - \E_N\{\tilde{h}_z^\MA(X_i; \tilde{\theta}_N^\MA)\}
    & = 
    \E_N^z\{\tilde{h}_z^\MA(X_i; {\theta}_N^\MA)\} - \E_N\{\tilde{h}_z^\MA(X_i; {\theta}_N^\MA)\} + o_{\Pr}(N^{-1/2}). 
\end{align}
By symmetric, below we consider only the case with $z=1$. 
By definition, 
\begin{align*}
    \E_N^1\{h_1^\MA(X_i; \hat{\theta}_N^\MA)\} - \E_N\{h_1^\MA(X_i; \hat{\theta}_N^\MA)\}
    & = 
    \sum_{j=1}^J \hat{\beta}_{N1j} 
    \big[ \E_N^1 \{h_{0j}^\MI(X_i; \hat{\theta}_{Nj}^\MI)\} - \E_N \{h_{0j}^\MI(X_i; \hat{\theta}_{Nj}^\MI)\} \big]\\
    & \quad + 
    \sum_{j=1}^J \hat{\gamma}_{N1j}
    \big[ \E_N^1 \{h_{1j}^\MI(X_i; \hat{\theta}_{Nj}^\MI)\} - \E_N \{h_{1j}^\MI(X_i; \hat{\theta}_{Nj}^\MI)\} \big],  
\end{align*}
\begin{align*}
    \E_N^1\{\tilde{h}_1^\MA(X_i; \tilde{\theta}_N^\MA)\} - \E_N\{\tilde{h}_1^\MA(X_i; \tilde{\theta}_N^\MA)\}
    & = 
    \sum_{j=1}^J \tilde{\beta}_{N1j} 
    \big[ \E_N^1 \{h_{0j}^\MI(X_i; {\theta}_{Nj}^\MI)\} - \E_N \{h_{0j}^\MI(X_i; {\theta}_{Nj}^\MI)\} \big]\\
    & \quad + 
    \sum_{j=1}^J \tilde{\gamma}_{N1j}
    \big[ \E_N^1 \{h_{1j}^\MI(X_i; {\theta}_{Nj}^\MI)\} - \E_N \{h_{1j}^\MI(X_i; {\theta}_{Nj}^\MI)\} \big], 
\end{align*}
and
\begin{align*}
    \E_N^1\{\tilde{h}_1^\MA(X_i; {\theta}_N^\MA)\} - \E_N\{\tilde{h}_1^\MA(X_i; {\theta}_N^\MA)\}
    & = 
    \sum_{j=1}^J {\beta}_{N1j} 
    \big[ \E_N^1 \{h_{0j}^\MI(X_i; {\theta}_{Nj}^\MI)\} - \E_N \{h_{0j}^\MI(X_i; {\theta}_{Nj}^\MI)\} \big]\\
    & \quad + 
    \sum_{j=1}^J {\gamma}_{N1j}
    \big[ \E_N^1 \{h_{1j}^\MI(X_i; {\theta}_{Nj}^\MI)\} - \E_N \{h_{1j}^\MI(X_i; {\theta}_{Nj}^\MI)\} \big]. 
\end{align*}
Note that under Condition \ref{cond:impute_adjust}, by Chebyshev's inequality and the property of simple random sampling, $\E_N^1 \{h_{0j}^\MI(X_i; {\theta}_{Nj}^\MI)\} - \E_N \{h_{0j}^\MI(X_i; {\theta}_{Nj}^\MI)\} = O_{\Pr}(N^{-1/2})$. 
From the discussion before, we then have, for $1\le j \le J$, 
\begin{align*}
    & \quad \ \hat{\beta}_{N1j} 
    \big[ \E_N^1 \{h_{0j}^\MI(X_i; \hat{\theta}_{Nj}^\MI)\} - \E_N \{h_{0j}^\MI(X_i; \hat{\theta}_{Nj}^\MI)\} \big]\\
    & = \{ {\beta}_{N1j} + o_{\Pr}(1) \} \cdot
    \big\{ \big[ \E_N^1 \{h_{0j}^\MI(X_i; {\theta}_{Nj}^\MI)\} - \E_N \{h_{0j}^\MI(X_i; {\theta}_{Nj}^\MI)\} \big] + o_{\Pr}(N^{-1/2}) \big\}\\
    & = {\beta}_{N1j}  \big[ \E_N^1 \{h_{0j}^\MI(X_i; {\theta}_{Nj}^\MI)\} - \E_N \{h_{0j}^\MI(X_i; {\theta}_{Nj}^\MI)\} \big] + o_{\Pr}(N^{-1/2}), 
\end{align*}
and 
\begin{align*}
    & \quad \ \tilde{\beta}_{N1j} 
    \big[ \E_N^1 \{h_{0j}^\MI(X_i; {\theta}_{Nj}^\MI)\} - \E_N \{h_{0j}^\MI(X_i; {\theta}_{Nj}^\MI)\} \big]\\
    & = \{ {\beta}_{N1j} + o_{\Pr}(1) \} \cdot
    \big[ \E_N^1 \{h_{0j}^\MI(X_i; {\theta}_{Nj}^\MI)\} - \E_N \{h_{0j}^\MI(X_i; {\theta}_{Nj}^\MI)\} \big]
    \\
    & = {\beta}_{N1j} \big[ \E_N^1 \{h_{0j}^\MI(X_i; {\theta}_{Nj}^\MI)\} - \E_N \{h_{0j}^\MI(X_i; {\theta}_{Nj}^\MI)\} \big] + o_{\Pr}(N^{-1/2}). 
\end{align*}
Analogously, we have, for $1\le j \le J$, 
\begin{align*}
    & \quad \ \hat{\gamma}_{N1j}
    \big[ \E_N^1 \{h_{1j}^\MI(X_i; \hat{\theta}_{Nj}^\MI)\} - \E_N \{h_{1j}^\MI(X_i; \hat{\theta}_{Nj}^\MI)\} \big]\\
    & = {\gamma}_{N1j}
    \big[ \E_N^1 \{h_{1j}^\MI(X_i; {\theta}_{Nj}^\MI)\} - \E_N \{h_{1j}^\MI(X_i; {\theta}_{Nj}^\MI)\} \big] + o_{\Pr}(N^{-1/2}), 
\end{align*}
and 
\begin{align*}
    & \quad \ \tilde{\gamma}_{N1j}
    \big[ \E_N^1 \{h_{1j}^\MI(X_i; {\theta}_{Nj}^\MI)\} - \E_N \{h_{1j}^\MI(X_i; {\theta}_{Nj}^\MI)\} \big]\\
    & = {\gamma}_{N1j}
    \big[ \E_N^1 \{h_{1j}^\MI(X_i; {\theta}_{Nj}^\MI)\} - \E_N \{h_{1j}^\MI(X_i; {\theta}_{Nj}^\MI)\} \big] + o_{\Pr}(N^{-1/2}). 
\end{align*}
Therefore, \eqref{eq:h_z_ENz_EN} must hold. 

Fifth, for $z=0,1$, we define 
\begin{align*}
    Y_i(z; \hat{\theta}_N^\MA) & = Y_i(z) - h_z^\MA(X_i; \hat{\theta}_N^\MA) + \E_N\{h_z^\MA(X_i; \hat{\theta}_N^\MA)\}, \\
    \tilde{Y}_i(z; \tilde{\theta}_N^\MA) & = Y_i(z) - \tilde{h}_z^\MA(X_i; \tilde{\theta}_N^\MA) + \E_N\{ \tilde{h}_z^\MA(X_i; \tilde{\theta}_N^\MA)\}, \\
    \tilde{Y}_i(z; {\theta}_N^\MA) & = {Y}_i(z) - \tilde{h}_z^\MA(X_i; {\theta}_N^\MA) + \E_N\{\tilde{h}_z^\MA(X_i; {\theta}_N^\MA)\}, 
\end{align*}
and 
\begin{align*}
    \hat{\zeta}_{Nz} & \equiv \bar{Y}_{z}(\hat{\theta}_N^\MA) 
    \equiv n_z^{-1} \sum_{i:Z_i=z} Y_i(z; \hat{\theta}_N^\MA)
    = \bar{Y}_z - \E_N^z\{h_z^\MA(X_i; \hat{\theta}_N^\MA)\} + \E_N\{h_z^\MA(X_i; \hat{\theta}_N^\MA)\} \\
    \tilde{\zeta}_{Nz} & \equiv \bar{\tilde{Y}}_{z}(\tilde{\theta}_N^\MA) 
    \equiv n_z^{-1} \sum_{i:Z_i=z} \tilde{Y}_i(z; \tilde{\theta}_N^\MA)
    = \bar{Y}_z - \E_N^z\{\tilde{h}_z^\MA(X_i; \tilde{\theta}_N^\MA)\} + \E_N\{\tilde{h}_z^\MA(X_i; \tilde{\theta}_N^\MA)\}, \\
    \check{\zeta}_{Nz} & 
    \equiv \bar{\tilde{Y}}_{z}({\theta}_N^\MA) 
    \equiv n_z^{-1} \sum_{i:Z_i=z} \tilde{Y}_i(z; {\theta}_N^\MA)
    = \bar{Y}_z - \E_N^z\{\tilde{h}_z^\MA(X_i; {\theta}_N^\MA)\} + \E_N\{\tilde{h}_z^\MA(X_i; {\theta}_N^\MA)\}. 
\end{align*}
From \eqref{eq:h_z_ENz_EN}, we then have, for $z=0,1$, 
\begin{align}\label{eq:zeta_hat_tilde_check}
    \hat{\zeta}_{Nz} = \check{\zeta}_{Nz} + o_{\Pr}(N^{-1/2}), 
    \qquad
    \tilde{\zeta}_{Nz} = \check{\zeta}_{Nz} + o_{\Pr}(N^{-1/2}), 
\end{align}

Sixth, under Condition \ref{cond:impute_adjust}, 
by the finite population central limit theorem \citep[][Theorem 5]{fpclt2017} and almost the same logic as the proof of Lemma \ref{lemma:maa_2}, 
we can know that, as $N\rightarrow \infty$, 
\begin{align}\label{eq:zeta_check_joint_clt}
    \sqrt{N}
    \begin{pmatrix}
    \check{\zeta}_{N1} - \bar{Y}(1)\\
    \check{\zeta}_{N0} - \bar{Y}(0)
    \end{pmatrix}
    \ \dot\sim \  
    \mathcal{N} \left(
    0, \ 
    \begin{pmatrix}
    (r_1^{-1} - 1) \Omega_{N11} & - \Omega_{N10}\\
    -\Omega_{N01} & (r_0^{-1}-1) \Omega_{N00}
    \end{pmatrix}
    \right), 
\end{align}
where $\Omega_{N11}$, $\Omega_{N10}$ and $\Omega_{N10} = \Omega_{N01}$ are the finite population variances of $\tilde{Y}_i(1; \theta_N^\MA)$ and $\tilde{Y}_i(0; \theta_N^\MA)$ and their finite population covariance. 

Seventh,
noting that $\hat{\tau}_{g\MA\MI} = g(\hat{\zeta}_{N1}) - g(\hat{\zeta}_{N0})$, 
we further define $\tilde{\tau}_{g\MA\MI}  = g(\tilde{\zeta}_{N1}) - g(\tilde{\zeta}_{N0})$ and $\check{\tau}_{g\MA\MI}  = g(\check{\zeta}_{N1}) - g(\check{\zeta}_{N0})$. 
By the mean value theorem, we can then derive that 
$\hat{\tau}_{g\MA\MI}  = \check{\tau}_{g\MA\MI}  + o_{\Pr}(N^{-1/2})$
and 
$\tilde{\tau}_{g\MA\MI} = \check{\tau}_{g\MA\MI}  + o_{\Pr}(N^{-1/2})$. 
Similar to the proof of Theorem \ref{thm:clt_mae}, 
we can know that $\sqrt{N}( \check{\tau}_{g\MA\MI} - \tau_g )$ is asymptotically Gaussian distributed with mean 0 and variance 
\begin{align*}
    V_{g\MA\MI} & = 
    r_1^{-1} \Var_N\{ \dot g(\bar{Y}(1))\tilde{Y}_i(1; \bs{\theta}_N^\MA) \} + 
	r_0^{-1} \Var_N\{ \dot g(\bar{Y}(0)) \tilde{Y}_i(0; \bs{\theta}_N^\MA) \} 
	\nonumber
	\\
	& \quad \ 
	- 
	\Var_N\{ \dot g(\bar{Y}(1)) \tilde{Y}_i(1; \bs{\theta}_N^\MA) - \dot g(\bar{Y}(0)) \tilde{Y}_i(0; \bs{\theta}_N^\MA)  \}, 
\end{align*}
where $Y_i(z; \bs{\theta}_N) = Y_i(z) - \tilde{h}_z^\MA(X_i; {\theta}_N^\MA) + \E_N\{\tilde{h}_z^\MA(X_i; {\theta}_N^\MA)\}$ for $z=0,1$. 
By Slutsky's theorem, 
$\sqrt{N} ( \hat{\tau}_{g\MA\MI} - \tau_g )$ 
and $\sqrt{N}( \tilde{\tau}_{g\MA\MI} - \tau_g )$ follow the same asymptotic distribution as $\sqrt{N}( \check{\tau}_{g\MA\MI} - \tau_g )$. 

Finally, we consider variance estimation. 
From \eqref{eq:zeta_hat_tilde_check} and \eqref{eq:zeta_check_joint_clt}, $\hat{\zeta}_{Nz} \equiv \bar{Y}_{z}(\hat{\theta}_N^\MA) = \bar{Y}(z) + o_{\Pr}(1)$ for $z=0,1$. 
From \eqref{eq:AI_samp_var_est} and \eqref{eq:AI_samp_var_true} and by some algera, we can derive that 
\begin{align*}
    \Var_N^z\{Y_i(z; \hat{\theta}^\MA_N) \}
    & = 
    \Var_N^z\Big\{ Y_i - \sum_{j=1}^J \hat{\beta}_{Nzj} h_{0j}^\MI(X_i; \hat{\theta}_{Nj}^\MI) - \sum_{j=1}^J \hat{\gamma}_{Nzj} h_{1j}^\MI (X_i; \hat{\theta}^\MI_{Nj})   \Big\}\\
    & = 
    \Var_N\Big\{ Y_i(z) - \sum_{j=1}^J {\beta}_{Nzj} h_{0j}^\MI(X_i; {\theta}_{Nj}^\MI) - \sum_{j=1}^J {\gamma}_{Nzj} h_{1j}^\MI (X_i; {\theta}^\MI_{Nj})   \Big\} + o_{\Pr}(1)
    \\
    & = \Var_N\{\tilde{Y}_i(z; {\theta}^\MA_N) \} + o_{\Pr}(1). 
\end{align*}
Let 
$
\tilde{V}_{g\MA\MI} \equiv r_1^{-1} \Var_N\{ \dot g(\bar{Y}(1))\tilde{Y}_i(1; \bs{\theta}_N^\MA) \} + 
	r_0^{-1} \Var_N\{ \dot g(\bar{Y}(0)) \tilde{Y}_i(0; \bs{\theta}_N^\MA) \} \ge V_{g\MA\MI}. 
$
We then have 
\begin{align*}
    \hat{V}_{g\MA\MI} & = 
    r_1^{-1} \Var_N^1\{ \dot g( \bar{Y}_{1}(\hat{\theta}_N^\MA) )Y_i(1; \hat{\theta}_N^\MA) \} + 
	r_0^{-1} \Var_N^0\{ \dot g( \bar{Y}_{0}(\hat{\theta}_N^\MA) ) {Y}_i(0; \hat{\theta}_N^\MA) \} 
	\nonumber
	\\
	& = \tilde{V}_{g\MA\MI} + o_{\Pr}(1)
	\ge V_{g\MA\MI} + o_{\Pr}(1). 
\end{align*}

From the above, Theorem \ref{thm:clt_ai} holds. 
\end{proof}

\section{Additional remarks and technical details}\label{sec:add_detail}

\subsection{Additional remark for Condition \ref{cond:consis_MI}}\label{sec:glm_cond}

Below we show that the generalized linear models in \eqref{eq:glm} satisfy Condition \ref{cond:consis_MI}. 
By the density form in \eqref{eq:glm}, we can derive that, for $z=0,1$, 
$- \partial \log f_z(y,x;\theta) / \partial \alpha_{1-z} = 0$, and 
\begin{align*}
    - \frac{\partial}{\partial \alpha_z} \log f_z(y,x;\theta)
    = 
    - \frac{y - \dot{b}_z(\alpha_z+\beta_z^\top x)}{a_z(\phi_z)} 
    = 
    - \frac{y - h_z(x;\theta)}{a_z(\phi_z)}. 
\end{align*}
Without loss of generality, we assume $\alpha_0$ and $\alpha_1$ are the first two coordinates of $\theta$. 
Let $q_0(\theta)$ and $q_1(\theta)$ be two vectors that have the same dimension as $\theta$, with $q_0(\theta) = (-a_0(\phi_0), 0, 0, \ldots, 0)^\top$ and $q_1(\theta) = (0, -a_1(\phi_1), 0, \ldots, 0)^\top$, 
and recall that $\psi_z(y,x;\theta) = - \partial \log f_z(y,x;\theta) / \partial \theta$ for $z=0,1$. 
We can then verify that 
$q_z(\theta)^\top \psi_z(y,x;\theta) = y - h_z(x;\theta)$ 
and 
$q_z(\theta)^\top \psi_{1-z}(y,x;\theta) = 0$ for $z=0,1$. 
Therefore,  the generalized linear models 
in \eqref{eq:glm} 
satisfy Condition \ref{cond:consis_MI}. 

\subsection{Additional remark for the relation between model-imputed and model-assisted estimator}\label{sec:proof_equ_mi_ma}

Below we show how to view a model-assisted estimator as a model-imputed estimator, giving the details for the discussion in \S \ref{sec:conn_mi_ma}. 
Given the estimating equation in \eqref{eq:psi_ma_mi} for $\alpha_z$s, we can know that, for $z=0,1$.
\begin{align}\label{eq:alpha_hat_ma_mi}
    0 = \E_N^z\{Y_i - \hat{\alpha}_z - h_z(x;\hat{\theta})\}.
\end{align}
where $\hat{\theta}$ is the same parameter estimator used for the model-assisted estimator. 
Let $\hat{\tilde{\theta}} = (\hat{\alpha}_1, \hat{\alpha}_0, \hat{\theta}^\top)^\top$ denote the parameter estimator for the model-imputed estimator. We then have, for $z=0,1$,  
\begin{align*}
    \E_N\{\tilde{h}_z(X_i, \hat{\tilde{\theta}})\}
    & = 
    \E_N\{\hat{\alpha}_z  + h_z(X_i, \hat{\theta})\}
    = \hat{\alpha}_z + \E_N\{ h_z(X_i, \hat{\theta})\}
    \\
    & = \E_N^z\{Y_i - h_z(x;\hat{\theta})\} + \E_N\{ h_z(X_i, \hat{\theta})\}\\
    & = \E_N^z\{Y_i(\hat{\theta})\},
\end{align*}
where the second last equality follows from \eqref{eq:alpha_hat_ma_mi}, and the last equality follows from the definition in \S \ref{sec:MA_consist}. 
This then implies that 
\begin{align*}
    \hat{\tau}_{g\MI} 
    & = g(\E_N\{\tilde{h}_1(X_i, \hat{\tilde{\theta}})\})
    - g(\E_N\{\tilde{h}_0(X_i, \hat{\tilde{\theta}})\})
    = g(\E_N^1\{Y_i(\hat{\theta})\})
    - g(\E_N^0\{Y_i(\hat{\theta})\})
    = \hat{\tau}_{g\MA},
\end{align*}
showing that we can also view a model-assisted estimator as a model-imputed estimator. 

\subsection{Additional remark for general estimands of form $\nu(\bar{Y}(1), \bar{Y}(0))$}
From the proof of Theorem \ref{thm:clt_mae}, we can know that 
\begin{align}\label{eq:Y_bar_vec}
    \begin{pmatrix}
        \bar{Y}_1(\hat{\theta}_{N}) - \bar{Y}(1)\\
        \bar{Y}_0(\hat{\theta}_{N}) - \bar{Y}(0)
    \end{pmatrix}
    \lpc 
    \mathcal{N}
    \left( 
    \begin{pmatrix}
        0\\
        0
    \end{pmatrix}, 
    \begin{pmatrix}
        (r_1^{-1}-1) \Var_N\{Y_i(1;\theta_N)\}
        & -\Cov_N\{Y_i(1;\theta_N), Y_i(0;\theta_N)\}
        \\
        -\Cov_N\{Y_i(1;\theta_N), Y_i(0;\theta_N)\}
        & 
        (r_0^{-1}-1) \Var_N\{Y_i(0;\theta_N)\}
    \end{pmatrix}
    \right).
\end{align}
From the proof of Theorem \ref{thm:MA_ci}, 
we can know that 
\begin{align}\label{eq:Y_bar_vec_var_est}
    & \quad \ \begin{pmatrix}
        r_1^{-1} \Var_{N_t}^1\{Y_i(\hat{\theta}_{N})\} & 0\\
        0 & r_0^{-1} \Var_{N}^0\{Y_i(\hat{\theta}_{N})\}
    \end{pmatrix}
    \\
    & 
    \lpc 
    \begin{pmatrix}
        r_1^{-1} \Var_{N}\{Y_i(1;\theta_{N})\}  & 0\\
        0 & r_0^{-1} \Var_{N}\{Y_i(\hat{\theta}_{N})\}
    \end{pmatrix}
    \nonumber
    \\
    & \ge 
    \begin{pmatrix}
        r_1^{-1} \Var_{N}\{Y_i(1;\theta_{N})\}  & 0\\
        0 & r_0^{-1} \Var_{N}\{Y_i(\hat{\theta}_{N})\}
    \end{pmatrix}
    - 
    \Cov_N 
    \left( 
    \begin{pmatrix}
        Y_i(1;\theta_{N})\\
        Y_i(0;\theta_{N})
    \end{pmatrix}
    \right)
    \nonumber
    \\
    & = \begin{pmatrix}
        (r_1^{-1}-1) \Var_N\{Y_i(1;\theta_N)\}
        & -\Cov_N\{Y_i(1;\theta_N), Y_i(0;\theta_N)\}
        \\
        -\Cov_N\{Y_i(1;\theta_N), Y_i(0;\theta_N)\}
        & 
        (r_0^{-1}-1) \Var_N\{Y_i(0;\theta_N)\}
    \end{pmatrix}.
    \nonumber
\end{align}
Consider any estimand of form $\nu(\bar{Y}(1), \bar{Y}(0))$. 
Assuming that $\nu$ is continuously differentiable, 
from \eqref{eq:Y_bar_vec} and 
by the Delta method as in Theorem \ref{thm:clt_mae}, we can know that  $\nu(\bar{Y}_1(\hat{\theta}_{N}), \bar{Y}_0(\hat{\theta}_{N}))$ is consistent for $\nu(\bar{Y}(1), \bar{Y}(0))$ and asymptotically Normal.  
From \eqref{eq:Y_bar_vec_var_est}, we can further derive its conservative variance estimator.

{\lad 

\subsection{Working model for binary potential outcomes}

For binary outcomes, 
individual treatment effects 
can only take values in $\{-1,0,1\}$. We can consider the following working model 
\begin{equation}\label{eq:tau_model_binary}
    \tau \mid  X = 
        \begin{cases} 
        1 & \text{w.p. } \ \exp(x^\top \beta) / \{ \exp(x^\top \beta) + \exp(-x^\top \beta) + \gamma \}, \\ 
        0 & \text{w.p. } \ \gamma / \{ \exp(x^\top \beta) + \exp(-x^\top \beta) + \gamma \}, \\ 
        -1 & \text{w.p. } \ \exp(-x^\top \beta) / \{ \exp(x^\top \beta) + \exp(-x^\top \beta) + \gamma \},
        \end{cases}
\end{equation}
where $\gamma$ is a prespecified known constant. 
The estimation approach in \citet{Tian2014} is equivalent to the maximum likelihood estimation from the working model \eqref{eq:tau_model_binary} with $\gamma=2$. 
If we choose $\gamma=1$, then the model \eqref{eq:tau_model_binary} reduces to a multinomial logistic regression model. 
Intuitively, the model assumes that the the individual effect is more likely to be positive when $x^\top \beta$ is large, and more likely to be negative when $x^\top \beta$ is small. In other words, the individual effect is positively associated with $x^\top \beta$.

Under the working model \eqref{eq:tau_model_binary}, 
the probability mass function of $\tau \mid X$ has the following form:
\begin{equation*}
    f(\tau \mid x; \beta) = \exp(\tau \cdot x^\top \beta) \cdot \gamma ^ {\mathbb{I}\{\tau = 0\}} / \big[\exp(x^\top \beta) + \exp(-x^\top \beta) + \gamma \big],
\end{equation*}
which falls in the exponential dispersion family in \eqref{eq:model_tau}. 
Thus, its parameter is identifiable and can be inferred from the observed data using our Z-estimation theory. 

If, in addition, we know that the treatment effect is nonnegative for any unit, then the individual treatment effects can only take values in $\{0,1\}$. In this case, we can consider imposing a logistic regression working model on $\tau\mid X$. It also falls in the exponential dispersion family in \eqref{eq:model_tau}, and its parameter is identifiable and can be inferred from the observed data. 
}

\end{document}